\documentclass[11pt]{article}
\usepackage{amsmath, amssymb, amsthm}
\usepackage{graphicx,psfrag}
\usepackage{mathtools}
\usepackage{bm}
\usepackage{bbm}
\usepackage{enumerate}
\usepackage{subcaption}
\usepackage{natbib}
\usepackage[mathscr]{euscript}
\usepackage{url} 
\usepackage[bottom]{footmisc}
\usepackage{booktabs}
\usepackage{multirow}
\usepackage{color}

\usepackage[hypertexnames=false]{hyperref}
\hypersetup{
colorlinks = true,
allcolors = blue
}
\usepackage{caption}
\usepackage{adjustbox}
\usepackage[normalem]{ulem}
\usepackage[a4paper, 
            left=1.1in,right=1.1in,top=1in,bottom=1.2in,%
            footskip=.7in]{geometry}

\usepackage[toc,page,header]{appendix}
\usepackage{minitoc}   
\usepackage{xcolor}
\usepackage[linesnumbered,ruled,vlined]{algorithm2e}

\usepackage{cleveref}
\usepackage{wrapfig}

\makeatletter
\newcommand{\algorithmfootnote}[2][\footnotesize]{%
  \let\old@algocf@finish\@algocf@finish% Store algorithm finish macro
  \def\@algocf@finish{\old@algocf@finish% Update finish macro to insert "footnote"
    \leavevmode\rlap{\begin{minipage}{\linewidth}
    #1#2
    \end{minipage}}%
  }%
}
\makeatother

\newcommand{\bbR}{\mathbb{R}} % real numbers
\newcommand{\bbN}{\mathbb{N}} % natural numbers
\newcommand{\bbP}{\mathbb{P}} % probability measure
\newcommand{\bbS}{\mathbb{S}} % permutations

\def\ind{\mathbbm{1}} % indicator function

\def\CI{\mathrm{CI}}

\DeclareMathOperator\diag{diag}

\newcommand{\cG}{\mathcal{G}}
\newcommand{\cE}{\mathcal{E}}
\newcommand{\cL}{\mathcal{L}}

\def\Pa{\mathrm{Pa}} % parent set
\def\sX{\mathsf{X}}

\newcommand{\se}{\mathsf{e}}

\newcommand{\cA}{\mathcal{A}}
\newcommand{\cB}{\mathcal{B}}
\newcommand{\cC}{\mathcal{C}}

\newcommand{\cS}{\mathcal{S}}

\newcommand{\cM}{\mathcal{M}}

\newcommand{\cN}{\mathcal{N}} % neighborhood

\newcommand{\proj}{\Phi}

\SetKwInput{KwInput}{Input}
\SetKwInput{KwOutput}{Output}

\theoremstyle{remark}
\newtheorem{remark}{Remark}

\newtheorem{proposition}{Proposition}
\newtheorem{theorem}{Theorem}
\newtheorem{corollary}{Corollary}
\newtheorem{lemma}{Lemma}
\newtheorem{definition}{Definition}
\newtheorem{example}{Example}

\title{Leveraging heterogeneity for identifiability: Bayesian order-based learning of multiple DAGs} 
\usepackage{authblk}
\author{Hyunwoong Chang\thanks{hwchang@utdallas.edu}\,\, and Fariha Taskin}
\date{}
\affil{Mathematical Sciences, University of Texas at Dallas }

\begin{document}
\maketitle

\begin{abstract}
We propose a joint order-based scoring framework for causal structure learning of directed acyclic graph (DAG) models under heterogeneous data settings. We show that leveraging heterogeneity improves the accuracy of causal ordering estimation. In the most favorable case, the causal ordering is identifiable up to two permutations. 
Building on this framework, we propose an order-based Bayesian method for Gaussian DAG models and establish its theoretical properties in the high-dimensional regime. 
For posterior inference over the space of orderings, we introduce a random-to-random (R2R) proposal neighborhood for the Metropolis–Hastings algorithm, which is theoretically motivated and exhibits efficient mixing behavior. 
Simulation studies confirm the strong empirical performance of the proposed method, and an application to single-nucleus RNA sequencing data from major depressive disorder demonstrates practical utility.
\end{abstract}
\small

\section{Introduction}\label{sec-intro}

We consider the joint estimation of multiple directed acyclic graph (DAG) models from observational data collected from heterogeneous sources.
Specifically, the data consist of $K$ datasets, where each dataset $X^{(k)}$ is generated by a structural causal model associated with a DAG $G^{(k)}$, all defined on the same set of variables.
For example, gene expression measurements across individuals or cell types are known to exhibit substantial heterogeneity, so that treating observations as independently and identically distributed can result in biased inference.
Most existing structure learning methods, however, are designed to estimate a single DAG, implicitly assuming that observational data are generated from a single homogeneous system. While such methods can be applied separately to each data source, this approach fails to leverage potential structural similarities across sources.

From purely observational data, the underlying causal DAG is generally not identifiable without additional assumptions~\citep{spirtes2000causation}.
This is due to the existence of Markov equivalence classes, in which multiple distinct DAGs may encode the same set of conditional independence relations. 
Many existing methods relax the goal to identifying the Markov equivalence class rather than the exact causal DAG.
However, ambiguity in edge directions can lead to fundamentally different causal interpretations~\citep{heinze2018causal}. 
To overcome this limitation, the literature has primarily pursued two approaches for identifiability. 
One approach is to impose additional distributional or functional assumptions that guarantee causal identifiability, such as equal error variances, non-linear structural equations with additive noise, and additive models~\citep{Drton2017-oa}.
Although it is practical when randomized controlled experiments are infeasible, these conclusions can be spurious if the assumptions fail, and such assumptions are generally not testable. 
The other approach assumes access to experimental data, where certain variables are actively perturbed under various intervention schemes~\citep{korb2004varieties}, providing information to resolve ambiguity of edge directions~\citep{eberhardt2007causation}. 
They have studied model properties under several scenarios including static data cases and active learning~\citep{He2008-zs, Hauser2011-jg}. 
However, interventional data are often costly or infeasible to obtain in many applications.

We propose a new perspective on causal identifiability: heterogeneity in observational data can enhance identifiability. Under the assumption that datasets share a common causal variable ordering, our framework adopts an order-based approach that searches over the space of orderings rather than directly over the space of DAGs. 
Order-based methods have been popular in single DAG estimation~\citep{friedman2003being, agrawal2018minimal}, since  the acyclicity constraint guarantees that every DAG admits at least one consistent ordering, and that the main challenge lies in uncertainty over the ordering of variables. 
To study identifiability in this setting, we extend the notion of score equivalence to the order-based framework, which we call score-equivalent order-based scores. 
Here, score equivalence means that DAGs in the same Markov equivalence class receive the same score, a property widely regarded as natural in score-based structure learning \citep{andersson1997characterization, chickering2002learning}. 
Importantly, existing order-based scoring methods for single DAG estimation satisfy this property~\citep{van2013ell, raskutti2018learning, solus2017consistency}. 
We define the joint order-based score by aggregating score-equivalent order-based scores computed from individual datasets.
Under this joint score, we show that heterogeneity across datasets penalizes inconsistent orderings, thereby improving identifiability.
The underlying mechanism is that incorporating additional graphs expands the union of essential arrows. Here, essential arrows refer to edges whose orientations are invariant within a Markov equivalence class \citep{andersson1997characterization}. By aggregating individual scores, the joint order-based score assigns lower values to orderings inconsistent with at least one essential arrow across datasets, thereby discouraging such orderings. 
In the best-case scenario, 
we can identify two candidate orderings, one of which is guaranteed to be the true causal ordering. We characterize when this guarantee holds.
To our knowledge, this perspective has not been explicitly studied in prior work.
Interestingly, this result contrasts with existing approaches, which typically encourage shared structures across different data sources or perform well under structural similarity assumptions \citep{Wang2020-of,lee2022bayesian}.

To leverage structural similarities across data sources, a large body of existing work assumes that all datasets share a common causal ordering, with a few exceptions \citep{oates2016exact}.
Although the underlying graph structure may vary across datasets, the direction of causal influence is often preserved. 
This assumption is particularly plausible in settings where interventions or environmental changes are unlikely to reverse causal directionality, such as time-ordered or hierarchical systems (Remark~\ref{remark:common}). 

Relative to multiple undirected graph estimation, multiple DAG estimation has received comparatively less attention, perhaps due to its inherent structural and computational complexity. Depending on the formulation, the literature on multiple DAG learning includes joint skeleton estimation \citep{liu2019joint}, essential graph estimation \citep{castelletti2020bayesian}, and multiple DAG estimation under a given ordering \citep{lee2022bayesian}. 
\citet{Wang2020-of} establish high-dimensional consistency for a two-step procedure that first learns a union graph via GES~\citep{chickering2002optimal} and then estimates individual graph edges.
\citet{Li2024-af} adopt an order-based MCMC framework, learning a shared ordering via multi-task variable selection along the Markov chain, though their computational strategy is different from ours. 

We develop a Bayesian model that inherits the theoretical properties of the proposed framework. 
Under a Gaussian DAG model, we employ an empirical prior and construct the marginal posterior distribution over the ordering space to enable uncertainty quantification for causal orderings. 
We show that our Bayesian formulation retains the theoretical properties of the proposed framework in the high-dimensional regime. Specifically, the Bayes factor between admissible and inadmissible orderings diverges to infinity, showing strong asymptotic separation.
To approximate the posterior distribution, we employ a Metropolis-Hastings (MH) algorithm over the joint space of causal orderings and individual DAGs. 
This computation is challenging due to the factorial growth of the ordering space in the number of variables. Although various MH schemes have been developed, mixing behavior depends critically on the choice of the proposal neighborhood~\citep{chang2026dimension}. 
To ensure effective mixing, the proposal neighborhood must be large enough to avoid trapping in local modes, yet not so large as to significantly reduce sampling efficiency. 

We adopt a random-to-random (R2R) proposal neighborhood. 
We prove that it contains local moves that emulate the Chickering sequence \citep{chickering2002optimal}. Consequently, there exists a sequence of moves along which the individual score increases monotonically.
In simulations, we evaluate the performance of the proposed algorithm across varying numbers of data sources and demonstrate its performance as the level of heterogeneity changes (Section~\ref{subsec:performance}). 
We further show that the Markov chain induced by the R2R proposal neighborhood exhibits strong mixing behavior (Section~\ref{subsec:conv}). In comparison with competing approaches, our method achieves superior performance (Section~\ref{subsec:comparison}).
We also conduct a sensitivity analysis of the proposed method under violations of the common ordering assumption (Section~\ref{subsec:common_ordering}).
Finally, we analyze a major depressive disorder case–control dataset generated by publicly available single-nucleus RNA-seq profiling of post-mortem brain samples~\citep{nagy2020single}.

\section{Leveraging heterogeneous data sources under order-based scoring framework}\label{sec-main}

\subsection{Preliminary}\label{subsec-prelim}

Let $G = (V, E)$ denote a DAG, where $V$ is a node set and $E \subset V \times V$ is a set of directed edges that form no cycle. For a $p$-node DAG $G$, we label $V = [p] := \{1, \dots, p\}$, and write $(i,j) \in E$ to denote a directed edge from $i$ to $j$. For convenience, we use the shorthand $i \to j \in G$ in place of $(i,j) \in E$. We use $|\cdot|$ to denote the cardinality of a set, and let $|G| = |E|$ denote the number of edges in $G$. For a DAG $G$ and a node $j$, we denote by $\Pa_j(G)$ the parent set of $j$. 
Let $\sX = (\sX_1, \dots, \sX_p)$ be a random vector, and denote by $\CI(\sX)$ the set of conditional independence relations that hold in the distribution of $\sX$. 
Similarly, let $\CI(G)$ denote the set of conditional independence relations that a DAG $G$ entails by d-separation rule. 
A structural causal model (SCM)~\citep{pearl2009causality}  is a framework in which a random vector $\sX$ is causally associated with a DAG $G$. Specifically, each variable is generated according to a  structural equation 
\begin{align}\label{eq:scm}
    \sX_j = f_j(\sX_{\Pa_j(G)}, \varepsilon_j), \text{ for  } j \in [p],
\end{align}
where $f_j$ is a deterministic function and $(\varepsilon_1, \dots, \varepsilon_p)$ are mutually independent noise variables. If $\sX$ follows an SCM with a DAG $G$, the distribution of $\sX$ is Markov with respect to $G$; that is, $\CI(G) \subseteq \CI(\sX)$. Additionally, we say the distribution of $\sX$ is faithful to $G$ if $\CI(G) = \CI(\sX)$. 
Two DAGs $G$ and $G'$ are called Markov equivalent if $\CI(G) = \CI(G')$. The Markov equivalence class of $G$, denoted by $\cE(G) = \{G' : \CI(G') = \CI(G)\}$, is the set of  DAGs Markov equivalent to $G$. An edge  $i \rightarrow j$ in $G$ is called an essential arrow if it appears in every DAG $G' \in \cE(G)$. 
The set of essential arrows of $G$ is denoted by $E_G^*$.

A key object in our approach is a (topological) ordering of a DAG. 
A bijection $\sigma: [p] \rightarrow [p]$ is called an ordering for a DAG $G$ if the following holds: for every pair of indices $i < j$, whenever an edge exists between the nodes $\sigma(i)$ and $\sigma(j)$ in $G$, its direction must be $\sigma(i) \rightarrow \sigma(j)$, not  $\sigma(j) \rightarrow \sigma(i)$. Every DAG admits at least one ordering, which need not be unique.
%From the perspective of order theory, an ordering of $G$ is a linear extension of the partial order induced by the edge set $E$~\citep{davey2002introduction}.
Let $\bbS^p$ denote the set of all permutations of $[p]$, that is, all possible orderings. Let $\cL(G) \subseteq \bbS^p$ denote the set of orderings consistent with a DAG $G$. Further, define $\cL^e(E)\subseteq \bbS^p$ to be the set of linear extensions of the partial order induced by $E$. For a collection of DAGs $\cG$, let $\cL^\cap(\cG) = \bigcap_{G \in \cG} \cL(G)$ and $\cL^\cup(\cG) = \bigcup_{G \in \cG} \cL(G)$. 
For each node $j \in [p]$, let 
\begin{align}\label{eq:potential}
    P_j^{\sigma} = \{i\colon  \sigma^{-1}(i) < \sigma^{-1}(j)\} 
\end{align}
denote the set of predecessors of $j$ under the ordering $\sigma$, that is, the set of nodes that precede $j$ in $\sigma$. 
Here, $\sigma^{-1}$ is the inverse permutation of $\sigma$. Let $\cG_p$ denote the collection of all $p$-node DAGs, and let $\cG_p^\sigma$ be the collection of all $p$-node DAGs that are consistent with the ordering $\sigma$, that is,
$  \cG_p^\sigma = \{ G  \in \cG_p \colon$ for every edge $ i \rightarrow j$ in $G$, $\sigma^{-1}(i) < \sigma^{-1}(j) \}$.  
A DAG $G$ is an independence map (I-map) of $\sX$ if $\CI(G) \subseteq \CI(\sX)$. A DAG $G$ is a minimal I-map of $\sX$ if it is an I-map and no proper subgraph of $G$ is an I-map of $\sX$. For a fixed ordering $\sigma \in \bbS^p$, there exists a unique minimal I-map of $\sX$ consistent with $\sigma$, which we denote by $G_\sigma$. Let $\mathrm{MVN}_k$ denote the $k$-dimensional multivariate normal distribution.
%Finally, we use $\mathrm{MVN}_k$ to denote the $k$-dimensional multivariate normal distribution.
In this section, to convey the main ideas clearly,  we assume the faithfulness condition, under which the conditional independence relations in the distribution of $\sX$ coincide with those implied by the true DAG $G^*$, that is, $\CI(G^*) = \CI(\sX)$. We discuss the condition in more detail in Remark~\ref{remark:faithfulness}.
%\tcr{We defer to the supplement () for proof and main concept}

\subsection{Order-based scoring framework under data heterogeneity}\label{subsec:order-based}

\subsubsection{A general framework for order-based scores and score equivalence}\label{subsub:score_eq}

In score-based structure learning, each candidate DAG $G$ is assigned a score $\phi(G)$, such as the Bayesian information criterion (BIC), and the objective is to find a DAG that maximizes the score~\citep{Drton2017-oa}.
We consider an order-based approach to structure learning, which constructs an order-based score $\psi: \bbS^p \to \mathbb{R}$ and 
identifies an optimal ordering in $\bbS^p$.  Using an order-based approach offers several advantages; it operates on a substantially smaller search space than DAG-based methods, and has been shown to exhibit improved empirical performance and more stable convergence~\citep{friedman2003being}. An order-based score $\psi$ is typically defined from a DAG-based score $\phi$ as
\begin{equation}\label{eq:order-based_score}
    \psi(\sigma) = F\big(\{ \phi(G) : G \in \mathcal{G}_p^\sigma \}\big),
\end{equation}
where $F$ is an aggregation function. 
We present three examples of order-based scores below, assuming that a random vector $\sX = (\sX_1, \dots, \sX_p)$ follows an SCM with true DAG $G^*$.

\begin{example}[Sparsest permutation]\label{ex:sp}
In sparsest permutation methods~\citep{raskutti2018learning, solus2017consistency}, the order-based score is defined as
\begin{align*}
\psi_1(\sigma) := \max_{} \{ -|G|: \, G \in \cG_p^{\sigma}\; , \CI(G) \subseteq \CI (\sX)\}.
\end{align*}
Since the minimal I-map $G_\sigma$ uniquely minimizes the number of edges among such DAGs, we have $\psi_1(\sigma) = -|G_\sigma|.$
\end{example}

\begin{example}[$\ell_0$-penalized maximum likelihood]\label{ex:mle}
Assume that $\sX$ follows a multivariate normal distribution. Given an ordering $\sigma \in \bbS^p$ and $j \in [p]$, let $S \subseteq  P^\sigma_j$ and consider the linear projection $\sX_{\sigma(j)} = \beta_{\sigma(j)\mid S}^\top \sX_S + \varepsilon_{\sigma(j)\mid S}$  with conditional variance $\omega_{\sigma(j)\mid S}
\;:=\;
\mathrm{Var}\!\left(\sX_{\sigma(j)} \mid \sX_S\right).$
Following~\cite{van2013ell}, we consider the expected $\ell_0$-penalized Gaussian log-likelihood over DAGs consistent with $\sigma$, and define the order-based score as
\begin{align*}
    \psi_2(\sigma)
\;:=\; \max_{G\in\mathcal{G}_p^{\sigma} } \left(-\frac{1}{2}\sum_{j=1}^p
\log \omega_{\sigma(j)\mid \Pa_{\sigma(j)}(G)}
 - \lambda |G|\right),
\end{align*}
where $\lambda > 0$ is a regularization parameter.
\end{example}

\begin{example}[Bayesian order-based learning]\label{ex:bayes}

In Bayesian order-based structure learning, the score is proportional to the posterior probability of an ordering $\sigma \in \bbS^p$,
\begin{align*}
    \psi_3(\sigma) \;\propto\;
\pi_0(\sigma)
\sum_{G \in \mathcal{H}_p^\sigma}
\exp\bigl(\phi(G)\bigr),
\end{align*}
where $\phi(G)$ denotes the log posterior score of $G$, $\mathcal{H}^\sigma_p \subseteq \mathcal{G}_p^\sigma$ is a restricted collection of DAGs consistent with $\sigma$, and $\pi_0$ is a prior distribution over $\bbS^p$. The summation over $\mathcal{H}^\sigma_p $ corresponds to Bayesian model averaging over DAGs consistent with the ordering $\sigma$. In practice, $\mathcal{H}^\sigma_p$ is restricted to DAGs with bounded in-degree to ensure computational tractability~\citep{friedman2003being}.
Specifically, for a given $d \in \bbN$, define
\begin{align}\label{eq:restricted}
    \cG_p^\sigma (d) = \{G \in \mathcal{G}_p^\sigma: |\Pa_j(G)| \leq d   \text{ for all } j\in [p]\}.
\end{align}
 In certain implementations, $\mathcal{H}^\sigma_p$ may contain only a single representative DAG, such as the minimal I-map~\citep{agrawal2018minimal} or the MAP estimate~\citep{Chang2022-id}.
\end{example}
Since an order-based score $\psi$ is induced by a score $\phi$, it is natural to require that $\psi$ inherit properties of $\phi$ that are desirable for structure learning. One such property is score equivalence~\citep{andersson1997characterization, chickering2002learning}, which requires that $\phi(G)=\phi(G')$ for any pair of Markov equivalent DAGs $G$ and $G'$. Score equivalence is motivated by the fact that observational data cannot distinguish between DAGs in the same Markov equivalence class, thus a reasonable DAG-based score should not systematically prefer one DAG over another in the class. 
Analogously, a desirable property of an order-based score $\psi$ is that all orderings consistent with any DAG in $\cE(G^*)$ receive the same score, whereas orderings that are inconsistent with every DAG in $\cE(G^*)$ receive a strictly lower score. 
%To formalize this notion, we introduce score equivalence for order-based scores.
\begin{definition}[Score equivalence for order-based scores]\label{def:score-eq}
    Let $G^*$ be the true DAG, and let $\sigma^* \in \cL(G^*)$ be any ordering consistent with $G^*$.
    An order-based score $\psi$ is called score equivalent if $\psi(\sigma) = \psi(\sigma^*)$ whenever  $\cG_p^\sigma \cap \cE(G^*) \neq \emptyset$, and $\psi(\sigma) < \psi(\sigma^*)$ otherwise.
\end{definition}
Under the faithfulness assumption, the order-based score $\psi_1$ in Example~\ref{ex:sp} is score equivalent. Indeed, under faithfulness, we have $|G_\sigma| \geq |G^*|$ for all $\sigma \in \bbS^p$. 
Moreover, the Markov equivalence class $\cE(G^*)$ coincides with the set $\{ G_\sigma : \sigma \in \bbS^p,\ |G_\sigma| = |G^*| \},$ see~\citet{van2013ell}. Similarly, the score $\psi_2$ in Example~\ref{ex:mle} is also score equivalent. In particular,  Markov equivalent DAGs yield identical log-likelihood values under Gaussian models, and since they also have the same number of edges, the penalized likelihood is identical for all DAGs in $\cE(G^*)$.

When causal interpretation is of interest, the goal of order-based structure learning is to identify an ordering $\sigma^*$ consistent with the true DAG $G^*$, that is, $\sigma^* \in \cL(G^*)$. Once $\sigma^*$ is identified, estimation of the DAG reduces to a series of variable selection problems for identifying the parent set of each node, and several methods can consistently identify the true DAG~\citep{lee2019minimax}.
Suppose $\psi$ is a score equivalent order-based score. It is possible that an ordering $\sigma \notin \cL(G^*)$ still satisfies $\psi(\sigma) = \psi(\sigma^*)$. This occurs when $\sigma$ is inconsistent with $G^*$ itself but is consistent with some other DAG in $\cE(G^*)$. We refer to such an ordering $\sigma$ as a misspecified ordering, and  define the set of misspecified orderings for a DAG $G^*$ as 
\begin{align*}
    \cM^\psi(G^*) = \{\sigma: \psi(\sigma) = \psi(\sigma^*)\} \setminus \cL(G^*)=\cL^\cup(\cE(G^*)) \setminus \cL(G^*),
\end{align*}
where the second equality follows from score equivalence. Intuitively, $\mathcal{M}^\psi(G^*)$ represents the set of orderings that correspond to DAGs that are Markov equivalent to, but structurally distinct from $G^*$. Thus, the size of $\mathcal{M}^\psi(G^*)$ quantifies the degree of ambiguity in identifying orderings consistent with $G^*$. 
In particular,  $\mathcal{M}^\psi(G^*) = \emptyset$ implies that $\cE(G^*) = \{G^*\}$, so that the order-based score uniquely identifies the causal orderings consistent with $G^*$.
In contrast, a larger $\mathcal{M}^\psi(G^*)$ indicates greater ambiguity in identifying edge orientations in the underlying DAG. 

\subsubsection{Enhancing identifiability via joint order-based scores}\label{subsub:joint_score}

We now consider the scenario where multiple data distributions $\sX^{(1)}, \dots, \sX^{(K)}$
are available, and each $\sX^{(k)}$  is generated from an SCM with DAG $G^{(k)}$ for $k \in [K]$. 
Let $\cG = \{G^{(1)}, \dots, G^{(K)}\}$ denote the collection of DAGs, which may vary across $k \in [K]$. 
We assume that there exists an ordering $\sigma^*$ that is consistent with every $G \in \cG$. 
This assumption is standard in the literature~\citep{Wang2020-of, Li2024-af, lee2022bayesian} and reflects the presence of a common causal ordering across environments, while allowing the underlying graph structures to vary. Further discussion is given in Remark~\ref{remark:common}.
We show that pooling information across multiple datasets improves the identifiability of the shared causal ordering. 
To this end, we define the joint order-based score as follows.
\begin{definition}[Joint order-based score]\label{def:joint_score}
Let $\psi^{(k)}$ be a score-equivalent order-based score constructed from $\sX^{(k)}$ for each $k \in [K]$. Let $F: \bbR^K \to \bbR$ be strictly increasing in each coordinate, and define the joint order-based score  
\begin{equation}\label{eq:joint_score}
        \Psi(\sigma) = F\bigl(\psi^{(1)}(\sigma), \dots, \psi^{(K)}(\sigma)\bigr),
         \qquad \sigma \in \bbS^p.
    \end{equation}
\end{definition}
Note that possible choices of $F$ include weighted sums with positive weights, products, weighted geometric means, and log-sum-exp functions. The following lemma shows that the set of misspecified orderings under the joint score is a subset of the corresponding set under the individual score. 

\begin{lemma}\label{lemma:joint_score}
    
    Assume the common ordering assumption holds. Let $\Psi$ be the joint order-based score defined in Definition~\ref{def:joint_score}. Then
    \begin{itemize}
        \item [(1)] $\sigma \in \bbS^p$ maximizes $\Psi$ over $\bbS^p$ if and only if $\sigma \in \bigcap_{G\in \cG} \cL^{\cup}(\cE(G))$,
        %$\Psi$ attains its maximum value over $\bbS^p$ at $\sigma$ if and only if $\sigma \in \bigcap_{G\in \cG} \cL^{\cup}(\cE(G))$,
        \item [(2)] For each $k \in [K]$, $\cM^{\Psi}(G^{(k)}) \subseteq \cM^{\psi^{(k)}}(G^{(k)})$.
    \end{itemize}
\end{lemma}
\begin{proof}
    See Section~\ref{subsub:lemma:joint_score} in the supplement.
\end{proof}
%The lemma indicates that joint estimation can improve identifiability. We illustrate this idea with the following example.

\begin{figure}[!b]
    \centering
    \captionsetup{font={footnotesize,stretch=1}}
    \begin{minipage}[c]{0.3\linewidth}
        \centering
        \includegraphics[width=\linewidth]{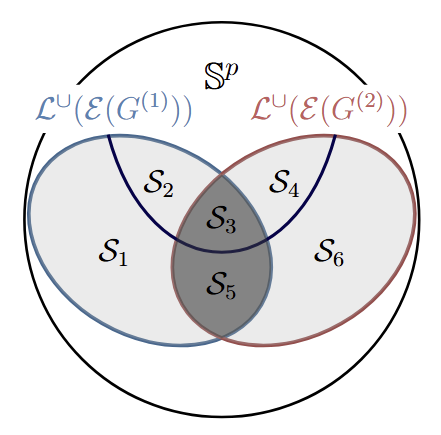}
    \end{minipage}
    \hspace{0.02\linewidth} % adjust this
    \begin{minipage}[c]{0.65\linewidth}
        \centering
        \includegraphics[width=\linewidth]{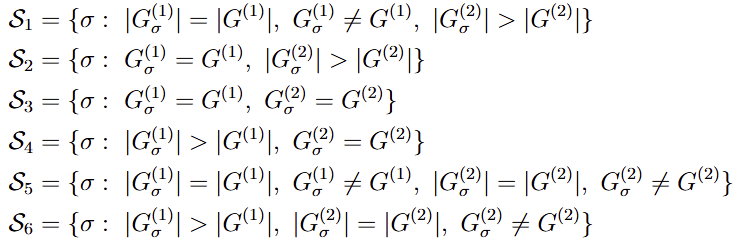}
    \end{minipage}
    \caption{Score landscape in $\bbS^p$ induced by the joint score $\Psi_1$ in Example~\ref{ex:landscape}. Regions labeled $\cS_1$--$\cS_6$ correspond to different sets of orderings, as indicated in the right panel. The blue and red regions represent the union of all orderings consistent with Markov equivalent DAGs of $G^{(1)}$ and $G^{(2)}$, respectively. The maximizers of the joint score correspond to the intersection of these regions $S_3 \cup S_5$. Region $S_3$ consists of orderings consistent with both true DAGs, whereas $S_5$ consists of orderings consistent with distinct Markov equivalent DAGs in both datasets. 
    Color intensity represents the  joint score value, with darker shades indicating higher scores.}
    \label{fig:ex1}
\end{figure}
\begin{example}\label{ex:landscape}
    Suppose that $\sX^{(1)}, \sX^{(2)}$ are generated from SCMs with DAGs $G^{(1)}$ and $G^{(2)}$, respectively. Under the faithfulness assumption, we define the joint score as $\Psi_1(\sigma)= \psi_1^{(1)}(\sigma) +\psi_1^{(2)}(\sigma) = - (|G_\sigma^{(1)}| + |G_\sigma^{(2)}|)$, where each individual score is obtained from the sparsest permutation method in Example~\ref{ex:sp}.  Under the common ordering assumption, the score landscape is illustrated in Figure~\ref{fig:ex1}.  In the figure, we have $\mathcal L(G^{(1)}) = \cS_2 \cup \cS_3$, $\mathcal L(G^{(2)}) = \cS_3 \cup \cS_4$, $\mathcal L^{\cup}(\mathcal E(G^{(1)})) = \cS_1 \cup \cS_2 \cup \cS_3 \cup \cS_5$, and $\mathcal L^{\cup}(\mathcal E(G^{(2)})) = \cS_3 \cup \cS_4 \cup \cS_5 \cup \cS_6$.
By Lemma~\ref{lemma:joint_score} (1), the joint score $\Psi_1$ is maximized on $\mathcal L^{\cup}(\mathcal E(G^{(1)})) \cap \mathcal L^{\cup}(\mathcal E(G^{(2)})) = S_3 \cup S_5$. 
For the individual scores, the misspecified ordering sets are  $\cM^{\psi_1^{(1)}}(G^{(1)}) =\cS_1 \cup \cS_5$  and $\cM^{\psi_1^{(2)}}(G^{(2)}) = \cS_5 \cup \cS_6$. In contrast, under the joint score, the misspecified orderings reduce to
$\cM^{\Psi_1}(G^{(1)}) = \cM^{\Psi_1}(G^{(2)}) = \cS_5 $.
\end{example}
This example illustrates how pooling heterogeneous datasets shrinks the set of maximizers and reduces the set of misspecified orderings.
The result extends to the general case with an arbitrary number $K$ DAGs in $\cG$. As $K$ increases, it can further refine the space of orderings, thereby reducing inconsistent causal orderings. 
One important observation is that if all graphs in $\cG$ are identical to a common DAG $G'$, so that $\bigcap_{G\in \cG} \cL^{\cup}(\cE(G)) = \cL^{\cup}(\cE(G')),$ 
then the joint score provides no additional reduction in the set of misspecified orderings. In this case, the resulting set coincides with that obtained under single DAG estimation.
Interestingly, this finding contrasts with the prevailing view in the multiple DAG learning literature, which suggests that greater similarity among the graphs in $\cG$ facilitates improved estimation~\citep{Wang2020-of, lee2022bayesian}.
In contrast, our analysis shows that increased heterogeneity across the graphs in $\cG$, particularly when it induces greater diversity in the sets $\{\cL^{\cup}(\cE(G)): G \in \cG \}$, can lead to improved identifiability of the causal ordering.
This naturally raises the question of the maximum possible improvement and the conditions under which such improvement can be achieved. 
The primary challenge is that the set $\cL^\cup(\cE(G))$, the union of orderings consistent with all Markov equivalent DAGs of $G$, is generally complex and difficult to characterize. Consequently, analyzing the intersection
$\bigcap_{G\in \cG} \cL^{\cup}(\cE(G))$ directly becomes challenging. 
The following lemma identifies a tractable superset that can be analyzed more easily. Recall that $E_G^*$ denotes the set of essential arrows of $G$, and that $\cL^{e}(E)$ denotes the set of linear extensions of an edge set $E$. 
\begin{lemma}\label{lemma:essential}
    For any collection of DAGs $\cG$,
    \begin{align*}
         \bigcap_{G\in \cG} \cL^{\cup}(\cE(G)) \subseteq \cL^{e} \left( \bigcup_{G\in \cG}  E^*_{G}\right).
    \end{align*}
\end{lemma}
\begin{proof}
    See Section~\ref{subsub:lemma:essential} in the supplement.
\end{proof}
The key observation is that $\cL^{\cup}(\cE(G)) \subseteq \cL^{e}(E^*_{G})$ for any given $G$. 
If $i \to j$ is an essential arrow of $G$, then it appears in every DAG in $\cE(G)$. 
Therefore, any ordering $\sigma$ consistent with a DAG in $\cE(G)$ must satisfy 
$\sigma^{-1}(i) < \sigma^{-1}(j)$; that is, node $i$ precedes node $j$ in $\sigma$. Hence $\sigma \in \cL^{e}(E_G^*)$.
The inclusion is generally strict, and equality need not hold. A concrete example illustrating this is provided in Section~\ref{subsec:notation_lemmas} in the supplement (Figure~\ref{fig:order_set}).

We now state the main result of this section. 
By controlling the set $\bigcup_{G\in \cG}  E^*_{G}$, we derive sufficient conditions on the collection $\cG$ for achieving the maximum possible improvement in the identifiability of causal orderings.
Under the common ordering assumption, let $\sigma^*$ be an ordering consistent with every DAG in $\mathcal{G}$. Define $\sigma^\dagger$ to be the ordering obtained from $\sigma^*$ by swapping its first two elements, that is,
\begin{align*}
    \sigma^\dagger(1) = \sigma^*(2), \qquad
\sigma^\dagger(2) = \sigma^*(1), \qquad
\sigma^\dagger(j) = \sigma^*(j) \quad \text{for all } j \ge 3.
\end{align*}
Define the edge set 
\begin{align*}
    E_{\max} = \left\{
\sigma^*(j) \to \sigma^*(j+1) : 2 \le j \le p-1 \right\} \cup \{\sigma^*(1) \to \sigma^*(3)\}.
\end{align*}
\begin{theorem}\label{thm:identifiability}
Assume that the common ordering assumption holds and  $\sX^{(k)}$ follow an SCM with $G^{(k)}$, for each $k \in [K]$. Let $\Psi$ be the joint order-based score defined in Definition~\ref{def:joint_score}. If 
\begin{align*}
    E_{\max} \subseteq \bigcup_{k = 1}^{K} E^*_{G^{(k)}},
\end{align*}
then $\arg\max_{\sigma \in \bbS^p} \Psi(\sigma)
=
\{\sigma^*, \sigma^\dagger\}.$
\end{theorem}
\begin{proof}
    See Section~\ref{subsub:thm1}  in the supplement.
\end{proof}
While the set $E_{\max}$ may appear arbitrary, this is not coincidental.
Because all DAGs in a Markov equivalence class share the same skeleton and v-structures, the essential arrow set $E^*_G$ always contains the edges forming v-structures in $G$. 
Other essential arrows are determined by orientation rules that preserve the existing v-structures while avoiding the creation of new ones, as characterized by~\citet{andersson1997characterization}. 
The key observation used here is that the edge $\sigma^*(1)\rightarrow \sigma^*(2)$, if present in $G$, cannot be an essential arrow, that is,  $(\sigma^*(1), \sigma^*(2) )\notin E^*_G$ for any  $G \in \cG_p^{\sigma^*}.$
To see this, the edge $\sigma^*(1)\rightarrow \sigma^*(2)$, if present in $G \in \cG_p^{\sigma^*}$, is always a covered edge and therefore is not essential~\citep{chickering1995transformational}. Consequently, $(\sigma^*(1), \sigma^*(2)) \notin E_G^*$ for any such $G$. This implies that no procedure based solely on essential arrows can uniquely identify the relative order between $\sigma^*(1)$ and $\sigma^*(2)$, and therefore no stronger identifiability result is attainable.
Although the remaining $p(p-1)/2 - 1$ edges may potentially be essential, Theorem~\ref{thm:identifiability} shows that only $p-1$ edges are needed to achieve the maximum improvement. 
Under a standard random DAG model, one can show that $K \ge C \log p$, for some constant $C>0$, is sufficient for the condition $E_{\max} \subseteq \bigcup_{k=1}^K E^*_{G^{(k)}}$
to hold with high probability, implying that a modest number of different datasets is needed. See Proposition~\ref{prop:K} of Section~\ref{subsub:propK} in the supplement. 
\begin{remark}[Common ordering assumption]\label{remark:common}
The common ordering assumption has been widely adopted in the multiple DAG estimation literature and is reasonable in many applications. For example, in systems biology, the directionality of genetic interactions is typically stable. Although specific interactions may emerge, disappear, or vary in strength across conditions, an upstream gene rarely becomes a downstream gene. As a result, biological pathways tend to preserve a consistent causal direction even under environmental or experimental perturbations. More broadly, this assumption is natural in settings where interventions or environmental changes are unlikely to reverse causal directionality, such as time-ordered systems or hierarchical processes. To assess the practical implications of violating this assumption, we conduct simulation studies in Section~\ref{subsec:common_ordering}. 
 The results indicate reliable convergence and stable performance when violations are mild, suggesting that the proposed method retains practical advantages.
\end{remark}

\begin{remark}[Faithfulness]\label{remark:faithfulness}
While several relaxations have been studied~\citep{raskutti2018learning}, the faithfulness assumption remains standard in the literature. 
 When a dataset $\sX$ violates faithfulness, an ordering that is not consistent with the true causal DAG may be preferred by the score function over orderings consistent with the true graph. To illustrate this, consider $\sX = (\sX_1, \sX_2, \sX_3)$ generated from the structural equation model
\begin{align*}
    \sX_1 = \varepsilon_1, \quad \sX_2 = a\,\sX_1 + \varepsilon_2, \quad \sX_3 = b\, \sX_1 + c\, \sX_2 + \varepsilon_3, \quad 
\end{align*}
where $\varepsilon_j  \overset{\mathrm{i.i.d.}}{\sim} N(0,1)$ for $ j \in [3]$. The corresponding causal graph $G$ is complete, with true causal ordering $\sigma^* = (1,2,3)$. Suppose faithfulness is violated; the nonzero parameters $a,b,c$ satisfy $\mathrm{Cov}(\sX_2, \sX_3) =ab + a^2c + c = 0$. Consider the score function $\psi_1$ defined in Example~\ref{ex:sp}. 
Under the true ordering $\sigma^*$, the minimal I-map DAG coincides with $G$, yielding $\psi_1(\sigma^*) = -|G| = -3$.
Now consider the incorrect ordering $\sigma = (2,3,1)$. The minimal I-map $G_\sigma$ of $\sX$ with respect to $\sigma$ is $2 \rightarrow 1 \leftarrow 3$, which contains only two edges, so that $\psi_1(\sigma) = -|G_\sigma| = -2$, which is strictly larger than $\psi_1(\sigma^*)$. 

This example demonstrates that order-based structure learning methods may fail under violation of faithfulness, particularly in the single DAG learning. Interestingly, our empirical results suggest that joint estimation is robust to such violations (see Section~\ref{subsec:faithfulness} in the supplement).  
The intuition is that although a violation of faithfulness in a single dataset may cause the individual score to prefer an incorrect ordering, this ordering is typically not consistent across other datasets. Consequently, when such an ordering is applied to those datasets, it tends to induce denser graphs and is therefore penalized under the joint score. As a result, joint estimation tends to discourage such spurious orderings and may help recover the true causal ordering.
\end{remark}

\section{Bayesian order-based joint structure learning for Gaussian DAG models}\label{sec:model}

\subsection{Model specification}\label{subsec:model}

We consider a Bayesian model formulation for joint structure learning of DAG models. This model serves as a concrete example of the framework proposed in the previous section. Suppose the data come from $K$ different sources. For each source $k \in [K]$, the observational dataset $X^{(k)} \in \mathbb{R}^{n_k \times p}$ consists of $n_k$ i.i.d.\ observations of a $p$-dimensional random vector $\sX^{(k)}= (\sX^{(k)}_1, \dots, \sX^{(k)}_p)$.
For each  $\sigma \in \bbS^p$  and $k \in [K]$,  we consider the following linear structural equation model for the random vector $\sX^{(k)}$ with underlying DAG $G^{(k)} \in \cG_p^\sigma$, 
\begin{equation}~\label{eq:ln.str.eq}
    \sX_j^{(k)} = \sum_{i \in \Pa_j(G^{(k)})}B^{(k)}_{i j} \,\sX_{i}^{(k)} + \se^{(k)}_j, \quad \se^{(k)}_j   \sim  N(0, \omega^{(k)}_j)  \text{ independently for } j \in [p], 
\end{equation}
where $\Pa_j(G^{(k)}) \subseteq P_j^\sigma$ for each $j$, and $P_j^\sigma$ is defined in~\eqref{eq:potential}. 
Here $B^{(k)}$ is a $p \times p$ matrix. Since entries of $B^{(k)}$ that are not involved in~\eqref{eq:ln.str.eq} are set to zero, 
$B^{(k)}$ can be seen as the weighted adjacency matrix of the DAG $G^{(k)}$, where $i \rightarrow j \in G^{(k)}$ if and only if  $B^{(k)}_{ij} \neq 0$.  

We adopt the empirical sparse Cholesky prior, which is commonly used for sparse DAG models~\citep{lee2019minimax, Zhou2023-cm}. This prior is motivated by ideas from empirical Bayes approaches to variable selection~\citep{martin2017empirical}.  The prior on the parameters $(\sigma, \{G^{(k)}, B^{(k)}, \Omega^{(k)} \}_{k=1}^K)$, where $\Omega^{(k)} = \diag(\omega_1^{(k)}, \dots, \omega_p^{(k)})$, is specified as follows. For all $ j \in [p], k\in [K],$ and $\sigma \in \bbS^p$,
\begin{align}
    B^{(k)}_{\Pa_j(G^{(k)}), j} \mid  G^{(k)},  \omega^{(k)}_j \overset{\mathrm{ind}}{\sim} \;&  \mathrm{MVN}_{|\Pa_j(G^{(k)})|}\left(\widehat{ B}^{(k)}_{\Pa_j(G^{(k)}), j},\frac{\omega^{(k)}_j}{\gamma}  \left(X_{\Pa_j(G^{(k)})}^{(k)\top}X_{\Pa_j(G^{(k)}) }^{(k)}\right)^{-1}\right), \label{eq:B}  \\ 
    \pi_0(\omega^{(k)}_j \mid \sigma) \propto \;& (\omega^{(k)}_j)^{-\frac{\kappa}{2}-1}, \qquad \pi_0(\sigma) \propto 1 \label{eq:omega} \\  
    \pi_0(G^{(1)}, \dots, G^{(K)} \,|\, \sigma) \propto \;& \prod_{k=1}^K \pi_0(G^{(k)} \,|\,  \sigma)  = \prod_{k=1}^K \left(p^{c_0} \right)^{-|G^{(k)}|} \ind_{ \{ \widehat{G}^{(k)}_\sigma \} }(G^{(k)}), \label{prior:G}  
\end{align}
where $\pi_0$ denotes the prior density (or mass) function, $\widehat{B}^{(k)}_{\Pa_j(G^{(k)}),j}$ is the least-squares estimator of $B^{(k)}_{\Pa_j(G^{(k)}),j} = (B^{(k)}_{ij})_{i \in \Pa_j(G^{(k)})}$ , and $c_0 > 0, \gamma > 0, \kappa > 0$ are hyperparameters of the prior. Following~\citet{martin2017empirical}, we introduce
a fractional power  $\alpha \in (0, 1)$ of the likelihood, often referred to as the $\alpha$-fractional posterior~\citep{bhattacharya2019bayesian}, to prevent the posterior from becoming overly influenced by the empirical prior. 
By standard normal-inverse-gamma calculations, we obtain the posterior distribution given the data matrices $X = \{X^{(k)}\}_{k=1}^K$, 
\begin{align*}
    \pi(G^{(1)}, \dots, G^{(K)}, \sigma ) \propto \left\{\prod_{k=1}^K \pi^{(k)}(G^{(k)} \mid \sigma) \right\} \pi_0(\sigma),
\end{align*}
where 
$\pi^{(k)}(G^{(k)} \mid  \sigma) \propto  \,\exp\left( \phi^{(k)}(G^{(k)})\right) \ind_{ \{ \widehat{G}^{(k)}_\sigma \} }(G^{(k)})$ defines the (unnormalized) conditional posterior contribution corresponding to the $k$-th dataset $X^{(k)}$ with 
\begin{align}\label{eq:score}
   \phi^{(k)}(G^{(k)})
= -\!\left(
c_0 \log p
+ \frac{1}{2}\log\!\left(1+\frac{\alpha}{\gamma}\right)
\right) |G^{(k)}|
- \frac{\alpha n_k + \kappa}{2}
\sum_{j=1}^p \log \bigl(n_k\,\widehat{\omega}^{(k)}_j(\Pa_j(G^{(k)}))\bigr),
\end{align}
where $\widehat{\omega}^{(k)}_j(S)$ is the residual variance from the linear regression of $X^{(k)}_j$ onto $X^{(k)}_{S}$, that is, $\widehat{\omega}^{(k)}_j(S)
= n_k^{-1}\,
X^{(k)\top}_j
\!\left(
I
- X^{(k)}_{S}
\bigl(X^{(k)\top}_{S} X^{(k)}_{S}\bigr)^{-1}
X^{(k)\top}_{S}
\right)
X^{(k)}_j.$ See Section~\ref{sec:derivation}  in the supplement for a detailed derivation.
When $K=1$, this formulation reduces to the standard conditional posterior distribution for a single dataset. Finally, we obtain the marginal posterior distribution
\begin{equation}\label{eq:post}
        \pi(\sigma) = 
        \sum_{G^{(1)}\in \cG_p^\sigma } \cdots
        \sum_{G^{(K)}\in \cG_p^\sigma } \pi(G^{(1)}, \dots, G^{(K)}, \sigma), 
\end{equation}
which quantifies the uncertainty over the ordering space $\bbS^p$. 
We will show that \eqref{eq:post} serves as a joint order-based score as defined in Definition~\ref{def:joint_score}. 

The prior has
another empirical component: 
$\widehat{G}^{(k)}_\sigma$ in~\eqref{prior:G} serves as the best estimate for the data $X^{(k)}$ within $\cG_p^\sigma$, as is standard in literature~\citep{agrawal2018minimal, Chang2022-id}. We choose the maximum a posteriori (MAP) estimate defined as 
\begin{equation}\label{eq:MAP}
    \widehat{G}^{(k)}_\sigma = \arg\max_{G \in \cG_p^\sigma} \pi^{(k)}(G \mid \sigma).
\end{equation}
In high-dimensional regimes, we impose sparsity by introducing a maximum indegree parameter $d$ when searching for the MAP estimate, which replaces $\cG_p^\sigma$ with the restricted space $\cG_p^\sigma(d)$ defined in \eqref{eq:restricted}.

\subsection{High-dimensional results}\label{subsec:highdim}

We now present high-dimensional results. Let $n = \min_{k \in[K]} n_k$, and assume that $n_k = \Theta(n)$ for all $k\in [K]$. We consider a high-dimensional setting in which  $p = p(n)$, $K=K(n)$, and $d = d(n)$ may grow with $n$. 
As in~\eqref{eq:ln.str.eq}, for each $k\in[K]$, we assume that the random vector $\sX^{(k)}$ follows a linear SCM associated with a DAG $G^{(k)}_*$ with Gaussian errors and true parameter values $B_*^{(k)}, \Omega_*^{(k)} =\diag(\omega^{(k)}_{*,1}, \dots, \omega^{(k)}_{*,p})$, 
\begin{equation}\label{eq:true_pop_data}
    \sX_j^{(k)} = \sum_{i \in \Pa_j(G_*^{(k)})}
    (B^{(k)}_*)_{i j} \,\sX_{i}^{(k)} + \se^{(k)}_j, \quad \se^{(k)}_j   \sim  N(0, \omega^{(k)}_{*,j}) \text{ independently for } j \in [p], 
\end{equation}
and that the dataset $X^{(k)}$ consists of $n_k$ i.i.d.\ observations of $\sX^{(k)}$. Since the errors are Gaussian, $\sX^{(k)}$ follows a multivariate normal distribution for each $k \in [K]$. Let $\Sigma^{(k)}$ denote its covariance matrix, that is, $\sX^{(k)} \sim \mathrm{MVN}_p(0, \Sigma^{(k)})$. We assume without loss of generality that the variables are centered.
Although the distribution of $\sX^{(k)}$ is fully specified by the covariance matrix $\Sigma^{(k)}$, different structural equation models with different DAGs can induce the same distribution. To see this, given an ordering $\sigma$, we define the linear projection of $\sX_j^{(k)}$ onto its predecessors
$\{\sX_i^{(k)}: i \in P_j^\sigma\}$ as
\begin{equation}\label{eq:permutedSEM}
\sX_j^{(k)} = \sum_{i \in P_j^\sigma} (B_\sigma^{(k)})_{ij}\,\sX_i^{(k)} + \se_{\sigma,j}^{(k)},
\quad
\se_{\sigma,j}^{(k)} \sim \mathcal{N}\!\left(0,\omega_{\sigma,j}^{(k)}\right)
 \text{ independently for } j \in [p].
\end{equation}
%where $\Omega_\sigma^{(k)}=\diag\!\bigl(\omega_{\sigma,1}^{(k)},\dots,\omega_{\sigma,p}^{(k)}\bigr)$.
Under Gaussianity of $\sX^{(k)}$, we have $(B_\sigma^{(k)})_{ij}\neq 0$ if and only if $\sX^{(k)}_i \not\!\perp\!\!\!\perp \sX^{(k)}_j \mid \sX^{(k)}_{P_j^\sigma\setminus\{i\}}$. Let $G_\sigma^{(k)}$ denote the DAG induced by the weighted adjacency matrix $B_\sigma^{(k)}$. Then $G_\sigma^{(k)}$ is the minimal I-map of the distribution of $\sX^{(k)}$  with respect to $\sigma$. If the distribution is faithful to $G_*^{(k)}$, $G_\sigma^{(k)}$ is a minimal I-map of $G_*^{(k)}$ with respect to $\sigma$. 

We assume the faithfulness and common ordering assumptions. Additionally, we introduce the following assumptions to derive high-dimensional results.
\begin{enumerate}
    \item[(A1)] There exist $\underline{\nu}=\underline{\nu}(n),\ \overline{\nu}=\overline{\nu}(n)>0$ and a universal constant $\delta_0>0$ such that
    \begin{align*}
        0 < \frac{\underline{\nu}}{(1-\delta_0)^2}
\le \lambda_{\min}(\Sigma^{(k)})
\le \lambda_{\max}(\Sigma^{(k)})
\le \frac{\overline{\nu}}{(1+\delta_0)^2}\quad \text{for all }k\in [K],
    \end{align*}
where $\lambda_{\min},\lambda_{\max}$ denote the smallest and largest eigenvalues, respectively. 
\item[(A2)] The parameter $K$ and $d$ satisfy $\log K + d \log p = o(n), $ $K = o(p)$ and $d^* \leq d$ where
\begin{equation}\label{eq:d*}
    d^* = \max_{\sigma \in \mathbb{S}^p} \max_{k\in[K]} \max_{j\in[p]} | \mathrm{Pa}_j(G^{(k)}_\sigma) |.
\end{equation}
\item[(A3)]  Prior parameters satisfy that $\kappa \le n,\sqrt{1+\alpha/\gamma} \le p,$ and $c_0 \ge \rho(\alpha+1) + t$, for some universal constant $t>0$ and for some $\rho = O(d)$.
\item[(A4)] For every $\sigma \in \mathbb{S}^p$, $\min_{i,j,k}\left\{ |(B^{(k)}_\sigma)_{ij}|^2 : (B^{(k)}_\sigma)_{ij} \neq 0 \right\}
\ge 16c_0\overline{\nu}^{\,2}\log p/(\alpha \underline{\nu}^{\,2} n).$
\end{enumerate}

The assumptions are standard in high-dimensional structure learning literature. Assumption (A1), commonly referred to as a restricted eigenvalue condition, requires the population covariance matrices to be uniformly well conditioned across groups. Assumption (A2) restricts the growth rates of $K, d,$ and $ p$ relative to the minimum sample size $n$. In particular, it implies $d \log p = o(n)$, which is the typical rate in the literature on sparse DAG recovery \citep{lee2019minimax}. Assumption (A3) imposes mild growth conditions on the prior hyperparameters. In particular, the lower bound on $c_0$ plays a key role in controlling false edge inclusions in high dimensions, as it serves as a penalization parameter on graph size. Assumption (A4) is a permutation beta-min condition, which requires that the signals corresponding to nonzero regression coefficients in~\eqref{eq:permutedSEM} for any $\sigma \in \bbS^p$ are sufficiently large to be detected.  
% It is quite Note that we do not assume strong faithfulness, as it is known to be a restrictive condition \citep{Uhler2013-ra}.
The proposition establishes that the posterior distribution $\pi$ in \eqref{eq:post} asymptotically satisfies the joint order-based scoring property defined in Definition~\ref{def:joint_score}. 

\begin{proposition}\label{prop:score-eq} 
Assume the faithfulness condition. Suppose (A1)--(A4) hold. Then, for sufficiently large $n$ with probability $1-p^{-1}$,
\begin{itemize}
    \item [(1)] For all $\sigma \in \bbS^p$ and $k\in[K]$, $\phi^{(k)}(G)$ is maximized over $G \in \cG_p^\sigma(d)$ at $G_\sigma^{(k)}$. 
    \item [(2)] For all $k\in [K]$, $\pi^{(k)}( \sigma) = \sum_{G\in \cG_p^\sigma }  \pi^{(k)}(G , \sigma)$ is a score equivalent order-based score as defined in Definition~\ref{def:score-eq}.
\end{itemize}
\end{proposition}
\begin{proof}
    See Section~\ref{subsub:MAP}  in the supplement.
\end{proof}
According to the proposition, under the condition $E_{\max} \subseteq \bigcup_{k  = 1}^K E^*_{G^{(k)}_*},$ the posterior probability $\pi$ in~\eqref{eq:post} is maximized at two orderings, as stated in Theorem~\ref{thm:identifiability}. 
We now show that the Bayes factor between two orderings diverges to infinity if one ordering belongs to $\bigcap_{G\in \cG_*} \cL^{\cup}(\cE(G))$ while the other does not. %In the most favorable case, only two orderings attain the maximum posterior value, and any other ordering yields a distinctively lower score.

\begin{theorem}\label{thm:score-eq}

Assume the faithfulness and common ordering conditions. Suppose (A1)--(A4) hold. Then, for sufficiently large $n$, with probability at least $1-p^{-1}$,
\begin{itemize}
    \item [(1)] The posterior distribution $\pi$ in~\eqref{eq:post} satisfies Definition~\ref{def:joint_score}.
    \item [(2)] For $\sigma \in \bigcap_{G\in \cG_*} \cL^{\cup}(\cE(G))$ and $\tau \notin \bigcap_{G\in \cG_*} \cL^{\cup}(\cE(G))$,  if $p = p(n) \to \infty$, the $\alpha$-fractional Bayes factor 
$\mathrm{BF}^{\alpha}_{\sigma, \tau}(X)  \rightarrow \infty$. 
\end{itemize}
\end{theorem}
\begin{proof}
    See Section~\ref{subsub:score-eq} in the supplement.
\end{proof}
We illustrate this result through simulations designed to emulate a high-dimensional asymptotic regime. See Section~\ref{subsec:highdim_simul} in the supplement.

\section{Efficient computation of the posterior distribution}

\subsection{Sampling from the ordering space}\label{subsec:compute}

Sampling-based methods are the standard approach for numerically approximating posterior distributions when closed-form expressions are available only up to a normalizing constant. In our setting, the posterior distribution is defined over the permutation space $\bbS^p$, which grows super-exponentially in $p$ and is therefore prohibitively large for the purpose of computing the normalizing constant. A common approach is to employ Markov chain Monte Carlo (MCMC), which constructs a Markov chain whose stationary distribution coincides with the target posterior.
For sampling on discrete spaces, the Metropolis–Hastings (MH) algorithm proposes a candidate state and accepts or rejects it according to an acceptance probability. The algorithm explores the state space using a predefined set of one-step transitions from the current state, defined by a neighborhood mapping $\cN: \bbS^p \to 2^{\bbS^p}$,  where $\cN(\sigma)$ specifies the set of candidate orderings that can be proposed from the current ordering $\sigma$. 
Algorithm~\ref{alg:full} presents the overview of the scheme.

\begin{algorithm}[!t]
\footnotesize
\captionsetup{font={footnotesize,stretch=1}}
\caption{Bayesian order-based joint structure learning of Gaussian DAG \\
models with random walk MH algorithm}
\label{alg:full}
\KwIn{A set of data matrices $\{X^{(k)}\}_{k=1}^{K}$ and an initial ordering $\sigma_0$}

Calculate $\{\widehat{G}_{\sigma_0}^{(k)}\}_{k=1}^K$  using forward-backward edge selection (Algorithm~\ref{alg:stepwise}), and evaluate $\pi(\sigma_0)$\\

\For{$t =  1, \ldots, T$}{
    Draw a proposal ordering $\sigma$ uniformly from $\cN(\sigma_{t-1})$\\
    Calculate $\{\widehat{G}_\sigma^{(k)}\}_{k=1}^K$  using forward-backward edge selection (Algorithm~\ref{alg:stepwise}), and evaluate $\pi(\sigma)$\\
    Set $a \leftarrow \min(\pi(\sigma)/\pi(\sigma_{t-1}), 1)$ and draw $u \sim \mathrm{Unif}[0,1]$\\
    \eIf{$u \leq a$}{
        $\sigma_t \leftarrow \sigma$; \quad         $\widehat{G}_{\sigma_t}^{(k)} \leftarrow \widehat{G}_{\sigma}^{(k)}$ for all $k \in [K]$
    }{
        $\sigma_t \leftarrow \sigma_{t-1}$; \quad         $\widehat{G}_{\sigma_t}^{(k)} \leftarrow \widehat{G}_{\sigma_{t-1}}^{(k)}$ for all $k \in [K]$
    }
    
    }
\KwOut{Sampled orderings $\{\sigma_t\}_{t=1}^T$ and corresponding DAGs $\bigl\{ \{\widehat G_{\sigma_t}^{(k)}\}_{k=1}^K \bigr\}_{t=1}^T$}
\end{algorithm}

There are two key components of Algorithm~\ref{alg:full} that require additional discussion. The first concerns evaluation of the unnormalized posterior probability at the $t$-th iteration. To compute the posterior value at a given ordering $\sigma$, it is necessary to obtain the corresponding MAP DAG $\widehat{G}_\sigma$  defined in~\eqref{eq:MAP}. Even for a given ordering, the number of candidate DAGs grows on the order of $p^d$, making exhaustive search impractical in high-dimensional settings. We estimate the DAGs by performing forward–backward edge selection. This procedure exploits the decomposability of the score function $\phi^{(k)}$ in \eqref{eq:score} \citep{chickering2002learning, friedman2003being}, which we write as
$\phi^{(k)}(G^{(k)}) = \sum_{j =1}^p \phi^{(k)}_j(\Pa_j(G^{(k)})) $, where 
\begin{equation}\label{eq:nodewise}
        \phi^{(k)}_j(S) = -\!\left(
c_0 \log p
+ \frac{1}{2}\log\!\left(1+\frac{\alpha}{\gamma}\right)
\right) |S|
- \frac{\alpha  n_k + \kappa}{2}
\log \bigl(n_k\,\widehat{\omega}^{(k)}_j(S)\bigr),
\end{equation}
for an index set $S$. Owing to this decomposable structure, for each node $j$, parent set selection can be carried out independently for each node given the ordering $\sigma$. 
%We initialize its parent set as empty and greedily add nodes in $P_j^\sigma$ that yield the largest improvement in the local score, continuing until no further improvement is possible. We then perform a backward pruning step, iteratively removing any parent whose deletion increases the score. 
The procedure is described in Algorithm~\ref{alg:stepwise} in the supplement (Section~\ref{sec:forwardbackward}). We prove its consistency.
\begin{proposition}\label{prop:MAP}
    Assume the faithfulness condition. Suppose (A1)--(A4) hold, and the forward step in Algorithm~\ref{alg:stepwise} terminates with at most $d$ parents per node. Then, for any ordering $\sigma \in \bbS^p$ and $k \in [K]$, Algorithm~\ref{alg:stepwise} outputs $\widehat{G}_\sigma^{(k)} = G_\sigma^{(k)}$ for sufficiently large $n$, with probability at least $1 - p^{-1}$.
\end{proposition}
\begin{proof}
    See Section~\ref{subsub:MAP_stepwise}  in the supplement.
\end{proof}
The condition for the forward step output is shown to hold with high probability when choosing $d = O(d^*)$~\citep{an2008stepwise}.

\begin{remark}[On parallel computation]\label{remark:parallel}
When evaluating the posterior at the \(t\)-th iteration, computation can be substantially accelerated through parallelization. Although the computational cost grows linearly with the number of datasets \(K\), the posterior contributions corresponding to different datasets can be evaluated independently. Consequently, the computation can be fully parallelized across \(K\) cores, resulting in wall-clock time comparable to the single dataset case. Moreover, owing to the decomposability across nodes, computation can be further parallelized across the \(p\) local node–parent scores within each dataset. With \(Kp\) available cores, the overall computation can be accelerated by up to a factor of \(Kp\).
\end{remark}

\subsection{Random-to-random proposal neighborhood}\label{subsec:r2r}

The mixing behavior of the Markov chain in Algorithm~\ref{alg:full} critically depends on the choice of the proposal neighborhood $\cN$.
We propose the random-to-random neighborhood $\mathcal{N}_{\mathrm{R2R}}$, motivated by the classical random-to-random shuffle~\citep{bernstein2019cutoff}, and define it using the operators $\operatorname{R2R}_{<}$ and $\operatorname{R2R}_{>}$.
Specifically, $\operatorname{R2R}_{<}(\sigma, i, j)$ denotes the operator that produces a new permutation obtained from $\sigma$ by removing the element at position $j$ and inserting it before the element at position $i$, and can be expressed as
\begin{align*}%\label{eq:r2r}
    \mathrm{R2R}_{<}(\sigma, i, j)\,:\, (\sigma(1), \dots, \sigma(i), \dots,  \sigma(j),  \dots, \sigma(p))  \mapsto (\sigma(1), \dots, \sigma(j), \sigma(i),  \dots, \sigma(p)),
\end{align*}
for $i < j$. Similarly, $ \operatorname{R2R}_{>}(\sigma, i, j)$ inserts the element at position $j$  after the element at position $i$, and can be written as, for $i > j$,
\begin{align*}%\label{eq:r2r-rev}
    \mathrm{R2R}_{>}(\sigma, i, j):\, (\sigma(1), \dots, \sigma(j), \dots,  \sigma(i),  \dots, \sigma(p))  \mapsto (\sigma(1), \dots, \sigma(i), \sigma(j),  \dots, \sigma(p)),
\end{align*}
Finally, for $i \neq j$, we have $\operatorname{R2R}(\sigma,i,j)
= \operatorname{R2R}_{<}(\sigma,i,j)$ if $i<j$, and
$\operatorname{R2R}_{>}(\sigma,i,j)$ if $i>j$,
and the corresponding neighborhood is given by $\cN_{\mathrm{R2R}}(\sigma) = \{\sigma' \in \bbS^p: \sigma' = \operatorname{R2R}(\sigma, i, j) \text{ for } i \neq j\}$. A recent study has shown that a hill climbing algorithm using the neighborhood $\mathcal{N}_{\mathrm{R2R}_{<}}$ is particularly effective for order-based structure learning, whereas other neighborhoods of the same size exhibit suboptimal performance; thus, it is recommended that the constructed neighborhood include at least $\mathcal{N}_{\mathrm{R2R}_{<}}$~\citep{chang2025identifiability}. 
%Following \citet{Zhou2023-cm}, who showed that greedy search and random-walk Metropolis–Hastings algorithms have similar computational complexity, we incorporate the neighborhood $\mathcal{N}_{\mathrm{R2R}_{<}}$ into our algorithm.
To ensure that the Markov chain converges to the posterior distribution, we include $\mathcal{N}_{\mathrm{R2R}_{>}}$ so that the resulting neighborhood $\mathcal{N}_{\mathrm{R2R}}$ is symmetric, that is,  $\sigma' \in \cN_{\mathrm{R2R}}(\sigma)$ if and only if  $\sigma \in \cN_{\mathrm{R2R}}(\sigma')$, for the reversibility of the chain. Note that the neighborhood size is $(p-1)^2$ after removing duplicates. The following proposition further explains why the neighborhood $\mathcal{N}_{\mathrm{R2R}_{<}}$ is effective. 

\begin{proposition}\label{prop:covered_minimap}
    Let $G_\sigma$ be the minimal I-map of the true DAG $G^*$ with respect to $\sigma \in \bbS^p$. If $G_\sigma \neq G^*$, then there exists an ordering $\tau \in \mathcal N_{\mathrm{R2R}_{<}}(\sigma)$ such that  $|G_\tau| \leq |G_\sigma|$. 
\end{proposition}
\begin{proof}
See Section~\ref{subsub:covered_minimap} in the supplement.
\end{proof}
\begin{corollary}\label{cor:covered}
Assume the faithfulness condition. Suppose (A1)--(A4) hold.  For any ordering $\sigma \notin \mathcal L(G^{(k)})$, there exists an ordering $\tau \in \mathcal N_{\mathrm{R2R}_{<}}(\sigma)$ such that $\pi^{(k)}(\tau)\geq \pi^{(k)}(\sigma)$. 
\end{corollary}
\begin{proof}
See Section~\ref{subsub:covered} in the supplement.
\end{proof}
Proposition~\ref{prop:covered_minimap} shows that the neighborhood $\mathcal{N}_{\mathrm{R2R}_{<}}$ contains an operator that moves from one minimal I-map to another without increasing the number of edges. In this sense, the operator emulates a step in the Chickering sequence~\citep{chickering2002optimal}, namely a covered edge reversal combined with the deletion of unnecessary edges. With repeated applications, these moves allow transitions to a Markov equivalent DAG of the true DAG. A related scheme was introduced in~\cite{solus2017consistency}, where a depth-first search strategy based on covered edge reversals combined with edge deletions was proposed and shown to be consistent. Our contribution is novel in that it interprets these operations as proposal operators in the ordering space and shows that the R2R neighborhood implicitly includes covered edge reversals combined with edge deletions.  Corollary~\ref{cor:covered} establishes the existence of a monotone scoring path in the ordering space with respect to each individual score. We note that analyzing such scoring path under the joint score is considerably more involved. 

\section{Simulation study}\label{sec:simulation}

\noindent
\textbf{Data simulation.}
A data matrix $X$ is generated according to the linear structural equation model in~\eqref{eq:ln.str.eq} using the following standard scheme. We first fix an ordering $\sigma \in \bbS^p$. Given the ordering, we generate a $p$-node random DAG $G$ by connecting each ordered pair of distinct nodes $(\sigma(i), \sigma(j))$ with $i < j$ by a directed edge $\sigma(i)\to \sigma(j)$ independently with probability $p_\mathrm{edge}$. We set $p_\mathrm{edge} = 3/(2p-2)$ unless stated otherwise.
Given $G$, we sample each edge weight $B_{ij}$ for $i\to j\in G$ independently from the uniform distribution on $[-1, -0.5] \cup [0.5, 1]$, and set all error variances to one. To match the sample size used in the real data example, we fix $n_k = 1000$ for all $k\in[K]$.

\noindent
\textbf{Hyperparameters.}
Unless otherwise specified, we run Algorithm~\ref{alg:full} with hyperparameters $\alpha= 0.99$, $\gamma = 0.01$, $\kappa = 0$, $c_0 = 3$, $d = p$, and run the MCMC sampler for $T = 20 \,p^2$ iterations, using samples after a burn-in period $T' = T/2$.

\noindent
\textbf{Evaluation criteria.}
For a DAG \(G\), let \(\Gamma^G = (\Gamma^G_{ij}) \in \{0,1\}^{p \times p}\) denote its adjacency matrix, where
\(\Gamma^G_{ij} = 1\) if and only if there is a directed edge \(i \to j\) in \(G\), and \(\Gamma^G_{ij} = 0\) otherwise. Since we take a Bayesian approach, we consider the posterior mean DAG $\widehat{G}_\pi$ and its weighted adjacency matrix $\Gamma^\pi \in [0,1]^{p \times p}$, computed from the Markov chain sample $\{ G_{t}\}_{t = 1}^{T}$. Each entry of $\Gamma^\pi$ is defined as the empirical posterior edge inclusion probability
\begin{align}\label{eq:pip}
    \Gamma^{\pi}_{ij} = \widehat{\bbP}_\pi \left(i\to j \in G \right) = \frac{1}{T} \sum_{t=1}^T \ind \left\{i\to j \in G_{t}\right\}.
\end{align}
We define the Hamming distance between any two DAGs \(G\) and \(G'\) as $\mathrm{HD}(G, G') = \sum_{i,j} |\Gamma^G_{ij} - \Gamma^{G'}_{ij}|,$
which measures the edge-wise difference between \(G\) and \(G'\). To measure estimation accuracy, we compute the average Hamming distance between the true DAGs and their estimates across all \(K\) data sources
\begin{align}\label{eq:delta}
\Delta
\;=\;
\frac{1}{K}\sum_{k=1}^K 
\mathrm{HD}\!\left(G^{(k)}, \widehat{G}^{(k)}\right).
\end{align}
When simulations are conducted under the common ordering assumption with true ordering $\sigma^*$, we measure the posterior mean rank correlation (Kendall's tau) from the Markov chain sample $\{ \sigma_t\}_{t=1}^T$ as
\begin{align}\label{eq:tau}
\tau^* \;=\; \frac{1}{T}\sum_{t=1}^T \tau_{\mathrm{Kendall}}\!\left(\sigma^*, \sigma_t\right),
\end{align}
where values of $\tau^*$ closer to 1 indicate the algorithm successfully recovers the causal ordering.
True positive rate (TPR) and false discovery rate (FDR) between any true DAG \(G\) and estimated DAG \(\widehat{G}\) are defined as $\sum_{i,j}\Gamma^G_{ij}\Gamma^{\widehat{G}}_{ij}/|G|$, and $ \sum_{i,j}(1-\Gamma^{G}_{ij})\Gamma^{\widehat{G}}_{ij}/|\widehat{G}|$, respectively. 

 %false positive rate (FPR),$\sum_{i,j}(1-\Gamma^G_{ij})\Gamma^{\widehat{G}}_{ij}/\sum_{i,j}(1-\Gamma^G_{ij})$, 

%To compare with other methods, we compute the true positive rate (TPR), the false positive rate (FPR), and the false discovery rate (FDR). We define the true positive edges between two DAGs $G$ and $G'$ as $\mathrm{TP} = \sum_{i,j}\mathbb{I}\!\left\{\Gamma^{G'}_{ij}=1 ,\Gamma^{G}_{ij}=1\right\}$, true negative edges between two DAGs

\subsection{Performance evaluation}\label{subsec:performance}

We assess the performance of the proposed method and examine how increasing the number of datasets
$K$ improves performance. We fix $\sigma^* = (1, \dots, p)$, and generate $K$ random DAGs with $p$ nodes.  For each graph $G^{(k)}$, $k \in[K]$, we generate the data matrix $X^{(k)}$ as described above. Given the simulated datasets \(\{X^{(k)}\}_{k = 1}^K\), we apply Algorithm~\ref{alg:full} with the proposal neighborhood $\cN_{\mathrm{R2R}}$. For each experiment, we calculate the posterior Hamming distance $\Delta$, and measure the posterior mean rank correlation $\tau^*$. 
\begin{table}[!t]
\centering
%\footnotesize
\captionsetup{font={footnotesize,stretch=1}}
\renewcommand{\arraystretch}{0.8}
\begin{tabular}{ccccccc}
\toprule
\multirow{2}{*}{} & \multicolumn{5}{c}{$p$} \\
\cmidrule(lr){3-7}
  & & 5 & 10 & 20 & 40 & 100\\ 
\midrule 
\multirow{2}{*}{$K$=1} 
 & $\Delta$ & 1.55 $\pm$ 0.04 & 2.22 $\pm$ 0.18 & 4.99 $\pm$ 0.27 & 9.18 $\pm$ 0.38 & 19.35 $\pm$ 0.45\\
  & $\tau^*$ & 0.25 $\pm$ 0.01 & 0.25 $\pm$ 0.02 & 0.23 $\pm$ 0.01 & 0.20 $\pm$ 0.01 & 0.20 $\pm$ 0.01\\
\midrule 
\multirow{2}{*}{$K$=5}
 & $\Delta$ & 0.98 $\pm$ 0.10 & 1.07 $\pm$ 0.06 & 2.53 $\pm$ 0.14 & 4.44 $\pm$ 0.16 & 10.40 $\pm$ 0.23\\
 & $\tau^*$ & 0.67 $\pm$ 0.03 & 0.69 $\pm$ 0.02 & 0.63 $\pm$ 0.01 & 0.63 $\pm$ 0.01 & 0.62 $\pm$ 0.00\\
 \midrule
\multirow{2}{*}{$K$=10}
 & $\Delta$ & 0.65 $\pm$ 0.06 & 0.72 $\pm$ 0.06 & 1.47 $\pm$ 0.09 & 2.52 $\pm$ 0.11 & 6.64 $\pm$ 0.14\\
 & $\tau^*$ & 0.81 $\pm$ 0.02 & 0.84 $\pm$ 0.01 & 0.81 $\pm$ 0.01 & 0.80 $\pm$ 0.01 & 0.78 $\pm$ 0.00\\
\midrule
\multirow{2}{*}{$K$=20} 
 & $\Delta$ & 0.43 $\pm$ 0.04 & 0.38 $\pm$ 0.04 & 0.73 $\pm$ 0.04 & 1.47 $\pm$ 0.06 & 4.01 $\pm$ 0.09\\
 & $\tau^*$ & 0.88 $\pm$ 0.01 & 0.93 $\pm$ 0.01 & 0.92 $\pm$ 0.00 & 0.90 $\pm$ 0.00 & 0.88 $\pm$ 0.00\\
\bottomrule
\end{tabular}
\caption{The average Hamming distance $\Delta$ and average rank correlation $\tau^*$ for varying $K$ and $p$. Each value represents $\mathrm{mean}\pm 1$ standard error.}
\label{tab:performance}
\end{table}
Table~\ref{tab:performance} reports the results of applying the proposed algorithm for different number of variables $p \in \{5, 10, 20, 40, 100\}$ and different number of sources $K \in \{1, 5, 10 ,20\}$. For each setting, we repeat the experiment 50 times.
Although we fix the sample size at $n_k = 1000$ in the simulations, the performance is very similar to that obtained with a smaller sample size of $n_k = 200$ (see Section~\ref{subsec:performance_supp} in the supplement).
Notably, the proposed algorithm can handle the setting with $K = 20$, whereas existing methods typically consider $K < 10$ ~\citep{Wang2020-of, Li2024-af, lee2022bayesian}. Moreover, it scales to $p = 100$, where the search space has a size of $100! \approx 9 \times 10^{157}$. The wall time for $p=100$ and $K =20$ was 2.5 hours for 10,000 iterations.
In the single DAG estimation case ($K = 1$), the values of $\tau^*$ are close to 0, indicating weak or no association between the estimated and true orderings. This implies that the algorithm is unable to identify the direction of edges among the Markov equivalent DAGs, which is also reflected in the large values of $\Delta$. As the number of sources $K$ increases,  we observe that $\Delta$ decreases uniformly, and $\tau^*$ gets larger and approaches 1.  
Importantly, we emphasize that the performance improvement observed for larger $K$ is not attributable to an increased total sample size. To demonstrate this, we conduct additional simulations in which the total sample size is fixed; the results are deferred to Section~\ref{subsec:performance_supp} in the supplement.

\subsection{MCMC convergence analysis} \label{subsec:conv}

We investigate the mixing behavior of the Markov chain (Algorithm~\ref{alg:full}) with the R2R proposal neighborhood $ \cN_{\mathrm{R2R}}$. For each $K \in \{1, 5, 10, 20\}$, we fix $p = 40$ and the true ordering $\sigma^* = (1,\dots, p)$, generate $K$ random DAGs, and simulate the corresponding datasets as described above. We randomly select 50 initial orderings, run an MCMC chain from each initialization, and plot the resulting Markov chain trajectories in Figure~\ref{fig:traj_plot}. The figure shows that, under the R2R proposal neighborhood, all trajectories reach orderings that achieve the same posterior score as those consistent with the true graphs. 
\begin{figure}[b!]
       \centering
       \footnotesize
        \captionsetup{font={footnotesize,stretch=1}}
        \includegraphics[width=0.18\linewidth]{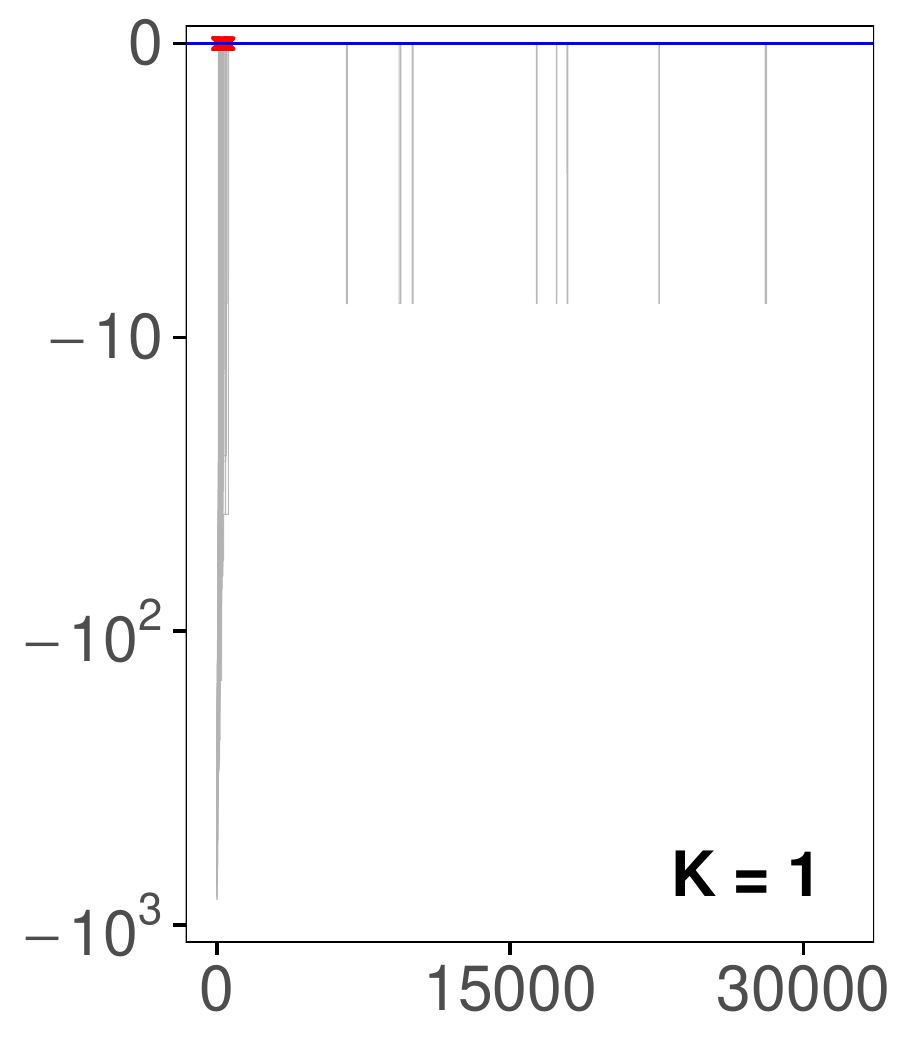}
        \includegraphics[width=0.18\linewidth]{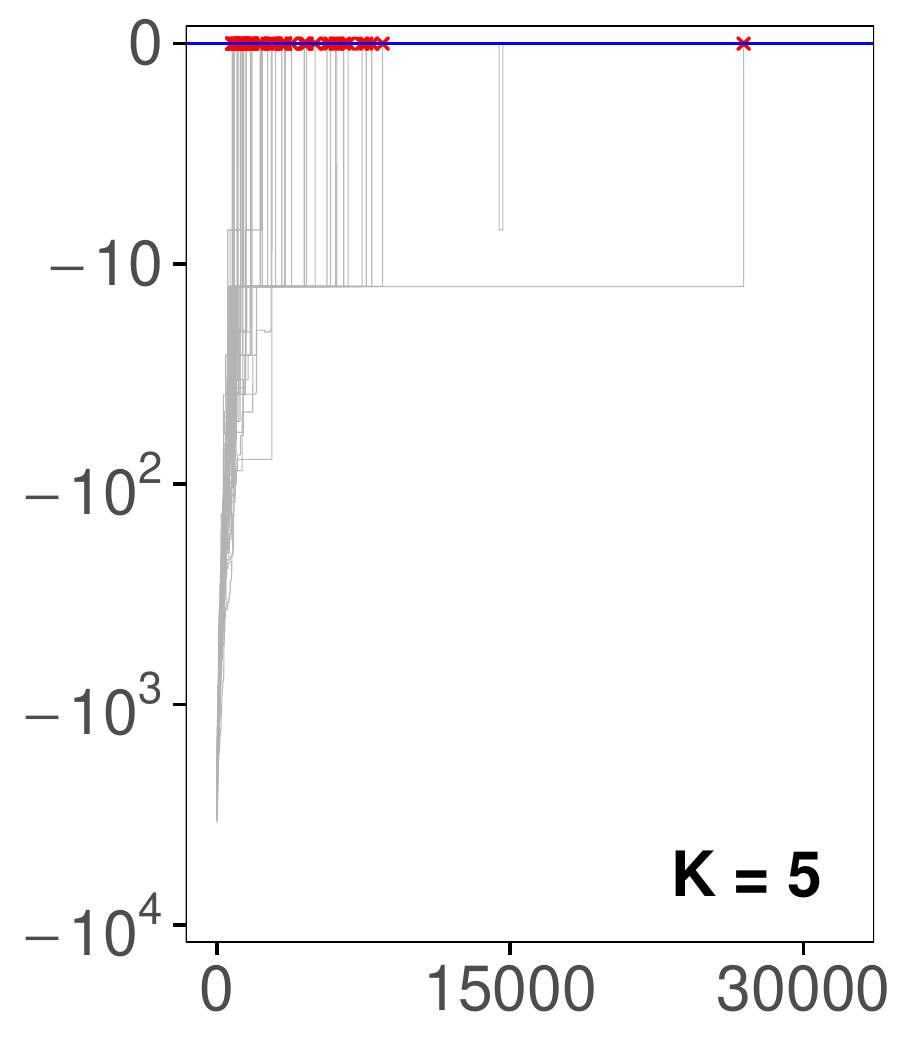}
        \includegraphics[width=0.18\linewidth]{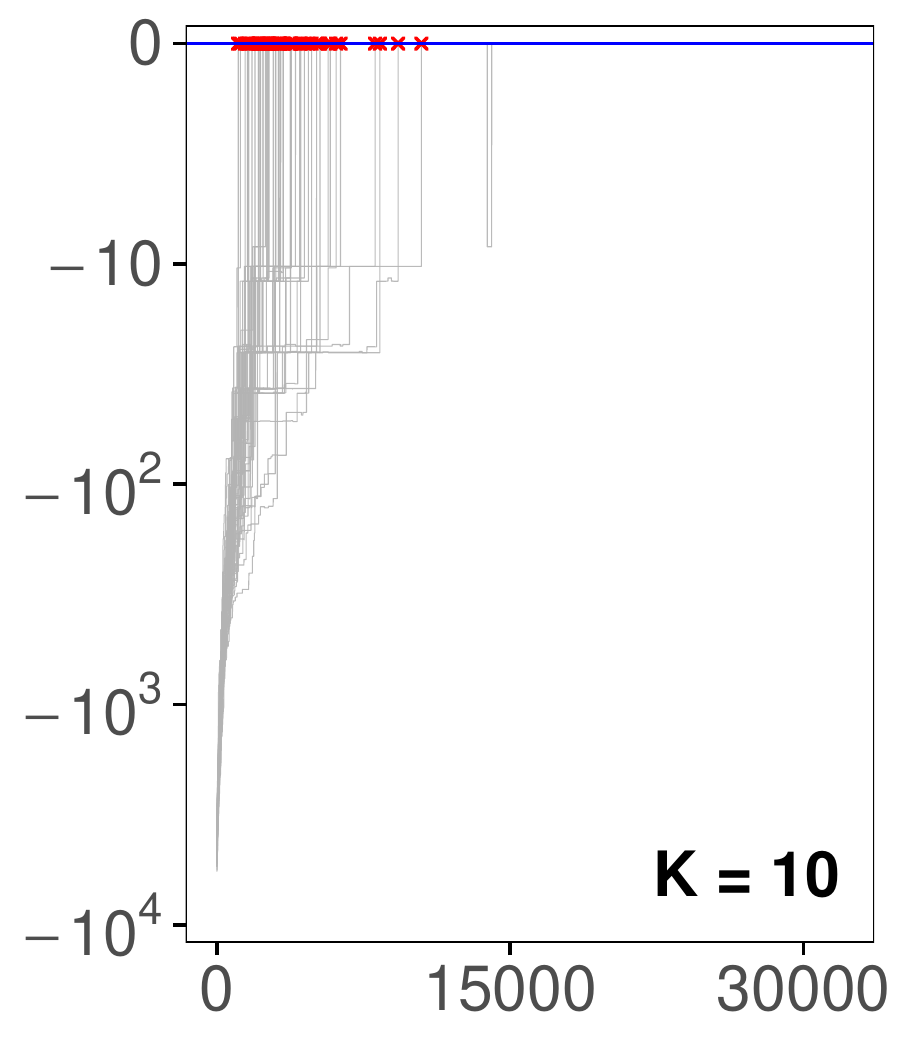}
        \includegraphics[width=0.18\linewidth]{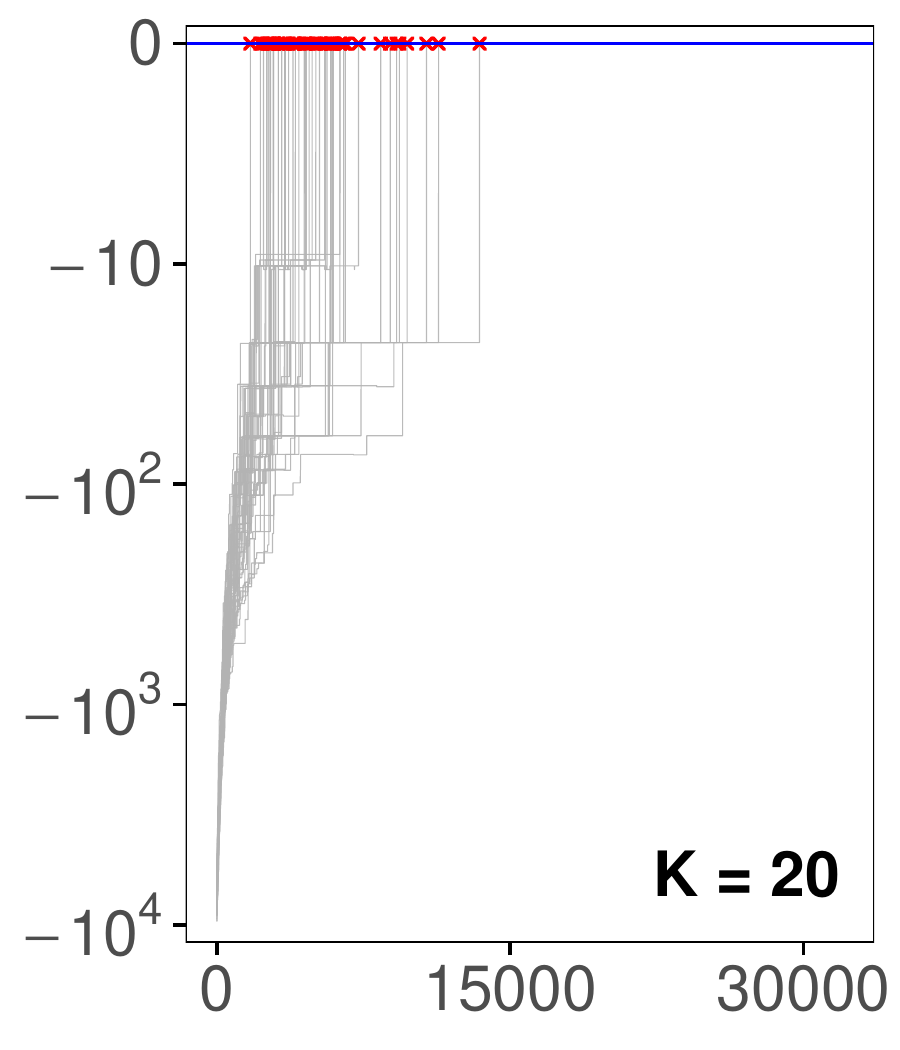}
        \includegraphics[width=0.18\linewidth]{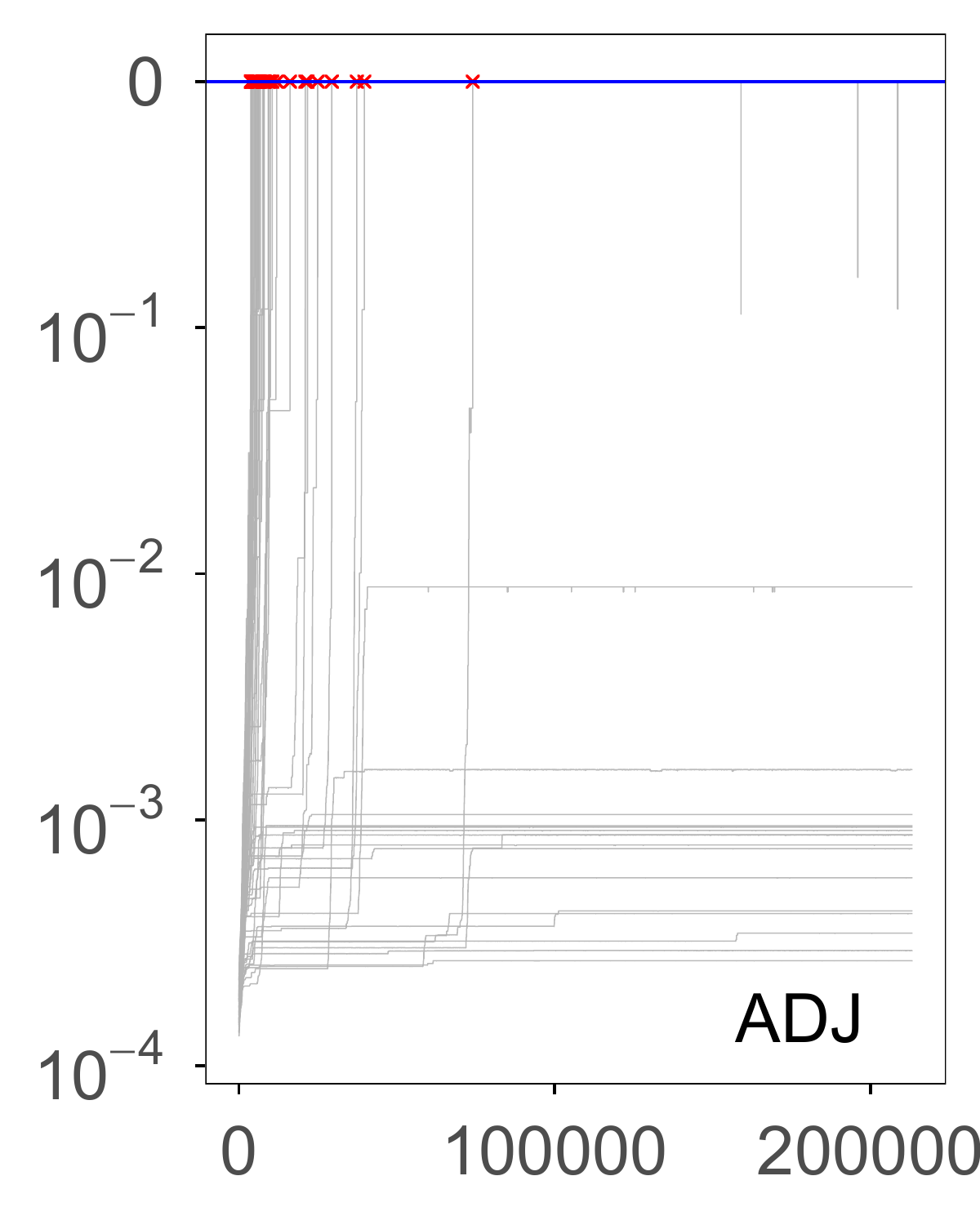}
        \caption{Scaled log posterior probability versus the number of iterations of 50 MCMC runs with random initialization. The first four plots correspond to \(p = 40\) with \(K = 1, 5, 10, 20\) using $\cN_{\mathrm{R2R}}$, and the last corresponds to \(p = 40, K = 20\) using $\cN_{\mathrm{ADJ}}$. The blue line represents the scaled log-posterior probability value of the true ordering. Each red cross marks the first time the chain attains the value.}
    \label{fig:traj_plot}
\end{figure}
At first glance, even the case $K=1$ appears to exhibit good mixing in terms of the ordering score. However, traversing among Markov equivalent orderings does not resolve edge directions, leading to instability in graph estimation. To examine this more closely, we compute the Gelman–Rubin (GR) statistic $\widehat{R}^{(k)}_{(i, j)}$ for each ordered node pair $(i,j)$ in each dataset $k$ to assess the convergence of edgewise marginal inclusion probabilities. This results in $Kp(p-1)$ GR statistics, which are summarized in Table~\ref{tab:GR}. We report the proportion of cases with $\widehat{R} < 1.1$, a commonly used threshold for good convergence, the proportion of cases with $\widehat{R} < 1.001$, and the maximum $\widehat{R}$. The maximum value of $\widehat{R} = 6.97$ in the $K=1$ case indicates poor mixing for certain edges, as expected due to the fundamental identifiability issue. As $K$ increases, the mixing behavior improves substantially, with the maximum $\widehat{R}$ decreasing markedly. These results confirm that increased data heterogeneity enhances identifiability, and that this improvement is effectively captured by the proposed computational scheme.
To highlight the efficacy of $\cN_{\mathrm{R2R}}$, we plot chain trajectories in Figure~\ref{fig:traj_plot} using the adjacent transposition neighborhood $\mathcal{N}_{\mathrm{ADJ}}$, which is frequently used in the literature~\citep{friedman2003being, agrawal2018minimal, Li2024-af}. This neighborhood is defined as $\mathcal{N}_{\mathrm{ADJ}}(\sigma)  = \{\sigma' \in \mathbb{S}^p \mid \sigma' = \mathrm{ADJ}(\sigma, i),\; i \in [p-1]\}$, where $\mathrm{ADJ}(\sigma, i)$ generates a new permutation by swapping the $i$-th and $(i+1)$-th elements of $\sigma$,
\begin{align*}
    \mathrm{ADJ}(\sigma, i):
(\sigma(1), \dots, \sigma(i), \sigma(i+1), \dots, \sigma(p))
\mapsto
(\sigma(1), \dots, \sigma(i+1), \sigma(i), \dots, \sigma(p)).
\end{align*}
Since $\mathcal{N}_{\mathrm{ADJ}}$ is smaller than $\mathcal{N}_{\mathrm{R2R}}$, we adjust the number of iterations to 213,333 to account for different per-iteration computational cost~\citep{Chang2022-id}. Notably, the adjacent transposition neighborhood exhibits suboptimal convergence, with 16 chains becoming trapped in local modes. Additional comparisons are provided in Section~\ref{subsec:mixing_other} in the supplement.
\begin{table}[t!]
\centering
%\footnotesize
\captionsetup{font={footnotesize,stretch=1}}
\renewcommand{\arraystretch}{0.8}
\begin{tabular}{lcccc}
\toprule
\textbf{Setting} & \textbf{K = 1} & \textbf{K = 5} & \textbf{K = 10} & \textbf{K = 20} \\
\midrule
$\%$ of \(\widehat{R}\) $<$ 1.001   & 98.69 & 99.44 & 99.64 & 99.81 \\
$\%$ of \(\widehat{R}\) $<$ 1.1 & 99.50 & 99.96 & 100 & 100  \\
Maximum \(\widehat{R}\) & 6.97 & 1.78 & 1.06 & 1.05  \\
%$\%$ of \(\widehat{R}\) = NA & 0.00 & 0.00 & 0.00 & 0.00  \\
\bottomrule
\end{tabular}
\caption{Gelman-Rubin statistic \(\widehat{R}\) in different settings for \(p = 40\) and varying \(K\).}
\label{tab:GR}
\end{table}

\subsection{Comparison with other methods}\label{subsec:comparison}

We compare the proposed method with four competing algorithms, PC, GES~\citep{chickering2002learning}, jointGES~\citep{Wang2020-of}, and muSuSiE-DAG~\citep{Li2024-af}. Since PC and GES are not designed for joint estimation across multiple datasets, we apply them separately to each dataset. We follow the simulation settings of~\citet{Li2024-af}. We set $K = 5$, $p = 100$, and $n_k =240$ for $k\in [5]$. Given the true ordering $\sigma^* = (1, 2, \dots, p)$, we randomly generate a common edge set $\mathcal{E}_{com}$ such that each edge in $\mathcal{E}_{com}$ is shared across all $K$ DAGs. For each $k \in [5]$, we generate a private edge set $\mathcal{E}_{pri}^{(k)}$ which does not overlap with $\mathcal{E}_{com}$ and consists of edges that appear only in the $k$-th DAG. We set $|\mathcal{E}_{pri}^{(k)}|$ to be the same for all $k$. To generate datasets from $K$ DAGs, for each DAG $G$, we draw edge weight for each edge $i \to j \in G$ independently from uniform distribution on $[-1, -0.1] \cup [0.1, 1]$, and set error variances to one. 
\begin{table}[b!]
\centering
\footnotesize
\captionsetup{font={footnotesize,stretch=1}}
\renewcommand{\arraystretch}{1.0}
%\small
\begin{tabular}{cccccc}
\toprule
 & \textbf{PC} & \textbf{GES} & \textbf{jointGES} & \textbf{muSuSiE-DAG} & \textbf{Proposed} \\
\midrule
$\Delta$ & 51.392 $\pm$ 0.767 & 26.172 $\pm$ 1.199 & 35.392 $\pm$ 1.384 & 71.612 $\pm$ 6.076 & 7.533 $\pm$ 0.304\\
TPR & 0.799 $\pm$ 0.004 & 0.927 $\pm$ 0.003 & 0.949 $\pm$ 0.004 & 0.853 $\pm$ 0.012 & 0.979  $\pm$ 0.001 \\
%FPR & 0.002 $\pm$ 0.000 & 0.002 $\pm$ 0.000 & 0.003 $\pm$ 0.000 & 0.005 $\pm$ 0.000 & 0.000 $\pm$ 0.000 \\
FDR & 0.150 $\pm$ 0.002 & 0.097 $\pm$ 0.005 & 0.162 $\pm$ 0.005 & 0.265 $\pm$ 0.020 & 0.029 $\pm$ 0.001 \\
\bottomrule
\end{tabular}
\caption{Comparison of methods for $K = 5$, $p = 100$, and $n_k = 240$ for $k \in [5]$, under the setting $|\mathcal{E}_{com}| = 100$, $|\mathcal{E}^{(k)}_{pri}| = 50$. 
$\Delta$ is the average Hamming distance, defined as in~\eqref{eq:delta}.
Each value represents $\mathrm{mean}\pm 1$ standard error.}
\label{tab:comparison}
\end{table}
We consider three different settings, where $(|\mathcal{E}_{com}|, |\mathcal{E}_{pri}^{(k)}|)$ is set to 
(50, 50), (100, 20), or  (100, 50), and repeat each setting for each algorithm 50 times. We present results for $|\mathcal{E}_{com}| = 100$ and $|\mathcal{E}_{pri}^{(k)}| = 50$ in Table~\ref{tab:comparison}. After tuning the hyperparameter, we select $c_0 = 7$. We defer the remaining cases and hyperparameter choices for the other methods to Section~\ref{subsec:compare_supp} in the supplement. Table~\ref{tab:comparison} shows that the proposed method outperforms the competing methods. 
% For each method, we calculate the edge difference between the true and estimated DAGs ($\Delta$), the average true positive rate (TPR), the average false positive rate (FPR), and the average false discovery rate (FDR). 
Specifically, the proposed method yields the smallest $\Delta$, demonstrating significantly higher accuracy in recovering the true DAGs.  Its lower false discovery rate implies that it produces far fewer spurious edge discoveries. Notably, the standard errors are consistently small, suggesting stable performance across repetitions.

\subsection{When the common ordering assumption fails}\label{subsec:common_ordering}

We conduct simulation studies to evaluate the impact of violating the common ordering assumption. We generate $K = 20$ datasets in which each dataset $X^{(k)}$ is consistent with a different underlying ordering $\sigma^{(k)}$. For a sensitivity analysis, we define the degree of similarity among those orderings as follows. We first fix a reference ordering $\sigma^* = (1, \dots, p)$. For a given value $\xi \in (0,1)$,  we randomly sample $K$ orderings $\{\sigma^{(k)}\}_{k=1}^{K}$
such that each rank correlation $\tau_{\mathrm{Kendall}}(\sigma^{(k)}, \sigma^*)$ is approximately $\xi$~\citep{diallo2020permutation}. 
For the sampled orderings, we compute their average pairwise rank correlation, denoted by 
\begin{align*}
    U \;=\; \frac{2}{K(K-1)} \sum_{1 \le i < j \le K} \tau_{\mathrm{Kendall}}\!\left(\sigma^{(i)}, \sigma^{(j)}\right),
\end{align*}
which implies that larger values of $U$ correspond to greater similarity among the sampled orderings. We report results corresponding to $U \in \{0, 0.3, 0.5, 0.7, 0.9, 1\}$. In particular, $U = 1$ corresponds to the case where the common ordering assumption holds with the true ordering $\sigma^*$. Given a value $U$, for each sampled ordering $\sigma^{(k)}$, we generate a random DAG $G^{(k)}$ and simulate the corresponding dataset $X^{(k)}$ as previously defined. The experiment is repeated 50 times, and our method is applied to each simulated dataset. The result is displayed in Figure~\ref{fig:violation_common_order}. When the average pairwise rank correlation is moderate ($U \geq 0.7$), joint estimation outperforms individual estimation, whereas individual estimation is more accurate than joint estimation at $U \leq 0.5$. 
Moreover, individual estimation exhibits noticeably greater variability than joint estimation when $U < 0.5$, indicating improved stability and more reliable convergence under joint modeling.

\begin{figure}[t!]
\centering
\captionsetup{font={footnotesize,stretch=1}}
    \begin{minipage}{0.45\textwidth}  % Adjust width as needed
            \includegraphics[width=\linewidth]{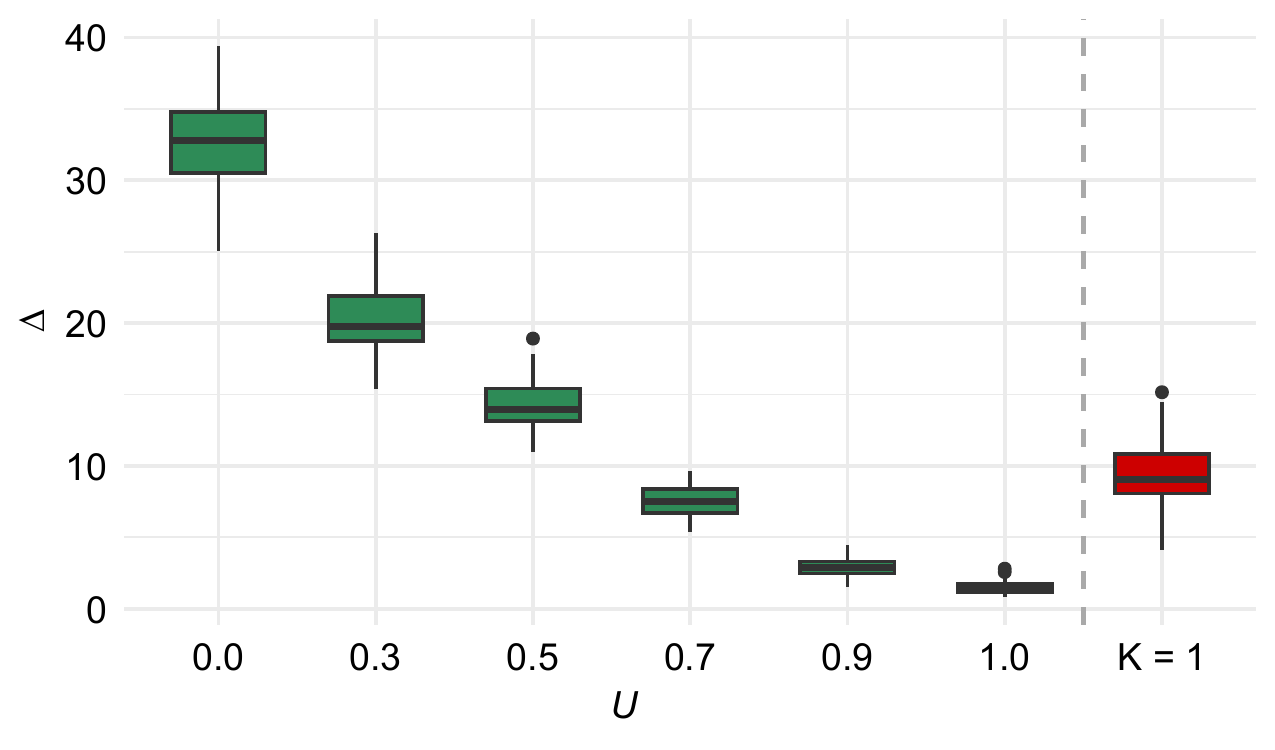}  % Replace with your figure file
    \end{minipage}%
    \hspace{0.2cm}
    \begin{minipage}{0.35\textwidth}  % Adjust width to control spacing
        \captionof{figure}{The green boxplots summarize the average Hamming distance $\Delta$ between the true DAGs and their estimates using Algorithm~\ref{alg:full} across different values of $U$ based on 50 experiments. The red boxplot corresponds to the case for $K = 1$.}
        \label{fig:violation_common_order}
    \end{minipage}
\end{figure}

\section{Real data analysis}

We apply the proposed method to analyze a single-nucleus RNA-seq database of patients with major depressive disorder (MDD)~\citep{nagy2020single}. We follow standard preprocessing procedures, followed by normalization of log-transformed expression levels. We removed one control dataset due to insufficient sample size. The  dataset consists of 33 data matrices, including 17 MDD case datasets and 16 healthy control datasets, with varying sample sizes ranging from 1,061 to 3,757 (see Table~\ref{tab:samplesize} in Section~\ref{sec:real_supp} in the supplement).
We only consider 41 genes reported in \cite{nagy2020single} that remained significant (FDR $<$ 0.1) after FDR correction across all cell clusters  following cell-type–specific differential expression analysis. 
We run the proposed  MCMC with $K = 33$ datasets for 30,000 iterations and then discard the first 15,000 iterations as burn-in. The other hyperparameters are set identically to those in Section~\ref{sec:simulation}.
 We run 50 independent chains to assess convergence of the Markov chains. Following the procedure described in Section~\ref{subsec:conv}, mixing behavior is evaluated using chain trajectory plots (the left panel of Figure~\ref{fig:edge_number and traj_plot}) and the GR statistic \(\widehat{R}^{(k)}_{(i,j)}\) for each pair \((i,j)\) in each dataset \(k\). The results demonstrate good mixing of the chains. Remarkably, all GR statistics are below 1.1, with a maximum value of 1.095. 
\begin{figure}[t!]
       \centering
       %\footnotesize
       \captionsetup{font={footnotesize,stretch=1}}
        \includegraphics[width=0.49\linewidth]{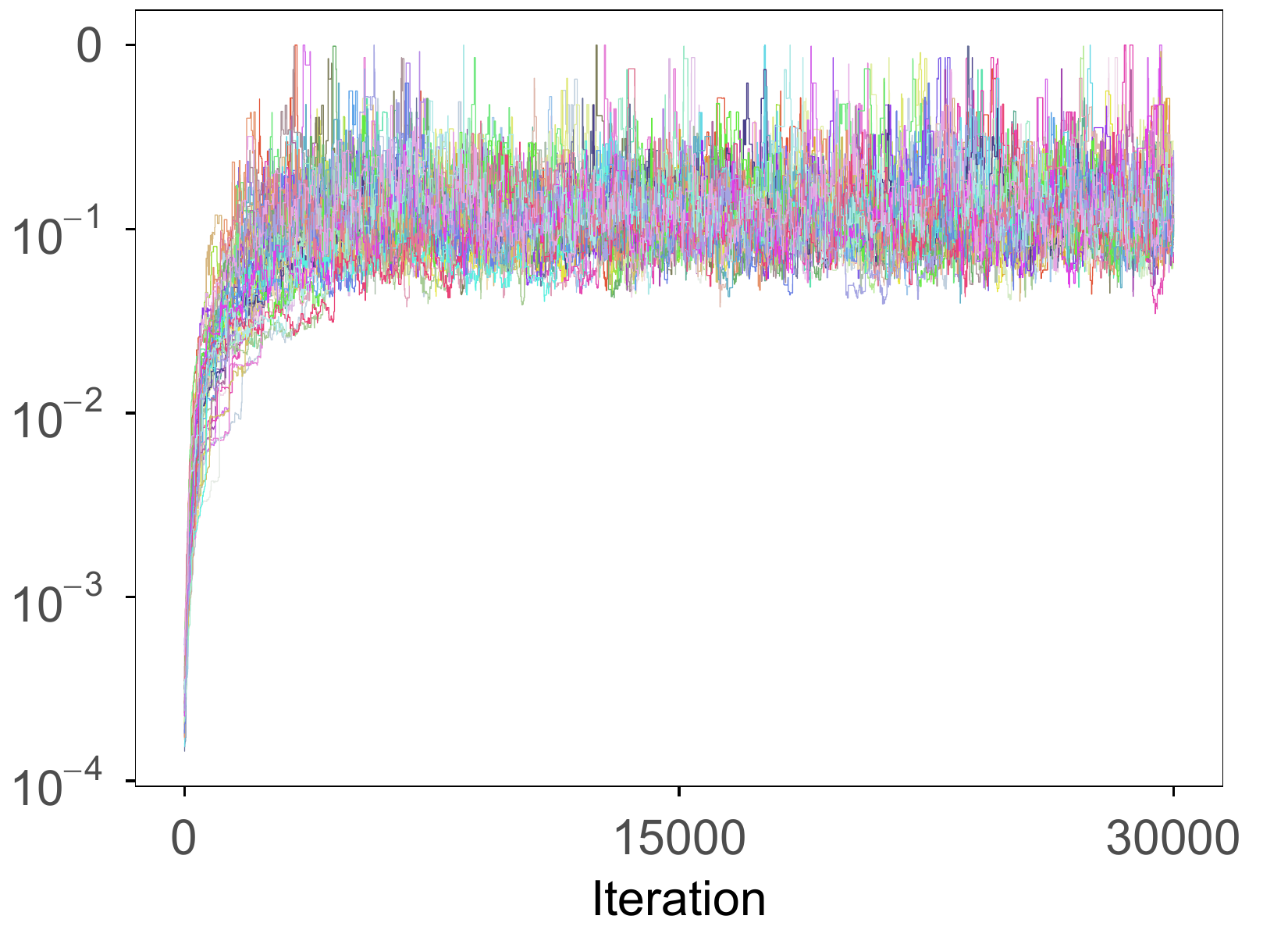}
        \includegraphics[width=0.49\linewidth]{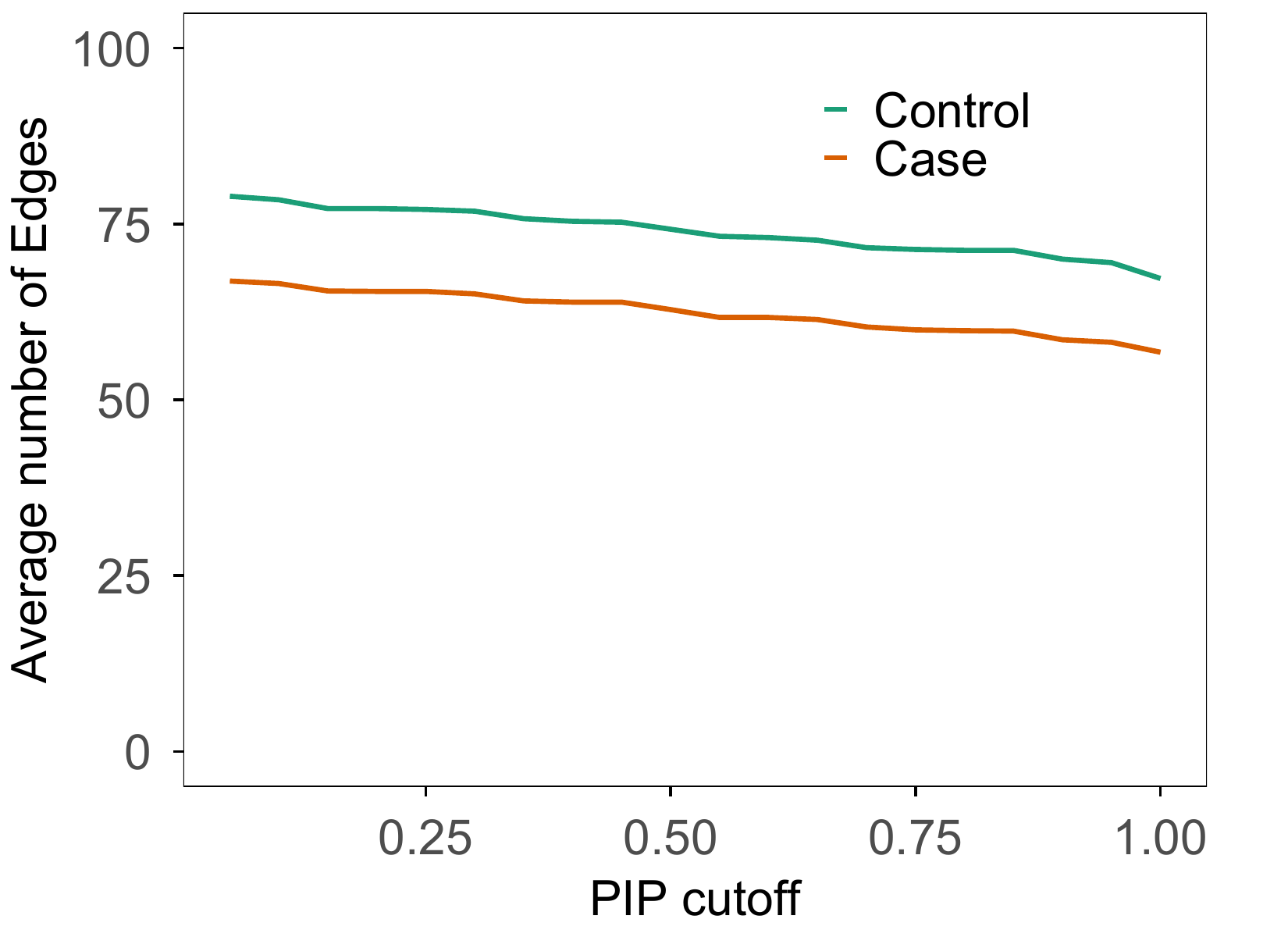}

        \caption{(Left) Scaled log posterior probability versus the number of iterations in 50 MCMC runs with random initialization. 
        (Right) With varying threshold $c \in (0,1)$, the green and orange lines show the average number of edges in the control and case groups, respectively.}
    \label{fig:edge_number and traj_plot}
\end{figure}
After obtaining the posterior mean DAGs, we construct point-estimate DAGs by including an edge whenever its posterior edge inclusion probability exceeds a pre-specified threshold $c \in (0,1)$. 
To assess group-level differences, we compute and plot the average number of edges in the case and control groups across a range of thresholds, as shown in the right panel of Figure \ref{fig:edge_number and traj_plot}. 
Across all thresholds, the control group consistently exhibits more edges than the case group, suggesting reduced network connectivity in the case group. 
To identify which genes drive the observed group-level connectivity differences, we summarize node-level connectivity using the posterior mean values for each dataset. For each dataset and each node, we compute its total connectivity by summing the posterior probabilities of outgoing and incoming edges. We then compute the average connectivity within the case and control groups for each node, and rank nodes by the magnitude of their group differences. Table~\ref{tab:top_genes_difference} reports the top 9 genes ranked by group-level connectivity difference. The top-ranked gene, \texttt{GRIN2A}, has been reported to be associated with psychiatric disorders, including mood-related conditions~\citep{shepard2024differential}. For the remaining 32 genes, the absolute differences were all less than 1.
\begin{table}[h]
\centering
\scriptsize
\captionsetup{font={footnotesize,stretch=1}}
\begin{tabular}{lccccccccc}
\hline
Rank 
& 1 & 2 & 3 & 4 & 5 & 6 & 7 & 8 & 9  \\
\hline
Gene & GRIN2A & CCND1 & PRKAR1B & PEBP1 & HSP90AA1 & NRXN2 
& ALDOA & TMSB4X & TOMIL1 \\
\hline
Diff & 3.0427 & 2.3583 & 2.0583 & 1.8204 & 1.3625 & 1.2005 
& 1.1565 & 1.1564 & -1.0003  \\
\hline
\end{tabular}
\caption{Top 9 genes ranked by group-level connectivity difference. Diff indicates the group average difference (Control $-$ Case).}
\label{tab:top_genes_difference}
\end{table}
\vspace{-5mm}
\section*{Data availability statement}
The single-nucleus RNA sequencing dataset from patients with major depressive disorder (MDD) is publicly available in the Gene Expression Omnibus (GEO) under accession number GSE144136.

%\vspace{-3mm}
\section*{Disclosure Statement}

The authors report there are no competing interests to declare. The authors used ChatGPT (OpenAI, GPT-5.2) for language editing and grammatical refinement. 

%\newpage
\bibliographystyle{plainnat}
\bibliography{reference}

\renewcommand{\thesection}{\Alph{section}}
\setcounter{section}{0}
\renewcommand{\thetheorem}{\thesection\arabic{theorem}}
\setcounter{theorem}{0}
\renewcommand{\thelemma}{\thesection\arabic{lemma}}
\setcounter{lemma}{0}
\renewcommand{\thecorollary}{\thesection\arabic{corollary}}
\setcounter{corollary}{0}
\renewcommand{\thedefinition}{\thesection\arabic{definition}}
\setcounter{definition}{0}
\renewcommand{\theexample}{\thesection\arabic{example}}
\setcounter{example}{0}
\renewcommand{\thefigure}{\thesection\arabic{figure}}
\setcounter{figure}{0}
\renewcommand{\thetable}{\thesection\arabic{table}}
\setcounter{table}{0}
\renewcommand{\theremark}{\thesection\arabic{remark}}
\setcounter{remark}{0}

\setcounter{section}{0}
\setcounter{subsection}{0}
\setcounter{figure}{0}
\setcounter{table}{0}
\setcounter{equation}{0}
\renewcommand{\thesection}{\Alph{section}}
\renewcommand{\thesubsection}{\thesection\arabic{subsection}}
\renewcommand{\thefigure}{\thesection\arabic{figure}}
\renewcommand{\thetable}{\thesection\arabic{table}}
\renewcommand{\theequation}{\thesection\arabic{equation}}

%\renewcommand{\thealg}{\thesection\arabic{alg}}
%\setcounter{alg}{0}

%\setcounter{section}{0}
%\renewcommand\thesection{\Alph{section}}
%\renewcommand\thesubsection{\Alph{section}.\arabic{subsection}}

%\setcounter{equation}{16}
%\counterwithin{algocf}{section}
%\renewcommand{\thealgocf}{\Alph{section}.\arabic{algocf}}
%\counterwithin{figure}{section}
%\renewcommand{\thefigure}{\Alph{section}.\arabic{figure}}
%\counterwithin{table}{section}
%\renewcommand{\thetable}{\Alph{section}.\arabic{table}}
\setcounter{algocf}{1}
\newpage
\begin{center}
\LARGE{Supplementary material}
\end{center}
\section{Proofs}

\subsection{Preliminary definitions and lemmas}\label{subsec:notation_lemmas}

In this section, we review standard concepts in DAG models that will be used throughout the proofs.

\noindent
\textbf{Characterization of the Markov equivalence class.}
Recall that the Markov equivalence class $\mathcal{E}(G)$ of a DAG $G$ is the set of all DAGs that are Markov equivalent to $G$.  It is well known that Markov equivalence admits the following characterization.
\begin{lemma}[{\citealp{verma1991equivalence}}]
Two DAGs are Markov equivalent if and only if they have the same skeleton and the same set of v-structures.
\end{lemma}
Here, the skeleton of a DAG is the unique undirected graph obtained by
replacing all edges in the DAG with undirected ones. A v-structure in a DAG is a triple of distinct vertices
\(i, j, k\) such that
\(i \to j \leftarrow k\) and there is no edge between \(i\) and \(k\). 
Recall that an edge is called an essential arrow if its orientation is invariant across the entire Markov equivalence class.
By the characterization above, any edge that participates in a v-structure is essential. 
The remaining essential arrows are determined by orientation rules that preserve existing v-structures while avoiding the creation of new ones; see Section~3 of~\cite{andersson1997characterization} for a complete characterization.
Essential arrows capture the orientations that are invariant within a Markov equivalence class. 
These invariant orientations admit a canonical graphical representation of the equivalence class, called the completed partially directed acyclic graph (CPDAG).
A CPDAG is a chain graph, i.e., a graph that may contain both directed and undirected edges but has no directed cycles.
It is defined as follows.
\begin{definition}[CPDAG]
Let $\mathcal{E}$ be a Markov equivalence class of DAGs. 
The CPDAG associated with $\mathcal{E}$ is the chain graph that  
(i) has the same skeleton as any $G \in \mathcal{E}$, and  
(ii) contains a directed edge $i \to j$ if and only if it is an essential arrow in $\cE$.
\end{definition}
Any Markov equivalence class $\mathcal{E}$ admits a unique CPDAG representation, and conversely, each CPDAG corresponds to a unique equivalence class. In particular, directed edges in the CPDAG correspond precisely to essential arrows in $\mathcal{E}$.
Using the CPDAG representation, we justify the inclusion $\cL^{\cup}(\cE(G)) \subseteq \cL^{e}(E^*_{G})$, which is used in Lemma~\ref{lemma:essential}. Indeed, any ordering that is consistent with some $G' \in \cE(G)$ must respect all essential arrows of $\cE(G)$, since these orientations are invariant across the equivalence class. 
Hence such orderings belong to $\cL^{e}(E^*_G)$.
\begin{figure}[!b]
    \centering
    \captionsetup{font={footnotesize,stretch=1}}
    \includegraphics[width=1\linewidth]{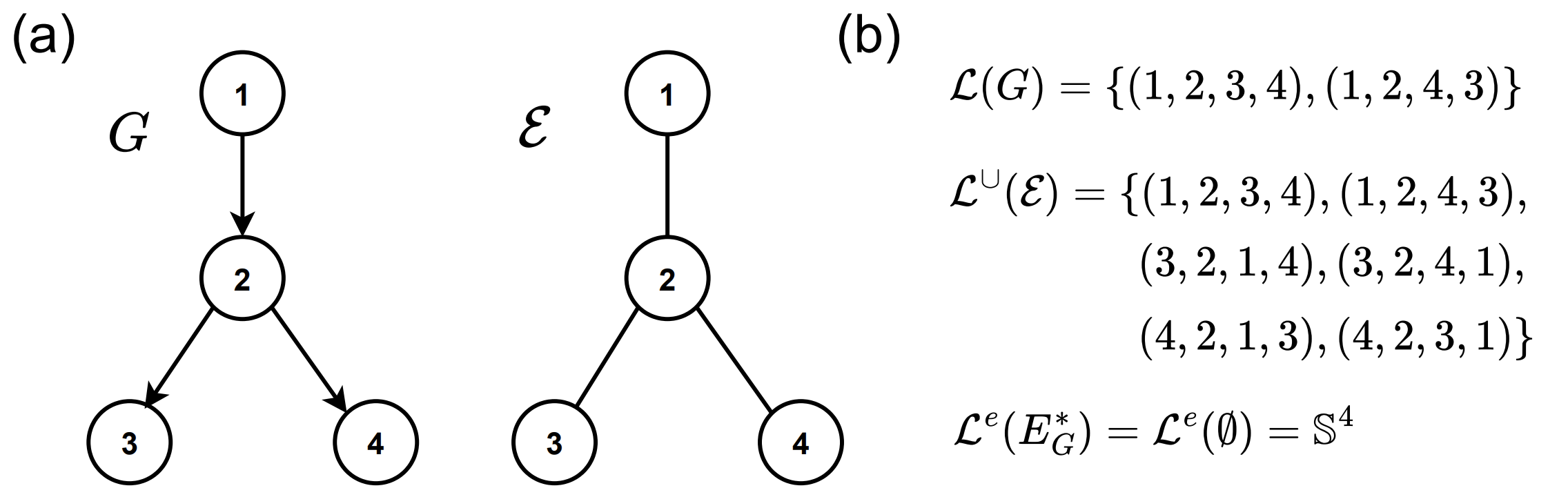}
    \caption{(a) $\cE$ is the CPDAG of $G$. Since $G$ has no essential arrows, all edges are undirected in $\cE$. (b) The corresponding ordering sets.}
    \label{fig:order_set}
\end{figure}
Figure~\ref{fig:order_set} illustrates that the inclusion is strict in general. In the example, $\cL(G)$ contains only two orderings, whereas 
$\cL^{\cup}(\cE(G))$ contains six orderings, corresponding to the three Markov equivalent DAGs
\begin{align*}
    G_1 = \{3\to 2,\, 2 \to 1,\, 2 \to 4 \}, \quad
G_2 = \{4\to 2,\, 2 \to 1,\, 2 \to 3\}, \quad
G.
\end{align*}
Since the equivalence class has no essential arrows, 
$\cL^{e}(E^*_G)$ contains all possible orderings, and therefore the inclusion is strict.

\noindent
\textbf{Covered edge and Chickering sequence.} An edge $i \to j$ in a DAG $G$ is called covered if $\Pa_j(G) = \Pa_i(G) \cup \{i\}.$
Covered edges characterize exactly those edge reversals that preserve Markov equivalence.
\begin{lemma}[{\citealp[Lemma 1]{chickering1995transformational}}]
\label{lemma:reversal_covered}
Let $i \to j$ be a covered edge in $G$, and let $G'$ be the DAG obtained by reversing this edge, that is,
\begin{align*}
    G' := \bigl(G \cup \{j \to i\}\bigr)\setminus \{i \to j\}.
\end{align*}
Then $G'$ is Markov equivalent to $G$.
\end{lemma}
Since a covered edge can be reversed without leaving the Markov equivalence class, it cannot be an essential arrow. The following lemma identifies an edge that is necessarily covered under a fixed ordering and hence cannot be essential.
\begin{lemma}\label{lemma:ordering-essential}
Let $\sigma \in \mathbb{S}^p$ and consider the set of DAGs $\mathcal{G}_\sigma^p$ consistent with $\sigma$.

(i) For any $G \in \mathcal{G}_\sigma^p$ such that 
$\sigma(1)\to\sigma(2) \in G$, 
the edge $\sigma(1)\to\sigma(2)$ is not an essential arrow.

(ii) For every edge $e \neq \sigma(1)\to\sigma(2)$ 
with $e=\sigma(i)\to\sigma(j)$ and $i<j$, 
there exists a DAG $G \in \mathcal{G}_\sigma^p$ 
in which $e$ is an essential arrow.
\end{lemma}
\begin{proof}
(i) Let $G \in \mathcal{G}_\sigma^p$ such that $\sigma(1)\to\sigma(2) \in G$. Recall that $P_j^\sigma$ is the set of predecessors of $j$ under the ordering $\sigma$, defined in~\eqref{eq:potential}.
Since $G$ is consistent with $\sigma$, $P_{\sigma(1)}^\sigma = \emptyset$, so $\sigma(1)$ has no parents and 
$\Pa_{\sigma(2)}(G) \subseteq P_{\sigma(2)}^\sigma = \{\sigma(1)\}$.
Hence, whenever the edge $\sigma(1)\to\sigma(2)$ appears, we necessarily have 
\begin{align*}
    \Pa_{\sigma(2)}(G) = \Pa_{\sigma(1)}(G) \cup \{\sigma(1)\},
\end{align*}
so the edge is covered.
By Lemma~\ref{lemma:reversal_covered}, a covered edge can be reversed
without leaving the Markov equivalence class,
and therefore it is not essential.

(ii) Fix an edge $e=\sigma(i)\to\sigma(j)$ with $i<j$
and $e \neq \sigma(1)\to\sigma(2)$.
Then $j \ge 3$, so there exists some $k<j$ with $k\neq i$.
Construct a DAG $G \in \mathcal{G}_\sigma^p$ containing the
v-structure
\[
\sigma(i)\to\sigma(j)\leftarrow\sigma(k),
\]
and no edge between $\sigma(i)$ and $\sigma(k)$.
Such a graph is consistent with $\sigma$.
Since any edge participating in a v-structure is essential,
the edge $e$ is essential in this $G$.
\end{proof}

Covered edges play a fundamental role in DAG learning, as they are the only edge reversals that preserve Markov equivalence. 
More generally, together with edge additions, covered edge reversals generate paths in the space of DAGs under the partial order induced by I-map inclusion. 
This property, known as Meek’s conjecture and proved by \citet{chickering2002optimal}, is stated below.

\begin{theorem}[Meek's conjecture]\label{thm:meek}
Let $G$ and $G'$ be two DAGs such that $G'$ is an I-map of $G$.  
Then there exists a sequence of DAGs
\begin{align}\label{eq:chickering}
    G = G_0, G_1, \dots, G_L = G',
\end{align}
such that for each $\ell = 0, \dots, L-1$, 
$G_{\ell+1}$ is obtained from $G_\ell$ by either 
adding a single edge or reversing a covered edge, 
and remains an I-map of $G_\ell$.
\end{theorem}

We refer to the sequence in \eqref{eq:chickering} as a Chickering sequence.

\newpage

\subsection{Proofs in Section~\ref{subsec:order-based}}\label{subsec:proof_order-based}

\subsubsection{Proofs of Lemma~\ref{lemma:joint_score}}\label{subsub:lemma:joint_score}

\begin{proof}[Proof of Lemma~\ref{lemma:joint_score}]

(1) Suppose $\sigma \in \bigcap_{G \in \cG} \cL^{\cup}(\cE(G))$. Then for every $k \in [K]$, we have 
$\sigma \in \cL^{\cup}(\cE(G^{(k)}))$.
By Definition~\ref{def:score-eq}, this implies that 
$\psi^{(k)}(\sigma)$ attains its maximum value over $\bbS^p$ for each $k$. Since $F$ is strictly increasing in each coordinate, $\Psi(\sigma)$ is maximized over $\bbS^p$ whenever each $\psi^{(k)}(\sigma)$ is maximized over $\bbS^p$.
Thus $\sigma$ maximizes $\Psi$.

Now let $\tau \notin \bigcap_{G \in \cG} \cL^{\cup}(\cE(G))$. 
Then there exists at least one $k \in [K]$ such that 
$\tau \notin \cL^{\cup}(\cE(G^{(k)}))$.
Hence $\psi^{(k)}(\tau) < \psi^{(k)}(\sigma).$
Since $F$ is strictly increasing in each coordinate, it follows that $\Psi(\tau) < \Psi(\sigma).$
Therefore, $\sigma$ maximizes $\Psi$ over $\bbS^p$ if and only if $\sigma \in \bigcap_{G \in \cG} \cL^{\cup}(\cE(G)).$

(2) Fix $k \in [K]$. 
By part (1), the set of maximizers of $\Psi$ over $\bbS^p$ is $\bigcap_{G \in \cG} \cL^{\cup}(\cE(G)).$
Similarly, the set of maximizers of $\psi^{(k)}$ over $\bbS^p$ is $\cL^{\cup}(\cE(G^{(k)})).$
Since $\bigcap_{G \in \cG} \cL^{\cup}(\cE(G))
\subseteq 
\cL^{\cup}(\cE(G^{(k)})),$
we obtain
\begin{align*}
        \cM^{\Psi}(G^{(k)}) = \bigcap_{G\in \cG} \cL^{\cup}(\cE(G)) \setminus \cL (G^{(k)})\subseteq  \cL^{\cup}(\cE(G^{(k)})) \setminus \cL (G^{(k)}) =  \cM^{\psi^{(k)}}(G^{(k)}).
    \end{align*}
\end{proof}

\subsubsection{Proofs of Lemma~\ref{lemma:essential}}\label{subsub:lemma:essential}

\begin{proof}[Proof of Lemma~\ref{lemma:essential}]
For any DAG $G$, we first note that
 $\cL^{\cup}(\cE(G)) \subseteq \cL^{e}(E^*_{G})$.
 %See the example with Figure~\ref{fig:order_set} in Section~\ref{subsec:notation_lemmas}. 
 To see this, suppose $i \to j$ is an essential arrow of $G$. 
Since this edge appears in every DAG in $\cE(G)$, 
any ordering $\sigma$ that is consistent with a DAG in $\cE(G)$ must satisfy 
$\sigma^{-1}(i) < \sigma^{-1}(j)$; that is, node $i$ precedes node $j$ in $\sigma$. 
Therefore, $\sigma \in \cL^{e}(E_G^*)$. Taking intersections over $G \in \cG$ yields 
\begin{align*}
         \bigcap_{G\in \cG} \cL^{\cup}(\cE(G)) \subseteq \bigcap_{G\in \cG}\cL^{e}(   E^*_{G}).
    \end{align*}
It remains to show that
\begin{align*}
         \bigcap_{G \in \mathcal G} \mathcal L^{e}(E_G^*)
\;\subseteq\;
\mathcal L^{e}\!\left(\bigcup_{G \in \mathcal G} E_G^*\right).
\end{align*}
Let $\sigma \in \bigcap_{G \in \mathcal G} \mathcal L^{e}(E_G^*)$. Then $\sigma$ is consistent with all directed edges in $E_G^*$ for every $G \in \mathcal G$; that is, for each $(i,j) \in E_G^*$, $\sigma^{-1}(i) < \sigma^{-1}(j).$
Hence, $\sigma$ is consistent with every edge in $\bigcup_{G \in \mathcal G} E_G^*$, which implies
\begin{align*}
    \sigma \in \mathcal L^{e}\!\left(\bigcup_{G \in \mathcal G} E_G^*\right).
\end{align*}
This concludes the proof. In fact, equality holds:
\[
\bigcap_{G \in \mathcal G} \mathcal L^{e}(E_G^*)
\;=\;
\mathcal L^{e}\!\left(\bigcup_{G \in \mathcal G} E_G^*\right).
\]
By definition,
\begin{align*}
    \cL^{e}(E) =\{\sigma \in \bbS^p : \sigma^{-1}(i) < \sigma^{-1}(j) \text{ for all } (i,j) \in E\}.
\end{align*}
Let $\sigma \in \mathcal L^{e}(\bigcup_{G\in\mathcal G} E_G^*)$. Then, $\sigma^{-1}(i) < \sigma^{-1}(j) \text{ for all } (i,j) \in \bigcup_{G\in\mathcal G} E_G^*$. In particular, for each fixed $G \in \cG$ and each $(i,j) \in E_G^*$,
we have $\sigma^{-1}(i) < \sigma^{-1}(j)$,
so $\sigma \in \cL^{e}(E_G^*)$. Since $G$ was arbitrary, $\sigma \in \bigcap_{G \in \cG} \cL^{e}(E_G^*).$
\end{proof}

\subsubsection{Proofs of Theorem~\ref{thm:identifiability}}\label{subsub:thm1}

\begin{proof}[Proof of Theorem~\ref{thm:identifiability}]

First, observe that $\cL^e(E_{\max}) = \{\sigma^*, \sigma^\dagger\}$. To see this, $E_{\max}$ contains the chain constraints
$\sigma^*(2)\to\sigma^*(3)\to\cdots\to\sigma^*(p)$ and the additional constraint
$\sigma^*(1)\to\sigma^*(3)$. Hence any $\sigma\in\cL^{e}(E_{\max})$ must satisfy
$\sigma^{-1}(\sigma^*(j))<\sigma^{-1}(\sigma^*(j+1))$ for all $2\le j\le p-1$, so the relative order of
$\{\sigma^*(2),\ldots,\sigma^*(p)\}$ is fixed, and also $\sigma^*(1)$ must appear before $\sigma^*(3)$.
The only remaining freedom is the relative order between $\sigma^*(1)$ and $\sigma^*(2)$, which yields exactly the two orderings $\sigma^*$ and $\sigma^\dagger$.

We next show that $\sigma^*$ and $\sigma^\dagger$ are maximizers of $\Psi$. Fix $k\in[K]$. Since $\sigma^*$ is consistent with $G^{(k)}$ (common ordering assumption), we have
$\sigma^*\in \cL(G^{(k)}) \subseteq \cL^\cup(\cE(G^{(k)}))$, and thus $\psi^{(k)}(\sigma^*)$ attains its maximum value over $\bbS^p$ by score equivalence. It remains to show that $\psi^{(k)}(\sigma^\dagger)$ also attains its maximum. Note that $\sigma^\dagger$ differs from $\sigma^*$ only in the relative order of the two nodes $\sigma^*(1)$ and $\sigma^*(2)$. If there is no edge between $\sigma^*(1)$ and $\sigma^*(2)$ in $G^{(k)}$, then $G^{(k)}$ is also consistent with $\sigma^\dagger$, so $\sigma^\dagger \in \cL(G^{(k)}) \subseteq \cL^\cup(\cE(G^{(k)}))$.
Otherwise, since $\sigma^*$ is consistent with $G^{(k)}$, the only possible directed edge between them is $\sigma^*(1)\to\sigma^*(2)$.
By Lemma~\ref{lemma:ordering-essential} (i), the edge $\sigma^*(1)\to\sigma^*(2)$ is covered whenever it appears in a DAG consistent with $\sigma^*$. Therefore, by Lemma~\ref{lemma:reversal_covered}, there exists a Markov equivalent DAG $\widetilde G \in \cE(G^{(k)})$ in which this edge is oriented as $\sigma^*(2)\to\sigma^*(1)$.
Since all other pairwise order relations coincide for $\sigma^*$ and $\sigma^\dagger$, the DAG $\widetilde G$ is consistent with $\sigma^\dagger$, implying
$\sigma^\dagger \in \cL^\cup(\cE(G^{(k)}))$. In either case, score equivalence yields that $\psi^{(k)}(\sigma^\dagger)$ attains its maximum over $\bbS^p$.

Consequently, $(\psi^{(1)}(\sigma^*),\ldots,\psi^{(K)}(\sigma^*))$ and
$(\psi^{(1)}(\sigma^\dagger),\ldots,\psi^{(K)}(\sigma^\dagger))$
are coordinatewise maximizers. Since $F$ is strictly increasing in each coordinate, both $\sigma^*$ and $\sigma^\dagger$ maximize
\[
\Psi(\sigma)=F\bigl(\psi^{(1)}(\sigma),\ldots,\psi^{(K)}(\sigma)\bigr)
\quad\text{over } \bbS^p.
\]

Finally, we show that there are no other maximizers.
Let $\tau\notin \cL^{e}(E_{\max})$. Then there exists an edge $(i,j)\in E_{\max}$ such that
$\tau^{-1}(j)<\tau^{-1}(i)$.
By the assumption $E_{\max}\subseteq \bigcup_{k=1}^K E^*_{G^{(k)}}$, there exists $k'\in[K]$ such that
$(i,j)\in E^*_{G^{(k')}}$.
By definition of essential arrows, every DAG in $\cE(G^{(k')})$ contains the directed edge $i\to j$.
Hence no DAG in $\cE(G^{(k')})$ can be consistent with $\tau$, i.e.,
$\tau\notin \cL^\cup(\cE(G^{(k')}))$.
By score equivalence, this implies $\psi^{(k')}(\tau) < \max_{\sigma\in\bbS^p}\psi^{(k')}(\sigma)
= \psi^{(k')}(\sigma^*).$
Since $F$ is strictly increasing in each coordinate, the strict inequality in the $k'$th coordinate yields
$\Psi(\tau)<\Psi(\sigma^*)$.
Therefore, $\tau$ cannot maximize $\Psi$.
Combining the above and using $\cL^{e}(E_{\max})=\{\sigma^*,\sigma^\dagger\}$, we conclude $\arg\max_{\sigma\in\bbS^p}\Psi(\sigma)=\{\sigma^*,\sigma^\dagger\}.$
\end{proof}

\subsubsection{On the number of different datasets $K$}\label{subsub:propK}

We investigate how large $K$ needs to be for the condition $E_{\max} \subseteq \bigcup_{k = 1}^K E^*_{G^{(k)}}$ to hold with high probability. To gain intuition, suppose each graph is generated according to an ordered random DAG model with edge probability $p_{\mathrm{edge}}\in(0,1)$. Specifically, fixing an ordering $\sigma^*$, for each $k\in[K]$ and each ordered pair
$(\sigma^*(i), \sigma^*(j))$ with $i < j$, the directed edge $\sigma^*(i)\to \sigma^*(j)$ is included in $G^{(k)}$ independently with probability $p_{\mathrm{edge}}$. This scheme is commonly used to generate DAGs in simulation studies.
The following proposition provides a rough idea of how many datasets need to be collected in this case. 

\begin{proposition}\label{prop:K}
Fix an ordering $\sigma^*\in \mathbb S^p$ and sample $G^{(1)},\ldots,G^{(K)}$ from an ordered random DAG model with edge probability  $p_{\mathrm{edge}}\in(0,1)$ independently. Then, for any $\epsilon > 0$, there exists  $C = C(p_{\mathrm{edge}}, \epsilon)>0$, such that $K \ge C \log p$ implies
\begin{equation*}
\mathbb P\!\left(E_{\max}\subseteq \bigcup_{k = 1}^K E^*_{G^{(k)}}\right)
\geq 1 - \epsilon.  
\end{equation*}
\end{proposition}
\begin{proof}
Fix a DAG $G$ and an edge $a\to b\in G$. If there exists $c\neq a$ such that $c\to b\in G$ and $a$ is not adjacent to $c$ in $G$, then $a\to b\leftarrow c$ is a v-structure in $G$. Hence $a\to b$ is essential in the Markov
equivalence class of $G$. Pick an edge $e=a\to b\in E_{\max}$ and consider one sampled graph $G^{(k)}$. Recall that $P_b^{\sigma^*}$ denotes the set of predecessors of $b$ under the ordering $\sigma^*$.
If there exists $c\in P_b^{\sigma^*}\setminus \{a\}$ such that $c\to b$ is present and $a$ is not adjacent to $c$, then $e$ is essential in $G^{(k)}$ by the argument above. We can lower bound the probability that $e=a\to b$ is essential in $G^{(k)}$ by the probability that it participates in a v-structure. 
\begin{align*}
\mathbb P\!\left(e \in E^*_{G^{(k)}} \right)
\ \ge\ q_{\min} := p_{\mathrm{edge}}^2(1-p_{\mathrm{edge}}).
\end{align*}
This lower bound may be conservative, since an edge can be essential even when it does not belong to a v-structure. The probability also depends on the position of $b$. If $b$ stands later in the ordering, $|P_b^{\sigma^*}|$ increases, yielding more opportunities for $e$ to form a v-structure of the form $a\to b \leftarrow c$ with $c\in P_b^{\sigma^*}\setminus\{a\}$. Consequently,
$\mathbb P(e \in E^*_{G^{(k)}} )$ is at least $q_{\min}$ and is strictly larger than $q_{\min}$ whenever
$|P_b^{\sigma^*}\setminus\{a\}|\ge 2$.
In particular, among edges in $E_{\max}$, the smallest  probability of being essential is attained by $\sigma^*(1)\to \sigma^*(3)$ and $\sigma^*(2)\to \sigma^*(3)$. 
Since $\{G^{(k)}\}_{k=1}^K$ are independent,
\begin{align*}
    \mathbb P\!\left(e\notin \bigcup_{k=1}^K E^*_{G^{(k)}}\right)
=\prod_{k=1}^K \mathbb P(e\notin E^*_{G^{(k)}})
\le (1-q_{\min})^K.
\end{align*}
Finally, applying the union bound over $|E_{\max}|=p-1$ edges gives
\begin{align*}
    \mathbb P\!\left(E_{\max}\not\subseteq \bigcup_{k=1}^K E^*_{G^{(k)}}\right)
\le \sum_{e\in E_{\max}} \mathbb P\!\left(e\notin \bigcup_{k=1}^K E^*_{G^{(k)}}\right)
\le (p-1)(1-q_{\min})^K.
\end{align*}
It is enough to require $(p-1)(1-q_{\min})^K \leq \epsilon$, rearranging the term yields 
 \begin{align*}
     K \ \ge\ \frac{\log\!\bigl((p-1)/\epsilon\bigr)}{-\log(1-q_{\min})}.
 \end{align*}
 For fixed $p_{\mathrm{edge}} \in (0, 1)$, we have $q_{\min} > 0$. Therefore, the sufficient order is $K = O(\log p)$ and we choose $C = C(p_{\mathrm{edge}}, \epsilon)$. 
\end{proof}
The bound is conservative and may be tightened. Moreover, different graph generating mechanisms may require substantially different values of $K$.

\newpage

\subsection{Proofs in Section~\ref{subsec:highdim}}\label{subsec:proof_model}

\subsubsection{Some event sets for proofs}\label{subsub:event}

Recall that the random vector $\sX^{(k)}$ is generated by a linear SCM associated with a DAG $G^{(k)}_*$ with Gaussian errors, as described in~\eqref{eq:true_pop_data}. The dataset $X^{(k)}$ then follows as
\begin{equation}\label{eq:true_data}
    X_j^{(k)} = \sum_{i \in \Pa_j(G^{(k)}_*)} (B^{(k)}_*)_{ij} X_i^{(k)} + \varepsilon^{(k)}_{j}, \quad \varepsilon^{(k)}_{j} \sim \mathrm{MVN}_{n_k}(0, \omega^{(k)}_{*,j}I) \text{ independently for } j \in [p]. 
\end{equation}
For any ordering $\sigma \in \bbS^p$, the distribution of $\sX^{(k)}$ admits an equivalent expression~\eqref{eq:permutedSEM}. 
At the sample level, we can write the  corresponding structural equation model using the data matrix $X^{(k)}$ as
\begin{equation}\label{eq:data}
    X_j^{(k)} = \sum_{i \in \Pa_j(G^{(k)}_\sigma)} (B^{(k)}_\sigma)_{ij} X_i^{(k)} + \varepsilon^{(k)}_{\sigma,j}, \quad \varepsilon^{(k)}_{\sigma,j} \sim \mathrm{MVN}_{n_k}(0, \omega^{(k)}_{\sigma,j}I) \text{ independently for } j \in [p]. 
\end{equation}
Since the ordering corresponding to the true structural causal model is unknown, our analysis requires uniform control of the error vectors $\{\varepsilon^{(k)}_{\sigma,j}: \sigma \in \bbS^p, j \in [p], k \in [K]\}$. To achieve this, we construct high-probability events on which these vectors satisfy suitable concentration bounds simultaneously. Define the set of the standardized error vectors by 
\begin{equation*}
\mathcal{Z} = \{ z_{\sigma,j}^{(k)} = (\omega^{(k)}_{\sigma,j})^{-1/2} \varepsilon^{(k)}_{\sigma,j} : \sigma \in \mathbb{S}^p,\; j \in [p], \; k \in [K] \}.
\end{equation*}
Although the representation is indexed by permutations $\sigma \in \bbS^p$, it is not necessary to control all $p!$ permutations to apply a union bound over the elements of $\mathcal{Z}$. To see this, the regression coefficients $\{(B^{(k)}_\sigma)_{ij}: i\in \Pa_j(G^{(k)}_\sigma)\}$ and the error term $\varepsilon^{(k)}_{\sigma,j}$ in~\eqref{eq:data} depend only on the parent set $\Pa_j(G^{(k)}_\sigma)$, not directly on the ordering $\sigma$ itself. 
To be more specific, for an index set $S \subseteq [p]$, let $\proj^{(k)}_{S}
= X^{(k)}_{S}\bigl(X^{(k)\top}_{S}X^{(k)}_{S}\bigr)^{-1}X^{(k)\top}_{S}$ 
denote the orthogonal projection matrix onto the column space of $X^{(k)}_{S}$, and let $\proj^{(k)\perp}_{S} = I-\proj^{(k)}_{S}$
denote the orthogonal projection onto its orthogonal complement. Then, for any $S \subseteq [p]$, the residual vector obtained from regressing $X_j^{(k)}$ on $X_S^{(k)}$ is given by $\proj_S^{(k)\perp} X_j^{(k)}.$ In particular, when $S = \Pa_j(G_\sigma^{(k)})$, we have 
$\varepsilon^{(k)}_{\sigma,j} = \proj^{(k)\perp}_{S} X_j^{(k)}$. Therefore, if two permutations $\sigma$ and $\sigma'$ induce the same parent set $S$ for node $j$, then $\varepsilon^{(k)}_{\sigma,j} = \varepsilon^{(k)}_{\sigma',j}.$

Recall that each possible parent set has cardinality at most $d^*$, as defined in~\eqref{eq:d*}.
For a fixed node $j \in [p]$, the number of possible parent sets of $X_j^{(k)}$ is bounded by $\sum_{\ell = 0}^{d^*} {p \choose \ell} \leq p^{d^*}$, and hence the number of possible standardized error vectors $\{z_{\sigma,j}^{(k)}: \sigma \in \bbS^p\}$ is also bounded by $p^{d^*}$. Therefore, $|\mathcal{Z}| \leq Kp^{d^* +1}$. Let $\mathcal{P}_j^\sigma(d) = \{ S \subseteq P_j^\sigma: |S| \leq d\}$. Define
    \begin{align*}
        \cA_k &= \left\{\min_{\sigma \in \bbS^p} \min_{j \in [p]} \min_{S \in \mathcal{P}_j^\sigma(2d)} z_{\sigma,j}^{(k)\top} \proj^{(k)\perp}_{S} z_{\sigma,j}^{(k)} \geq n_k/2 \right\}\\
        \cB_k &= \left\{\max_{\sigma \in \bbS^p} \max_{j \in [p]} \max_{S \in \mathcal{P}_j^\sigma(2d-1)} \max_{\ell \in P^\sigma_j \setminus S }z_{\sigma,j}^{(k)\top} (\proj^{(k)}_{S\cup \{\ell\}} - \proj^{(k)}_{S}) z_{\sigma,j}^{(k)} \leq \rho \log p\right\}\\
        \cC_k &= \Bigl\{
n_k\underline{\nu}
\;\le\;
\min_{S \subseteq [p]: |S| \leq 2d }
\lambda_{\min}\bigl(X^{(k)\top}_S X^{(k)}_S\bigr)
\;\le\;
\max_{S \subseteq [p]: |S| \leq 2d }
\lambda_{\max} \bigl(X^{(k)\top}_S X^{(k)}_S\bigr)
\;\le\;
n_k\overline{\nu}
\Bigr\},
    \end{align*}
and let $\cA = \bigcap_{k\in [K]} \cA_k, \cB = \bigcap_{k\in [K]} \cB_k$ and $\cC = \bigcap_{k\in [K]} \cC_k$. The following lemma implies that $\cA \cap \cB \cap \cC$ is a high-probability event, on which our analysis is carried out.
\begin{lemma}
    Assume that (A1), (A2), and (A3) hold. Then, for sufficiently large $n$, we have $\mathbb{P}^*(\cA \cap \cB \cap \cC) \ge 1 - p^{-1}$, 
    where $\bbP^*$ denotes the probability measure over the true model from~\eqref{eq:true_data}.
\end{lemma}
\begin{proof}
    The individual probabilities of the events $\cA_k$, $\cB_k$, and $\cC_k$ are controlled using Lemmas~F1 and~F2 of~\citep{Zhou2023-cm}, and the extension to the joint event is obtained by applying a union bound over $k \in [K]$. Observe $d \log p = o(n_k)$ since $n_k = \Theta(n)$ for all $k \in [K]$.
    
    For each $k\in [K]$, by applying the same tail bound used in Lemma F2 of~\citep{Zhou2023-cm} to candidate sets $S \in \mathcal{P}_j^\sigma (2d)$,  we have $\mathbb{P}^*(\cA_k) \;\ge\; 1 - \exp (-n_k /96)$, and there exists $\rho = O(d)$ such that  
    $\mathbb{P}^*(\cB_k) \;\ge\; 1 - 2p^{-2}$.
    Then, by the union bound, we have
    \begin{align*}
        \mathbb{P}^*(\cA^c) &= \mathbb{P}^*\Bigl(\bigcup_{k\in [K]} \cA_k^c\Bigr) \leq 
\sum_{k=1}^K \mathbb{P}^*(\cA_k^c) \le
 \sum_{k=1}^K \exp\!\left(-n_k/96 \right) \leq K \exp\!\left(-n/96\right),
\\
\mathbb{P}^*(\cB^c) &= \mathbb{P}^*\Bigl(\bigcup_{k\in [K]} \cB_k^c\Bigr) \leq 
\sum_{k=1}^K \mathbb{P}^*(\cB_k^c) \le 2 K p^{-2},
    \end{align*}
    From the Lemma F1 of~\citep{Zhou2023-cm}, for each $k\in [K]$, we have $\mathbb{P}^*(\cC_k) \;\ge\; 1 - 2 \exp(-n_k \delta_0^2 / 16).$
    Therefore, by the union bound,
    \begin{align*}
\mathbb{P}^*(\cC^c)
&=
\mathbb{P}^*\Bigl(\bigcup_{k\in [K]} \cC_k^c\Bigr)
\le
\sum_{k=1}^K \mathbb{P}^*(\cC_k^c) \le
2 \sum_{k=1}^K \exp\!\left(-n_k \delta_0^2/16\right)
\\
& \le
2K \exp\!\left(-n \delta_0^2/16\right),
\end{align*}
where the last inequality follows from $n_k \ge n$ for all $k \in [K]$. Since $\log K = o(n)$, we have $K \exp (-c n) = o(1)$ for any constant $c>0$, and $Kp^{-2} = o(p^{-1})$ from $K = o(p)$, we conclude   $\mathbb{P}^*(\cA^c \cup \cB^c  \cup \cC^c) \leq 2K \exp\!\left(-n/96\right)+ 2 K p^{-2}+ 2K \exp\!\left(-n \delta_0^2/16\right) \leq p^{-1}$.
\end{proof}

\subsubsection{Proof of Proposition~\ref{prop:score-eq}}\label{subsub:MAP}

For proving (1), we use decomposability of the score $\phi^{(k)}(G)=\sum_{j\in[p]}\phi^{(k)}_j(\Pa_j(G))$, where $\phi^{(k)}_j$ is defined in~\eqref{eq:nodewise}. 
Fix $\sigma\in\bbS^p$, $k\in[K]$, and $j\in[p]$, and let $S^*=\Pa_j(G^{(k)}_\sigma)$. 
It suffices to show that for any $S\in \mathcal{P}_j^\sigma(d)$, if $S\neq S^*$ then $S$ cannot maximize $\phi^{(k)}_j(\cdot)$ over $\mathcal{P}_j^\sigma(d)$.

\textbf{Case 1: } $S^*\subsetneq S$.
Pick any $\ell\in S\setminus S^*$ and define $S'=S\setminus\{\ell\}$. Then
\begin{align*}
\phi^{(k)}_j(S)-\phi^{(k)}_j(S')
&=
-\Big(c_0\log p+\frac{1}{2}\log(1+\alpha/\gamma)\Big)\big(|S|-|S'|\big)
-\frac{\alpha n_k+\kappa}{2}\log\!\Big(\frac{n_k\,\widehat\omega^{(k)}_j(S)}{n_k\,\widehat\omega^{(k)}_j(S')}\Big) \\[4pt]
&=
-\Big(c_0\log p+\frac{1}{2}\log(1+\alpha/\gamma)\Big)
+\frac{\alpha n_k+\kappa}{2}\log\!\Big(\frac{\widehat\omega^{(k)}_j(S')}{\widehat\omega^{(k)}_j(S)}\Big) \\[4pt]
&=
-\Big(c_0\log p+\frac{1}{2}\log(1+\alpha/\gamma)\Big)
+\frac{\alpha n_k+\kappa}{2}\log\!\Big(1+
\frac{\widehat\omega^{(k)}_j(S')-\widehat\omega^{(k)}_j(S)}{\widehat\omega^{(k)}_j(S)}
\Big)\\
& \leq -\Big(c_0\log p+\frac{1}{2}\log(1+\alpha/\gamma)\Big)
+\frac{\alpha n_k+\kappa}{2} \cdot
\frac{\widehat\omega^{(k)}_j(S')-\widehat\omega^{(k)}_j(S)}{\widehat\omega^{(k)}_j(S)},
\end{align*}
where we used $\log(1+x) \leq x$ for $x \geq 0$, and $\widehat\omega^{(k)}_j(S') \geq \widehat\omega^{(k)}_j(S)$ since $S'\subset S$.
Next, we upper bound the increment $\widehat\omega^{(k)}_j(S')-\widehat\omega^{(k)}_j(S)$ and lower bound $\widehat\omega^{(k)}_j(S)$.

Since $S^*\subseteq S'\subset S$, the regression residual equals the error term in~\eqref{eq:data}, and thus
\begin{align*}
\widehat\omega^{(k)}_j(S')-\widehat\omega^{(k)}_j(S)
&=
\frac{1}{n_k} \,X^{(k)\top}_j \Big(
\proj^{(k)}_{S}-\proj^{(k)}_{S'}
\Big) X^{(k)}_j  \\
&=
\frac{1}{n_k} \,\varepsilon^{(k)\top}_{\sigma, j} \Big(
\proj^{(k)}_{S}-\proj^{(k)}_{S'}
\Big)\varepsilon^{(k)}_{\sigma, j}\\
&=
\frac{\omega^{(k)}_{\sigma, j}}{n_k}\,z^{(k)\top}_{\sigma, j} \Big(
\proj^{(k)}_{S}-\proj^{(k)}_{S'}
\Big) z^{(k)}_{\sigma, j}\\
& \leq \frac{\omega^{(k)}_{\sigma, j}\rho \log p}{n_k},
\end{align*}
where the last inequality holds on the event $\cB_k$.

Moreover,
\begin{align*}
\widehat\omega^{(k)}_j(S)
&= \frac{1}{n_k}\,X^{(k)\top}_j \proj_{S}^{(k)\perp} X^{(k)}_j
= \frac{1}{n_k}\,\varepsilon^{(k)\top}_{\sigma, j} \proj_{S}^{(k)\perp}\varepsilon^{(k)}_{\sigma, j}\\
&= \frac{\omega^{(k)}_{\sigma, j}}{n_k}\,z^{(k)\top}_{\sigma, j} \proj_{S}^{(k)\perp} z^{(k)}_{\sigma, j}
\ \ge\ \frac{\omega^{(k)}_{\sigma, j}}{2},
\end{align*}
where the last inequality holds on the event $\cA_k$. Combining the above bounds, on $\cA_k\cap\cB_k$ we obtain
\[
\frac{\widehat\omega^{(k)}_j(S')-\widehat\omega^{(k)}_j(S)}{\widehat\omega^{(k)}_j(S)}
\ \le\ 
\frac{\omega^{(k)}_{\sigma, j}\rho\log p}{n_k \omega^{(k)}_{\sigma, j}/2}
=
\frac{2\rho\log p}{n_k}.
\]
Therefore,
\begin{align*}
\phi^{(k)}_j(S)-\phi^{(k)}_j(S')
&\le
-\Big(c_0\log p+\frac{1}{2}\log(1+\alpha/\gamma)\Big)
+\frac{\alpha n_k+\kappa}{2}\cdot \frac{2\rho\log p}{n_k}\\
&=
-\Big(c_0\log p+\frac{1}{2}\log(1+\alpha/\gamma)\Big)
+\rho\Big(\alpha+\frac{\kappa}{n_k}\Big)\log p\\
&\le
\bigl(\rho(\alpha+1)-c_0\bigr)\log p
\ \le\ -t\log p\ <\ 0,
\end{align*}
where we used $\kappa\le n\le n_k$ and assumption (A3) in the last two inequalities. 
Hence $\phi^{(k)}_j(S')>\phi^{(k)}_j(S)$, so no strict superset $S\supsetneq S^*$ can maximize $\phi^{(k)}_j(\cdot)$ over $\mathcal{P}_j^\sigma(d)$. We note that the result holds for $S \in \mathcal{P}_j^\sigma(2d)$ such that $S^* \subsetneq S$.

\textbf{Case 2: $S^* \not\subseteq S$.} Here, we note that since $|S|, |S^*| \leq d$, we have $|S \cup S^*| \leq 2d$.
We will show that there exists $\ell \in S^* \setminus S$ such that 
$\phi^{(k)}_j(S) <\phi^{(k)}_j(S \cup \{\ell\})$. 
Let $S' = S \cup \{\ell\}$. Then 
\begin{align*}
\phi^{(k)}_j(S)-\phi^{(k)}_j(S')
&=
c_0\log p+\frac{1}{2}\log(1+\alpha/\gamma)
-\frac{\alpha n_k+\kappa}{2}\log\!\Big(\frac{\widehat\omega^{(k)}_j(S)}{\widehat\omega^{(k)}_j(S')}\Big) \\[4pt]
&=
c_0\log p+\frac{1}{2}\log(1+\alpha/\gamma)
-\frac{\alpha n_k+\kappa}{2}\log\!\Big(1+
\frac{\widehat\omega^{(k)}_j(S)-\widehat\omega^{(k)}_j(S')}{\widehat\omega^{(k)}_j(S')}
\Big).
\end{align*}
Since $S\subset S'$, we have $\widehat\omega^{(k)}_j(S)\ge \widehat\omega^{(k)}_j(S')$.
Using the inequality $\log(1+x)\ge x/(1+x)$ for $x\ge 0$, we obtain
\begin{align*}
\phi^{(k)}_j(S)-\phi^{(k)}_j(S')
&\le 
c_0\log p+\frac{1}{2}\log(1+\alpha/\gamma)
-\frac{\alpha n_k+\kappa}{2}\cdot
\frac{\widehat\omega^{(k)}_j(S)-\widehat\omega^{(k)}_j(S')}{\widehat\omega^{(k)}_j(S)}\\
&= c_0\log p+\frac{1}{2}\log(1+\alpha/\gamma)
-\frac{\alpha n_k+\kappa}{2}\cdot
\frac{X^{(k)\top}_j \Big(
\proj^{(k)}_{S'}-\proj^{(k)}_{S}
\Big) X^{(k)}_j }{X^{(k)\top}_j
\proj^{(k)\perp}_{S}
 X^{(k)}_j }
\end{align*}

Next we  upper bound $X^{(k)\top}_j 
\proj^{(k)\perp}_{S}
 X^{(k)}_j$ and lower bound $X^{(k)\top}_j \Big(
\proj^{(k)}_{S'}-\proj^{(k)}_{S}
\Big) X^{(k)}_j $. For the first term, we have
\begin{align*}
    X^{(k)\top}_j 
\proj^{(k)\perp}_{S}
 X^{(k)}_j \leq X^{(k)\top}_j X^{(k)}_j  \leq n_k\overline{\nu},
\end{align*}
where the last inequality is from the event $\cC_k$. 
For the second term, let $Z^{(k)}_j  = \sum_{i \in S^*} (B^{(k)}_\sigma)_{ij} X_i^{(k)}$ in~\eqref{eq:data}, and let $\ell\in S^*\setminus S$ be chosen as
\[
\ell \in \arg\max_{u\in S^*\setminus S} 
\ Z^{(k)\top}_j\bigl(\proj^{(k)}_{S\cup\{u\}}-\proj^{(k)}_{S}\bigr)Z^{(k)}_j.
\]
Since $\proj^{(k)}_{S'}-\proj^{(k)}_{S}$ is an orthogonal projection matrix, we have
\[
X^{(k)\top}_j(\proj^{(k)}_{S'}-\proj^{(k)}_{S})X^{(k)}_j
=
\bigl\|(\proj^{(k)}_{S'}-\proj^{(k)}_{S})X^{(k)}_j\bigr\|_2^2.
\]
By the triangle inequality,
\begin{align*}
\bigl\|(\proj^{(k)}_{S'}-\proj^{(k)}_{S})X^{(k)}_j\bigr\|_2
&\ge
\bigl\|(\proj^{(k)}_{S'}-\proj^{(k)}_{S})Z^{(k)}_j\bigr\|_2
-
\bigl\|(\proj^{(k)}_{S'}-\proj^{(k)}_{S})\varepsilon^{(k)}_{\sigma,j}\bigr\|_2,
\end{align*}
and hence
\begin{align*}
X^{(k)\top}_j(\proj^{(k)}_{S'}-\proj^{(k)}_{S})X^{(k)}_j
\ &\ge\
\Bigl(
\bigl\|(\proj^{(k)}_{S'}-\proj^{(k)}_{S})Z^{(k)}_j\bigr\|_2
-
\bigl\|(\proj^{(k)}_{S'}-\proj^{(k)}_{S})\varepsilon^{(k)}_{\sigma,j}\bigr\|_2
\Bigr)^2 \\
 &\ge \frac{9c_0}{\alpha}\overline{\nu}\log p ,
\end{align*}
where the last inequality is followed by Lemma 5 of~\cite{Chang2022-id}, where we have
\begin{align*}
\bigl\|(\proj^{(k)}_{S'}-\proj^{(k)}_{S})Z^{(k)}_j\bigr\|_2^2
& \ge
\frac{16c_0}{\alpha}\overline{\nu}\log p  \\ 
\bigl\|(\proj^{(k)}_{S'}-\proj^{(k)}_{S})\varepsilon^{(k)}_{\sigma, j}\bigr\|_2^2
& = \omega^{(k)}_{\sigma, j} \bigl\|(\proj^{(k)}_{S'} -\proj^{(k)}_{S}) z^{(k)}_{\sigma, j}\bigr\|_2^2 \leq \rho \omega^{(k)}_{\sigma, j}\log p \leq \frac{c_0}{\alpha} \overline{\nu} \log p.
\end{align*}
Finally, choosing $\rho$ sufficiently large in (A3), we can conclude with 
\begin{align*}
\phi^{(k)}_j(S)-\phi^{(k)}_j(S')
&\leq c_0\log p+\frac{1}{2}\log(1+\alpha/\gamma)
-\frac{\alpha n_k+\kappa}{2}\cdot
\frac{9c_0 \overline{\nu} \log p }{\alpha n_k\overline{\nu}} \leq (1 - 7c_0/2) \log p < 0.
\end{align*}
Starting from $S \in \mathcal{P}_j^\sigma(d)$, we apply the result repeatedly until all variables in $S^* \setminus S$ are added. Since each intermediate set has size at most $|S \cup S^*| \leq 2d$, we have 
\begin{align*}
    \phi_j^{(k)}(S) < \phi_j^{(k)}(S \cup S^*).
\end{align*}
If $S \cup S^* \neq S^*$, we apply \textbf{Case 1} repeatedly to delete the variables in $(S \cup S^*) \setminus S^*$, obtaining $   \phi_j^{(k)}(S^*) > \phi_j^{(k)}(S^* \cup S)$. Therefore, $\{\phi_j^{(k)}(S):S\in \mathcal{P}_j^\sigma(d)\}$ is maximized at $S^* = \Pa_j(G^{(k)}_\sigma)$ for all $j \in [p]$ and $k \in [K]$, and $\phi^{(k)}$ is maximized at $G^{(k)}_\sigma$.

To see (2), we note that the score $\phi^{(k)}$
is score equivalent; that is, $\phi^{(k)}(G)=\phi^{(k)}(G')$
 for any Markov equivalent DAGs  $G$ and $G'.$
The proof follows the same argument as Lemma~3 of~\citet{Zhou2023-cm}.  
 The proof is based on the fact that any two Markov equivalent DAGs
can be connected by a sequence of covered edge reversals (Lemma~\ref{lemma:reversal_covered}), and $\phi^{(k)}$ is invariant under a single
covered reversal. Let $G'$ be obtained from $G$ by reversing a covered edge $i\to j$ in $G$ to $j\to i$ in $G'$. Observe that the parent sets of all nodes except $i$ and $j$ coincide in $G$ and $G'$.  Let $S = \Pa_i(G) = \Pa_j(G')$, we have
\begin{align*}
\phi^{(k)}(G)-\phi^{(k)}(G')
&=
\phi^{(k)}_i\!\bigl(S\bigr)
+
\phi^{(k)}_j\!\bigl(S \cup \{i\}\bigr)
-
\phi^{(k)}_i\!\bigl(S \cup \{j\}\bigr)
-
\phi^{(k)}_j\!\bigl(S\bigr) \\
&=
-\frac{\alpha n_k+\kappa}{2}\,
\log\!\Bigg(
\frac{
\widehat\omega^{(k)}_i(S)\,
\widehat\omega^{(k)}_j(S \cup \{i\})
}{
\widehat\omega^{(k)}_i(S \cup \{j\})\,
\widehat\omega^{(k)}_j(S)
}
\Bigg) \\
&=
-\frac{\alpha n_k+\kappa}{2}\,
\log\!\Bigg(
\frac{
\left(X_i^{(k)}\right)^{\top}\proj_{S}^{\perp}X_i^{(k)}
\;
\left(X_j^{(k)}\right)^{\top}\proj_{S \cup \{i\}}^{\perp}X_j^{(k)}
}{
\left(X_i^{(k)}\right)^{\top}\proj_{S \cup \{j\}}^{\perp}X_i^{(k)}
\;
\left(X_j^{(k)}\right)^{\top}\proj_{S}^{\perp}X_j^{(k)}
}
\Bigg) \\
&= 0.
\end{align*}

To see this, define residuals
\(r_i \coloneqq \proj_S^{(k)\perp} X_i^{(k)}\), \(r_j \coloneqq \proj_S^{(k)\perp} X_j^{(k)}\). Since neither \(X_i^{(k)}\) nor \(X_j^{(k)}\) lies in the column space of
\(X_S^{(k)}\), we have \(r_i \neq 0\) and \(r_j \neq 0\).
Adding one variable yields the rank-one updates
\[
\proj_{S\cup \{i\}}^{(k)\perp} \;=\; \proj_S^{(k)\perp} - \frac{r_i r_i^\top}{r_i^\top r_i},
\qquad
\proj_{S\cup \{j\}}^{(k)\perp} \;=\; \proj_S^{(k)\perp} - \frac{r_j r_j^\top}{r_j^\top r_j}.
\]
Plugging in these expressions gives
\begin{align*}
\left(X_i^{(k)}\right)^{\top}\proj_{S}^{(k)\perp}X_i^{(k)}
\;
\left(X_j^{(k)}\right)^{\top}\proj_{S \cup \{i\}}^{(k)\perp}X_j^{(k)}
&=
(r_i^\top r_i)
\left(
r_j^\top r_j
-
\frac{(r_i^\top r_j)^2}{r_i^\top r_i}
\right) \\
&=
r_i^\top r_i\, r_j^\top r_j
-
(r_i^\top r_j)^2 \\
&=
\left(
r_i^\top r_i
-
\frac{(r_i^\top r_j)^2}{r_j^\top r_j}
\right)
(r_j^\top r_j) \\
&=
\left(X_i^{(k)}\right)^{\top}\proj_{S \cup \{j\}}^{\perp}X_i^{(k)}
\;
\left(X_j^{(k)}\right)^{\top}\proj_{S}^{\perp}X_j^{(k)}.
\end{align*}
Now, we show that $\pi^{(k)}(\sigma)$ is score equivalent. Let 
\begin{align}\label{eq:unnormalized}
\tilde{\pi}^{(k)}(  \sigma)  =  \sum_{G\in \cG_p^\sigma}\,\exp\left( \phi^{(k)}(G)\right) \ind_{ \{ \widehat{G}^{(k)}_\sigma \} }(G) = \exp \left( \phi^{(k)}(\widehat{G}^{(k)}_\sigma) \right)
\end{align}
denote the unnormalized marginal contribution of the $k$-th dataset. 
Let $G^{(k)}_*$ be the true DAG generating $X^{(k)}$, and let $\sigma^*$ be an ordering consistent with $G^{(k)}_*$.   For any ordering $\sigma$ such that the minimal I-map $G^{(k)}_\sigma$ with respect to the ordering $\sigma$ is Markov equivalent to $G^{(k)}_*$, (1) implies
\begin{align*}
\log \left(\tilde{\pi}^{(k)}(  \sigma) \right) = \phi^{(k)}(\widehat{G}^{(k)}_\sigma) = \phi^{(k)}(G^{(k)}_\sigma) = \phi^{(k)}(G^{(k)}_*) = \phi^{(k)}(\widehat{G}^{(k)}_{\sigma^*}) = \log \left(\tilde{\pi}^{(k)}(\sigma^*) \right).
\end{align*}
Therefore, the corresponding normalized quantities satisfy $\pi^{(k)}(\sigma) = \pi^{(k)}(\sigma^*)$.

Next, let $G^{(k)}_\tau$ be the minimal I-map of $G^{(k)}_*$ with respect to $\tau$, and suppose that $G^{(k)}_\tau$ and $G^{(k)}_*$ are not Markov equivalent. Then $|G^{(k)}_*| < |G^{(k)}_\tau|$. 
By Theorem~\ref{thm:meek}, there exists a sequence of DAGs
\begin{align*}
    G^{(k)}_* = G_0, G_1, \dots, G_L = G^{(k)}_\tau,
\end{align*}
which satisfies the conditions of the theorem.
Let 
\begin{align*}
     U = \{u \in [L]: G_{u} \text{ is obtained from } G_{u-1} \text{ by adding a single edge} \}.
\end{align*}
Since $G^{(k)}_\tau$ and $G^{(k)}_*$ are not Markov equivalent, we have $|U| \geq 1$. If $u \notin U$, then $\phi^{(k)}(G_{u-1}) = \phi^{(k)}(G_{u})$ since $G_u$ and $G_{u-1}$ are Markov equivalent, differing only by a single covered edge reversal (Lemma~\ref{lemma:reversal_covered}). 

For any $u \in U$, we claim that $\phi^{(k)}(G_{u-1}) > \phi^{(k)}(G_{u})$. To see this, choose an ordering $\sigma$ that is consistent with $G_{u}$. Then $G_{u-1}$ is also consistent with $\sigma$ since every edge of $G_{u-1}$ appears in $G_u$. Since $G_\sigma^{(k)}$ is the minimal I-map of  $G^{(k)}_*$ with respect to $\sigma$, and since $G_{u-1}$ is an I-map of $G^{(k)}_*$ that is also consistent with $\sigma$,
it follows that $G_\sigma^{(k)} \subseteq G_{u-1}$. 

Let $S_j = \Pa_j(G^{(k)}_\sigma)$. Then we have $S_{j} \subseteq  \Pa_{j}(G_{u-1}) \subseteq \Pa_{j}(G_{u})$ for all $j \in [p]$. Also, there exists a unique index $j^* \in [p]$ such that $\Pa_{j^*}(G_{u-1}) \subsetneq \Pa_{j^*}(G_{u})$ while for all $j\neq j^*$,  $\Pa_j(G_{u}) = \Pa_j(G_{u-1})$. By using the result of \textbf{Case 1} in Section~\ref{subsub:MAP}, we have 
$\phi^{(k)}_{j^*}(\Pa_{j^*}(G_{u-1}) ) > \phi^{(k)}_{j^*}(\Pa_{j^*}(G_{u}))$ and $\phi^{(k)}_{j}(\Pa_{j}(G_{u-1}) ) = \phi^{(k)}_{j}(\Pa_{j}(G_{u}))$ for $j \neq j^*$. Therefore, 
\begin{align*}
\log \left(\tilde{\pi}^{(k)}(  \tau) \right) = \phi^{(k)}(\widehat{G}^{(k)}_\tau) \leq \phi^{(k)}(G^{(k)}_*) - t\log p =  \log \left(\tilde{\pi}^{(k)}(\sigma^*) \right)- t\log p.
\end{align*}
Therefore, we have $\pi^{(k)}(\tau) < \pi^{(k)}(\sigma^*)$ after normalization.

\subsubsection{Proof of Theorem~\ref{thm:score-eq}}\label{subsub:score-eq}

By observing
\begin{align*}
    \pi(\sigma) &= \sum_{G^{(1)}\in \cG_p^\sigma } \cdots
        \sum_{G^{(K)}\in \cG_p^\sigma } \pi(G^{(1)}, \dots, G^{(K)}, \sigma)\\
    &= \sum_{G^{(1)}\in \cG_p^\sigma } \cdots
        \sum_{G^{(K)}\in \cG_p^\sigma }\left\{\prod_{k=1}^K \pi^{(k)}(G^{(k)} , \sigma)  \right\} \\
    &= \prod_{k=1}^K \left\{ \sum_{G^{(k)}\in \cG_p^\sigma } \pi^{(k)}(G^{(k)} , \sigma)  \right\} \\
    &= \prod_{k=1}^K \pi^{(k)}(\sigma), 
\end{align*}
where the second last equality holds due to $\pi^{(k)}(G^{(k)}, \sigma) \propto \exp(\phi^{(k)}(G^{(k)}))\ind_{\{\widehat{G}_\sigma^{(k)}\}}(G^{(k)})$.
By Proposition~\ref{prop:score-eq}, we have $\pi^{(k)}$ asymptotically score equivalent. 
The posterior distribution $\pi$ satisfies Definition~\ref{def:joint_score} with $F$ given by $F(x_1, \dots, x_K) = \prod_{k=1}^K x_k $ for all $x_k > 0, k \in [K]$, as $F$ is strictly increasing in each coordinate. Moreover, the $\alpha$-fractional posterior Bayes factor
\begin{align*}
    \mathrm{BF}^\alpha_{\sigma, \tau}(X) &= \frac{
\sum_{G^{(1)} \in \mathcal G_p^\sigma}\cdots
\sum_{G^{(K)} \in \mathcal G_p^\sigma}
\prod_{k=1}^K
L_k^{\alpha}\!\left(X^{(k)} \mid G^{(k)}\right)
\,\pi_0\!\left(G^{(k)} \mid \sigma\right)
}{
\sum_{G^{(1)} \in \mathcal G_p^\tau}\cdots
\sum_{G^{(K)} \in \mathcal G_p^\tau}
\prod_{k=1}^K
L_k^{\alpha}\!\left(X^{(k)} \mid G^{(k)}\right)
\,\pi_0\!\left(G^{(k)} \mid \tau\right)
} \\
&= \frac{\pi(\sigma)}{\pi(\tau)}
\cdot
\frac{\pi_0(\tau)}{\pi_0(\sigma)} = \frac{\pi(\sigma)}{\pi(\tau)} = \frac{\prod_{k=1}^K \tilde{\pi}^{(k)}(\sigma)}{\prod_{k=1}^K \tilde{\pi}^{(k)}(\tau)},
\end{align*}
where we have $\pi_0(\sigma)\propto 1$ and $\tilde{\pi}^{(k)}$ is the unnormalized marginal contribution  defined in \eqref{eq:unnormalized}.  Suppose that $\sigma \in \bigcap_{G\in \cG_*} \cL^{\cup}(\cE(G))$ and $\tau \notin \bigcap_{G\in \cG_*} \cL^{\cup}(\cE(G))$. By Proposition~\ref{prop:score-eq}, for every $k$, $\log \tilde{\pi}^{(k)} (\sigma) - \log \tilde{\pi}^{(k)} (\tau) \geq 0$, and there exists $k' \in [K]$ such that $G^{(k')}_\tau$ is not Markov equivalent to $G^{(k')}_*$, so $\log \tilde{\pi}^{(k')} (\sigma) - \log \tilde{\pi}^{(k')} (\tau) \geq t \log p$. This leads to 
\begin{align*}
    \mathrm{BF}^\alpha_{\sigma, \tau}(X) &   =  \exp\left(\sum_{k=1}^K\phi^{(k)}(\widehat{G}_\sigma^{(k)}) - \sum_{k=1}^K\phi^{(k)}(\widehat{G}_\tau^{(k)})\right)
    \\ &= \exp\left(\sum_{k=1}^K\phi^{(k)}(G_\sigma^{(k)}) - \phi^{(k)}(G_\tau^{(k)})\right)
    \\ &\ge \exp\left( \phi^{(k')}(G_\sigma^{(k')}) - \phi^{(k')}(G_\tau^{(k')})\right) \ge p^{t} \rightarrow \infty,
\end{align*}
where the third inequality holds from Proposition~\ref{prop:score-eq} (1), and the last inequality holds from the result of \textbf{Case 1} in Section~\ref{subsub:MAP}.

\subsection{Proofs in Section~\ref{subsec:compute}}\label{subsec:proof_compute}

\subsubsection{Proof of Proposition~\ref{prop:MAP}}\label{subsub:MAP_stepwise}

\begin{proof}[Proof of Proposition~\ref{prop:MAP}]
    The proof closely follows that in Section~\ref{subsub:MAP}. Fix $\sigma\in\bbS^p$, $k\in[K]$, and $j\in[p]$, and let $S^*=\Pa_j(G^{(k)}_\sigma)$. 
We initialize the algorithm at $S = \emptyset$.  If $S^* = \emptyset$, for any $\ell \in [p]$, \textbf{Case 1} in Section~\ref{subsub:MAP} implies  $\phi^{(k)}_j(\emptyset)  > \phi^{(k)}_j(\{\ell\})$. Therefore, Algorithm~\ref{alg:stepwise} correctly returns $S^* = \emptyset$. 

Suppose now that $S^* \neq \emptyset$. During the forward step, the algorithm continues adding elements until $S^* \subseteq S$. To see this, if $S^* \setminus S \neq \emptyset$, then by \textbf{Case 2} of Section~\ref{subsub:MAP}, there exists some $ \ell \in S^* \setminus S$, such that $\phi^{(k)}_j(S \cup \{\ell\})  > \phi^{(k)}_j(S)$. Therefore, the forward phase cannot terminate while $S^* \setminus S \neq \emptyset$. Since we have assumed that the forward phase terminates within $d$ iterations, the set $S^F$ at the final iteration satisfies $S^* \subseteq S^F$. 
Once the backward step begins with $S^* \subseteq S^F$, with $|S^F| \leq d$, \textbf{Case 2} of  Section~\ref{subsub:MAP} implies that no true parent $m \in S^*$ can be removed for any $S$ satisfying $S^* \subseteq S \subseteq S^F$. On the other hand, \textbf{Case 1} of Section~\ref{subsub:MAP} guarantees that, for 
any $ \ell \in S \setminus S^*$, $\phi^{(k)}_j(S \setminus \{\ell\})  > \phi^{(k)}_j(S)$, so such elements are removed. The procedure terminates exactly at $S = S^*$.
\end{proof}

\subsection{Proofs in Section~\ref{subsec:r2r}}\label{subsec:r2r_supp}

\subsubsection{Proof of Proposition~\ref{prop:covered_minimap}}\label{subsub:covered_minimap}

\begin{proof}[Proof of Proposition~\ref{prop:covered_minimap}]
If $G_\sigma \neq G^*$, then by Theorem~\ref{thm:meek}, there exists a sequence of DAGs $G^* = G_0, \dots, G_L = G_\sigma$ as in~\eqref{eq:chickering}, consisting of more than one DAG. 
Note that $G_{L-1}$ cannot be obtained by removing an edge from $G_\sigma$, because $G_\sigma$ is defined as the minimal I-map of $G^*$ with respect to $\sigma$. If such a $G_{L-1}$ existed, it would be a strictly smaller DAG that remains an I-map of $G^*$, contradicting the minimality of $G_\sigma$.
Therefore, $G_{L-1}$ is obtained from $G_\sigma$ by reversing a covered edge. Let $i\to j \in G_\sigma$ be such covered edge. Then, $G_{L-1} = G_\sigma \cup \{j \rightarrow i \} \setminus \{i \rightarrow j\}$ is the DAG with the reverse covered edge from  $G_\sigma$. By the definition of covered edge, there is no parent node of $j$ in $\{\sigma(k) : \sigma^{-1}(i) < k < \sigma^{-1}(j) \}$. This implies that $\tau =  \mathrm{R2R}_<(\sigma , \sigma^{-1}(i), \sigma^{-1}(j))$ is consistent with $G_{L-1}$. By Lemma 1 in~\cite{chickering1995transformational}, we have 
$G_{L-1}$ is Markov equivalent to $G_\sigma$, which implies that $G_{L-1}$ is an I-map of the true DAG with respect to $\tau$. Since $G_\tau$ is the minimal I-map, we can conclude that $|G_\tau| \leq |G_{L-1}| = |G_\sigma|$.   
\end{proof}

\subsubsection{Proof of Corollary~\ref{cor:covered}}\label{subsub:covered}

\begin{proof}[Proof of Corollary~\ref{cor:covered}]
Since $\sigma \notin \cL(G^{(k)})$, the minimal I-map $G^{(k)}_\sigma$ of $G^{(k)}$ with respect to $\sigma$ is not equal to $G^{(k)}$. Because $G^{(k)}_\sigma$ is an I-map of $G^{(k)}$, Theorem~\ref{thm:meek} implies that there exists a sequence of DAGs $G^{(k)} =  G_0, G_1, \dots, G_L = G^{(k)}_\sigma$, such that, for each $\ell = 0, \dots, L-1$, $G_{\ell+1}$ is an I-map of $G_\ell$ obtained by either adding a single edge or reversing a covered edge. 

If $G^{(k)}_\sigma$ is Markov equivalent to $G^{(k)}$, then all DAGs in the sequence $G_0, G_1, \dots, G_L$ are Markov equivalent. In this case,  there exists a covered edge $j \to i \in G_{L-1}$ such that $G_{L}$ is obtained from $G_{L-1}$ by reversing this covered edge. Define $\tau = \mathrm{R2R}_<(\sigma, \sigma^{-1}(i), \sigma^{-1}(j))$. Then $|G_\tau^{(k)}| = |G_\sigma^{(k)}|$, and $\phi^{(k)}(G_\tau^{(k)}) = \phi^{(k)}(G_\sigma^{(k)})$. Applying Proposition~\ref{prop:score-eq} (1), we conclude that $\pi^{(k)}(\sigma) = \pi^{(k)}(\tau)$.

If $G^{(k)}_\sigma$ is not Markov equivalent to $G^{(k)}$, we have $|G_\sigma^{(k)}| > |G^{(k)}|$. Hence, in the sequence $G_0, G_1, \dots, G_L$, there must exist at least one step involving edge addition.
Since $G_\sigma^{(k)}$ is a minimal I-map, the final transition from $G_{L-1}$ to $G_\sigma^{(k)}$ cannot be edge addition. Therefore, the final step must be a covered edge reversal. In particular, there exists a covered edge $j \to i \in G_{L-1}$ such that $G_\sigma^{(k)}$ is obtained by reversing this edge, yielding $i \to j \in G_\sigma^{(k)}$. Define $\tau = \mathrm{R2R}_<(\sigma, \sigma^{-1}(i), \sigma^{-1}(j))$ again. Then, $\tau$ is  consistent with $G_{L-1}$. If $G_{L-1} = G_\tau^{(k)}$, then by Proposition~\ref{prop:score-eq} (1), we have $\pi^{(k)}(\sigma) = \pi^{(k)}(\tau)$. Otherwise $G_{L-1} \neq G_\tau^{(k)}$, in which case there exists a DAG $G'$ by removing a single edge from $G_{L-1}$. To see this, $G_\tau^{(k)}$ is the minimal I-map with respect to $\tau$, so it is contained in any I-map consistent with $\tau$, including $G_{L-1}$. Since $G_{L-1} \neq G_\tau^{(k)}$, $G_{L-1}$ contains an extra edge not in $G_\tau^{(k)}$, which can be removed while preserving the I-map property.
Therefore, 
\begin{align*}
    \phi^{(k)}(G_\tau^{(k)}) \geq \phi^{(k)}(G') > \phi^{(k)}(G_{L-1}) = \phi^{(k)}(G_\sigma^{(k)})
\end{align*}
Applying Proposition~\ref{prop:score-eq} (1) again, we conclude that $\pi^{(k)}(\tau)  > \pi^{(k)}(\sigma)$. Combining the two results, we obtain $\pi^{(k)}(\tau)  \geq \pi^{(k)}(\sigma)$. 
\end{proof}
\newpage

\section{Derivation of posterior distribution}\label{sec:derivation}

For each $k \in [K]$, let $L_k(B^{(k)}, \Omega^{(k)})$ denote
the likelihood under a linear Gaussian structural equation model,
\begin{equation*}
    X_j^{(k)} = \sum_{i \in \Pa_j(G^{(k)})} (B^{(k)})_{ij} X_i^{(k)} + \varepsilon^{(k)}_{j}, \quad \varepsilon^{(k)}_{j} \sim \mathrm{MVN}_{n_k}(0, \omega^{(k)}_{j}I) \text{ independently for } j \in [p],
\end{equation*}
where $\Omega^{(k)} = \diag(\omega_1^{(k)}, \dots, \omega_p^{(k)})$. The $\alpha$-fractional posterior distribution is given by
\begin{align*}
\pi\!\left(
\{B^{(k)}, \Omega^{(k)}, G^{(k)}\}_{k=1}^K, \sigma
\right) 
&= \pi\!\left(
\{B^{(k)}, \Omega^{(k)}, G^{(k)}\}_{k=1}^K \mid \sigma \right) \cdot \pi(\sigma)\\
&= \prod_{k=1}^K \pi^{(k)}\!\left(
B^{(k)}, \Omega^{(k)}, G^{(k)} \mid \sigma \right) \cdot \pi(\sigma)\\
&= \prod_{k=1}^K \left\{ L_k \left(B^{(k)}, \Omega^{(k)} \right)^\alpha
\, \pi_0\!\left(B^{(k)}, \Omega^{(k)} \mid G^{(k)} \right)
\, \pi_0\!\left(G^{(k)} \mid \sigma\right)
\,\right\} \cdot \pi(\sigma).
\end{align*}
By the normal-inverse-gamma conjugacy,
\begin{align*}
\pi^{(k)}(G^{(k)}\mid \sigma)
&\propto \pi_0(G^{(k)}\mid\sigma)\!
\int \pi_0(B^{(k)},\Omega^{(k)} \mid G^{(k)})\,
L_k(B^{(k)},\Omega^{(k)})^\alpha\,
d(B^{(k)}, \Omega^{(k)}) \\       
&\propto \pi_0(G^{(k)}\mid\sigma)
\int \prod_{j=1}^p 
\left(\frac{\omega_j^{(k)}}{\gamma}\right)^{-\frac{|\Pa_j(G^{(k)})|}{2}}
\det \Bigl(X^{(k)\top}_{\Pa_j(G^{(k)})} X^{(k)}_{\Pa_j(G^{(k)})}\Bigr)^{1/2} \\
&\hspace{-2cm} \times
\exp\!\left\{-\frac{\gamma}{2\omega_j^{(k)}}
(B^{(k)}_{\Pa_j(G^{(k)}),j}-\hat B^{(k)}_{\Pa_j(G^{(k)}),j})^\top
\Bigl(X^{(k)\top}_{\Pa_j(G^{(k)})} X^{(k)}_{\Pa_j(G^{(k)})}\Bigr)
(B^{(k)}_{\Pa_j(G^{(k)}),j}-\hat B^{(k)}_{\Pa_j(G^{(k)}),j})\right\} \\ %\displaybreak
&\hspace{-2cm} \times 
(\omega_j^{(k)})^{-\frac{\kappa}{2}-1-\frac{\alpha n_k}{2}}
\exp\!\left\{-\frac{\alpha}{2\omega_j^{(k)}}
\|X^{(k)}_j - X^{(k)}_{\Pa_j(G^{(k)})} B^{(k)}_{\Pa_j(G^{(k)}),j}\|^2\right\}\,
dB^{(k)}_{\Pa_j(G^{(k)}),j}\,d\omega_j^{(k)}.\\
&\hspace{-2.5cm} \propto 
\pi_0(G^{(k)}\mid\sigma)
\int \prod_{j=1}^p
\left( \frac{\omega_j^{(k)}}{\gamma} \right)^{-\frac{|\Pa_j(G^{(k)})|}{2}}
\left( \omega_j^{(k)} \right)^{-\frac{\alpha n_k + \kappa}{2} - 1} \exp\!\left\{
-\frac{\alpha}{2\omega_j^{(k)}}
X^{(k)\top}_j
\proj^{(k)\perp}_{\Pa_j(G^{(k)})}
X^{(k)}_j
\right\} \\
&\hspace{-2cm} \times \int \exp\!\left\{
-\frac{\alpha+\gamma}{2\omega_j^{(k)}}
\bigl(B^{(k)}_{\Pa_j(G^{(k)}),j}-\hat B^{(k)}_{\Pa_j(G^{(k)}),j}\bigr)^\top
\Bigl(X^{(k)\top}_{\Pa_j(G^{(k)})} X^{(k)}_{\Pa_j(G^{(k)})}\Bigr)
\bigl(B^{(k)}_{\Pa_j(G^{(k)}),j}-\hat B^{(k)}_{\Pa_j(G^{(k)}),j}\bigr)
\right\}\,
 \\
&\hspace{-2cm} \times 
\det \Bigl(X^{(k)\top}_{\Pa_j(G^{(k)})} X^{(k)}_{\Pa_j(G^{(k)})}\Bigr)^{1/2}
dB^{(k)}_{\Pa_j(G^{(k)}),j}\, d\omega_j^{(k)} \\
&\hspace{-2.5cm} = \pi_0 (G^{(k)}\mid\sigma) \prod_{j=1}^p  \left(1 + \frac{\alpha}{\gamma} \right)^{-\frac{|\Pa_j(G^{(k)})|}{2}}  \int 
       \left( \omega_j^{(k)} \right)^{-\frac{\alpha n_k + \kappa}{2} - 1} \exp \left\{ -\frac{\alpha}{2\omega_j^{(k)}} 
       X^{(k)\top}_j \proj^{(k)\perp}_{\Pa_j(G^{(k)})} X^{(k)}_j \right\} 
       d\omega_j^{(k)}\\
&\hspace{-2.5cm} \propto \pi_0 (G^{(k)}\mid\sigma) \left(1 + \frac{\alpha}{\gamma} \right)^{-\frac{|G^{(k)}|}{2}}\prod_{j=1}^p  \left(n_k\hat{\omega}^{(k)}_j(\Pa_j(G^{(k)}))\right)^{-\frac{\alpha n_k + \kappa}{2}},
\end{align*}
\noindent
where we recall that $\proj^{(k)\perp}_S = I - X^{(k)}_{S}\bigl(X^{(k)\top}_{S} X^{(k)}_{S}\bigr)^{-1}X^{(k)\top}_{S}$ is the orthogonal projection onto its orthogonal complement, and $\widehat{\omega}^{(k)}_j(S) = n_k^{-1} X^{(k)\top}_j \proj^{(k)\perp}_SX^{(k)}_j$ is the residual variance from the linear regression of $X^{(k)}_j$ onto $X^{(k)}_{S}$.
We obtain the joint posterior distribution of $(G^{(1)}, \dots ,G^{(K)}, \sigma)$ as 
\begin{align*}
\pi(G^{(1)},\dots,G^{(K)},\sigma)
&\propto \pi(\sigma) \cdot
\prod_{k=1}^K \pi^{(k)}(G^{(k)}\mid\sigma)\\
&\propto  
\prod_{k=1}^K  p^{-c_0|G^{(k)}|} \left(1 + \frac{\alpha}{\gamma} \right)^{-\frac{|G^{(k)}|}{2}}
\prod_{j=1}^p
\left(n_k\hat{\omega}_j^{(k)}(\Pa_j(G^{(k)}))\right)^{-\frac{\alpha n_k+\kappa}{2}} \;\ind_{\{\hat G^{(k)}_\sigma\}}(G^{(k)}).
\\
&=   \prod_{k=1}^K
\exp ({\phi^{(k)}(G^{(k)})} )\,
\ind_{\{\hat G^{(k)}_\sigma\}}(G^{(k)}).
\end{align*}
\newpage
\section{Forward-backward edge selection}\label{sec:forwardbackward}
Given an ordering $\sigma$, structure learning reduces to selecting parent sets from $P_j^\sigma$ for each $j \in [p]$. We employ a forward–backward edge selection procedure that leverages score decomposability to efficiently construct locally optimal parent sets. The forward phase greedily adds edges that improve the local score, while the backward phase removes redundant edges to mitigate overfitting and refine the solution. This node-wise optimization significantly reduces computational complexity compared to exhaustive DAG search. As stated in Remark~\ref{remark:parallel}, the procedure can be carried out independently for each node, allowing efficient parallel implementation. 

\begin{algorithm}[!b]
\caption{Forward-backward edge selection given an ordering}
\label{alg:stepwise}
\KwIn{Data matrix $X$, decomposable score function $\phi = \sum_{j=1}^p\phi_j$, and ordering $\sigma$} \vspace{-1mm}

\For{$j = 1, \ldots, p$}{ \vspace{-1mm}
\textbf{Forward edge selection}\; \vspace{-1mm}
    $S_j \leftarrow \emptyset$\; \vspace{-1mm}
    \While{$\mathrm{true}$}{ \vspace{-1mm}
    \If{$P_j^\sigma \setminus S_j = \emptyset$}{\textbf{break}\;}
        Choose $i^\star \in \arg\max_{i \in P_j^\sigma \setminus S_j} \Big( \phi_j(S_j\cup\{i\})-\phi_j(S_j) \Big)$\;  \vspace{-1mm}
        \If{$\phi_j(S_j \cup \{i^\star\}) - \phi_j(S_j) > 0$}{ \vspace{-1mm}
            $S_j \leftarrow S_j \cup \{i^\star\}$\; \vspace{-1mm}
        } \vspace{-1mm}
        \Else{ \vspace{-1mm}
            \textbf{break}\; \vspace{-1mm}
        } 
    } \vspace{-1mm}
    \textbf{Backward edge deletion}\; \vspace{-1mm}

    \While{$S_j \neq \emptyset$}{ \vspace{-1mm}
        Choose $i^\star \in \arg\max_{i \in S_j} 
        \Big( \phi_j(S_j \setminus \{i\}) - \phi_j(S_j) \Big)$\; \vspace{-1mm}
        \If{$\phi_j(S_j \setminus \{i^\star\}) - \phi_j(S_j) > 0$}{ \vspace{-1mm}
            $S_j \leftarrow S_j \setminus \{i^\star\}$\; \vspace{-1mm}
        } \vspace{-1mm}
        \Else{ \vspace{-1mm}
            \textbf{break}\; \vspace{-1mm}
        
    } 
} 
}
\KwOut{Estimated parent sets $\{S_j\}_{j=1}^p$ and induced graph $\widehat{G}_\sigma$}
\end{algorithm}

\newpage
\section{Additional simulation results}\label{sec:additional_simul}

\subsection{Performance evaluation on different settings}\label{subsec:performance_supp}

To complement the results in Section~\ref{subsec:performance}, we examine the performance of the proposed method under additional configurations. These experiments are designed to verify that the observed performance gains are robust to smaller sample sizes and are not driven by an increase in the total number of observations.

\noindent
\textbf{Smaller sample size.}
We repeat the simulation study described in Section~\ref{subsec:performance}, but reduce the sample size of each dataset to $n_k = 200$. All other settings remain unchanged. The results in Table~\ref{tab:performance_supp_1} are similar to those obtained with $n_k=1000$. This indicates that the proposed method remains effective even with substantially fewer observations per dataset.
\begin{table}[!h]
\centering
\captionsetup{font={footnotesize,stretch=1}}
\renewcommand{\arraystretch}{0.8}
\begin{tabular}{ccccccc}
\toprule
\multirow{2}{*}{} & \multicolumn{5}{c}{$p$} \\
\cmidrule(lr){3-7}
  & & 5 & 10 & 20 & 40 & 100\\ 
\midrule 
\multirow{2}{*}{$K$=1} 
 & $\Delta$ & 2.24 $\pm$ 0.29 & 2.50 $\pm$ 0.29 & 5.76 $\pm$ 0.35 & 9.55 $\pm$ 0.40 & 21.61 $\pm$ 0.54\\
 & $\tau^*$ & 0.17 $\pm$ 0.04 & 0.25 $\pm$ 0.02 & 0.21 $\pm$ 0.01 & 0.20 $\pm$ 0.01 & 0.20 $\pm$ 0.01\\
\midrule 
\multirow{2}{*}{$K$=5}
 & $\Delta$ & 1.52 $\pm$ 0.20 & 1.73 $\pm$ 0.15 & 2.75 $\pm$ 0.14 & 5.13 $\pm$ 0.19 & 12.36 $\pm$ 0.34\\
 & $\tau^*$ & 0.58 $\pm$ 0.04 & 0.66 $\pm$ 0.02 & 0.65 $\pm$ 0.01 & 0.62 $\pm$ 0.01 & 0.60 $\pm$ 0.01\\
\midrule
\multirow{2}{*}{$K$=10}
 & $\Delta$ & 0.98 $\pm$ 0.15 & 1.10 $\pm$ 0.09 & 1.74 $\pm$ 0.10 & 3.18 $\pm$ 0.13 & 8.28 $\pm$ 0.19\\
 & $\tau^*$ & 0.74 $\pm$ 0.04 & 0.83 $\pm$ 0.01 & 0.81 $\pm$ 0.01 & 0.79 $\pm$ 0.00 & 0.76 $\pm$ 0.00\\
\midrule
\multirow{2}{*}{$K$=20} 
 & $\Delta$ & 0.66 $\pm$ 0.08 & 0.62 $\pm$ 0.04 & 1.03 $\pm$ 0.06 & 1.87 $\pm$ 0.08 & 5.50 $\pm$ 0.12\\
 & $\tau^*$ & 0.83 $\pm$ 0.02 & 0.92 $\pm$ 0.01 & 0.91 $\pm$ 0.00 & 0.90 $\pm$ 0.00 & 0.87 $\pm$ 0.00\\
\bottomrule
\end{tabular}
\caption{The average Hamming distance $\Delta$ and average rank correlation $\tau^*$ for varying $K$ and $p$ with $n_k = 200$. Each value represents $\mathrm{mean}\pm 1$ standard error.}
\label{tab:performance_supp_1}
\end{table}

\noindent
\textbf{Fixed total number of samples.}
To verify that the performance improvement is not due to a larger total sample size, we conduct additional simulations in which the total number of observations is fixed while varying $K$. Specifically, we set the total sample size to $n_{\text{tot}} = 4000$ and assign $n_k = n_{\text{tot}}/K$ observations to each dataset.  Table~\ref{tab:performance_supp_2} shows that increasing $K$ leads to improved estimation accuracy. This confirms that the performance gain is from the data heterogeneity rather than from an increase in the overall sample size.
\begin{table}[!h]
\centering
\captionsetup{font={footnotesize,stretch=1}}
\renewcommand{\arraystretch}{0.8}
\begin{tabular}{ccccccc}
\toprule
\multirow{2}{*}{} & \multicolumn{5}{c}{$p$} \\
\cmidrule(lr){3-7}
  & & 5 & 10 & 20 & 40 & 100\\ 
\midrule 
\multirow{2}{*}{$K$=1} 
 & $\Delta$ & 1.53 $\pm$ 0.18 & 2.02 $\pm$ 0.15 & 4.58 $\pm$ 0.24 & 8.76 $\pm$ 0.36 & 19.45 $\pm$ 0.47\\
 & $\tau^*$ & 0.25 $\pm$ 0.03 & 0.27 $\pm$ 0.02 & 0.23 $\pm$ 0.01 & 0.20 $\pm$ 0.01 & 0.20 $\pm$ 0.01\\
\midrule 
\multirow{2}{*}{$K$=5}
 & $\Delta$ & 1.05 $\pm$ 0.10 & 1.08 $\pm$ 0.08 & 2.22 $\pm$ 0.09 & 4.56 $\pm$ 0.22 & 10.46 $\pm$ 0.22\\
 & $\tau^*$ & 0.67 $\pm$ 0.03 & 0.71 $\pm$ 0.01 & 0.65 $\pm$ 0.01 & 0.63 $\pm$ 0.01 & 0.62 $\pm$ 0.00\\
\midrule
\multirow{2}{*}{$K$=10}
 & $\Delta$ & 0.86 $\pm$ 0.11 & 0.71 $\pm$ 0.07 & 1.31 $\pm$ 0.07 & 2.87 $\pm$ 0.11 & 7.02 $\pm$ 0.16\\
 & $\tau^*$ & 0.74 $\pm$ 0.03 & 0.86 $\pm$ 0.01 & 0.82 $\pm$ 0.01 & 0.80 $\pm$ 0.01 & 0.77 $\pm$ 0.00\\
\midrule
\multirow{2}{*}{$K$=20} 
 & $\Delta$ & 0.66 $\pm$ 0.08 & 0.62 $\pm$ 0.04 & 1.03 $\pm$ 0.06 & 1.87 $\pm$ 0.08 & 5.50 $\pm$ 0.12\\
 & $\tau^*$ & 0.83 $\pm$ 0.02 & 0.92 $\pm$ 0.01 & 0.91 $\pm$ 0.00 & 0.90 $\pm$ 0.00 & 0.87 $\pm$ 0.00\\
\bottomrule
\end{tabular}
\caption{The average Hamming distance $\Delta$ and average rank correlation $\tau^*$ for varying $K$ and $p$ with $n_k = n_{\mathrm{tot}}/K$ and $n_{\mathrm{tot}} = 4000$. Each value represents $\mathrm{mean}\pm 1$ standard error.}
\label{tab:performance_supp_2}
\end{table}

\subsection{When the faithfulness assumption fails}\label{subsec:faithfulness}

As noted in Remark~\ref{remark:faithfulness}, violations of the faithfulness assumption can cause order-based structure learning methods to fail. We examine the effect of joint estimation in mitigating the impact of such violations. To this end, we repeatedly sample a DAG $G$ until it contains at least 2 triangular motifs, where a triangular motif is defined as an induced three-node subgraph $(i, j, \ell)$ with $i \to j$, $i \to \ell$, and $j \to \ell$ are present in $G$, so that there exist two distinct directed paths from $i$ to $\ell$. We set $p=20$.
For each $K \in \{1, 5, 10, 20, 40\}$, we generate $K$ such DAGs and, for each DAG $G$, simulate both faithful and unfaithful datasets. To construct the unfaithful datasets, we randomly select two triangular motifs from $G$, and adjust edge coefficients within each selected triangular motif to induce path cancellation (as the example in Remark~\ref{remark:faithfulness}), yielding a Gaussian distribution that is not faithful with respect to $G$. To construct the faithful datasets, we generate data using the same graph $G$ as described in the main text.
Figure~\ref{fig:faithful} shows an increase in average Hamming distance $\Delta$ under the unfaithful case for the single DAG estimation ($K=1$), whereas the differences are at most marginal for $K \geq 5$. The mean rank correlation $\tau^*$ is slightly smaller in the unfaithful setting when $K=1$, but increases monotonically with $K$ and rapidly approaches one, with the gap between settings diminishing and indicating nearly perfect ordering recovery at $K = 40$.
\begin{figure}[h!]
    \centering
    \captionsetup{font={footnotesize,stretch=1}}
    \begin{minipage}{0.6\textwidth}  % Adjust width as needed
            \includegraphics[width=\linewidth]{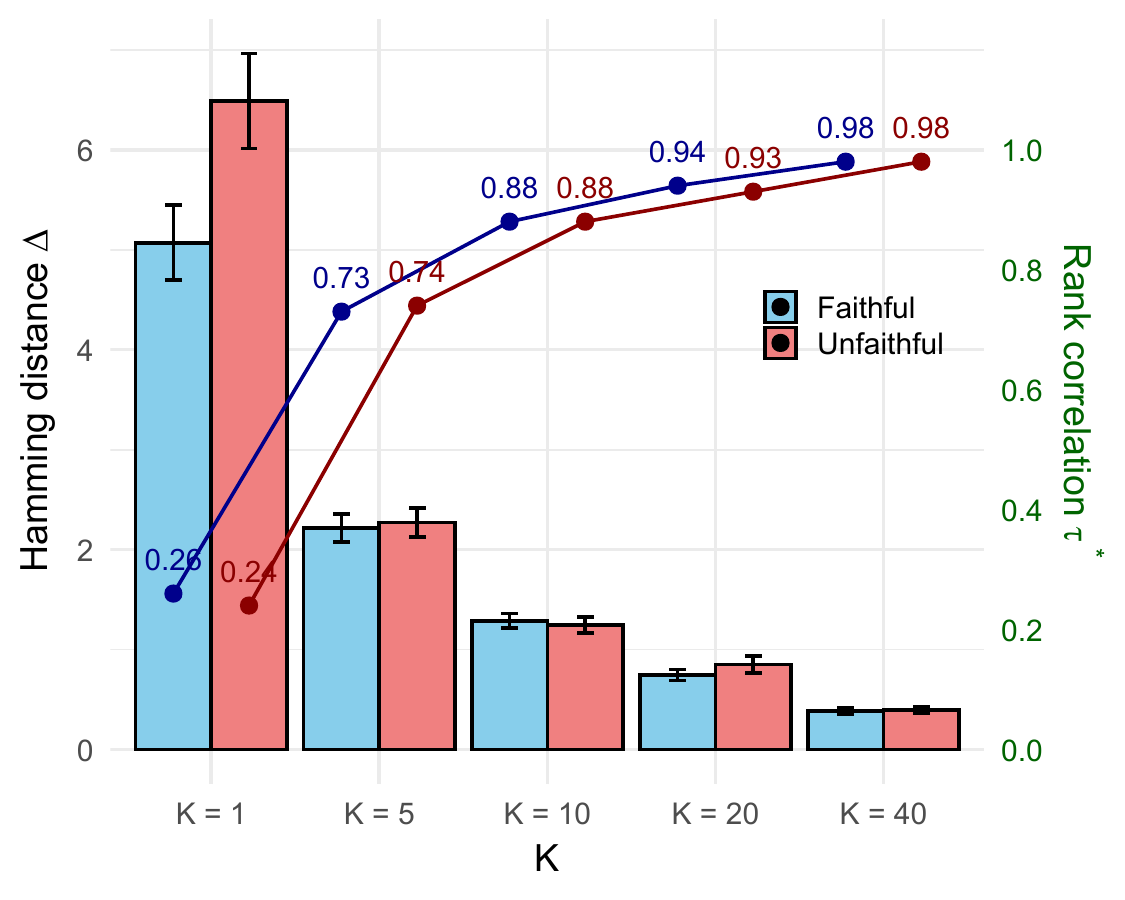}
    \end{minipage}%
    \hspace{0.2cm}
    \begin{minipage}{0.35\textwidth}  % Adjust width to control spacing
        \captionof{figure}{The average Hamming distance $\Delta$  between the true DAGs and their estimates (bar graph) and the mean rank correlation $\tau^*$ between the true ordering and estimated ordering (line graph) using Algorithm~\ref{alg:full}. We repeat 50 experiments.  The blue graphs illustrate results from a probability distribution that is faithful to the true DAG, while the red graphs depict results under unfaithful conditions. The error bars in the bar graphs represent one standard error.}
        \label{fig:faithful}
    \end{minipage}
\end{figure}

\subsection{Performance under high-dimensional regime.}\label{subsec:highdim_simul}

We empirically illustrate the result of Theorem~\ref{thm:score-eq} with varying $n$, $p$, and $K$. To emulate a high-dimensional asymptotic regime where \(n\) grows linearly and \(p\) increases exponentially, we follow the setting of Section 4.2 in \cite{Chang2022-id}, and consider five settings where \(n = 30(\ell+2)\) and \(p = 7 \cdot 2^{\ell}\) in the \(\ell\)-th setting. Note that \(\ell= 5\), we have \(p > n\). We fix $\sigma^* = (1, \dots, p)$, and set $K =  \left\lfloor (p \log p)/20 \right\rfloor$ with \(p_{\mathrm{edge}} = d/(p-1)\), where \(d = 0.2\sqrt{n}\) is the expected number of neighbors for each node. 
We generate $K$ random DAGs, and for each $k \in [K]$, simulate the corresponding data matrix $X^{(k)}$. We enumerate the union of essential arrows across all $K$ DAGs and denote this set by $\mathbf{E}$. We randomly select 50 edges from $\mathbf{E}$. For each essential arrow $i \to j \in \mathbf{E}$, we get an inconsistent ordering $\tau$ by placing $j$ before $i$ in $\sigma^*$.
We select $c_0 = 1$, and compute its posterior probabilities $\pi(\tau)$. We then compare each ordering $\tau$ with $\sigma^*$ using the Bayes factor $\mathrm{BF}^\alpha_{\sigma^*, \tau}.$ 
In Table~\ref{tab:highdim}, we can observe that the log Bayes factor grows rapidly as the problem dimension increases, indicating stronger evidence in favor of $\sigma^*$. 

\begin{table}[t]
\centering
\captionsetup{font={footnotesize,stretch=1}}
\begin{tabular}{c c c c c }
\toprule
$p$ & $n$ & $d$ & $K$ & log(BF) \\
\midrule
14  & 90  & 1.897 & 1  & 8.423 $\pm$ 0.245\\
28  & 120 & 2.191 & 4  & 43.052 $\pm$ 5.036\\
56  & 150 & 2.449 & 11 & 156.280 $\pm$ 17.167\\
112 & 180 & 2.683 & 26 & 494.572 $\pm$ 54.318\\
224 & 210 & 2.898 & 60 & 1613.487 $\pm$ 188.020\\
\bottomrule
\end{tabular}
\caption{Logarithm of the Bayes factor, denoted by $\log(\mathrm{BF})$, under the five different settings. Each value represents $\mathrm{mean}\pm 1$ standard error.
}
\label{tab:highdim}
\end{table}

\subsection{Comparison with other methods with different settings}\label{subsec:compare_supp}
For all competing methods, tuning parameters are selected according to the suggestion in~\cite{Li2024-af}. For the PC algorithm, the significance level used in the conditional independence tests is set to $\alpha = 0.005$. For GES and joint GES, the $\ell_0$-penalization parameter (scaled by $\log p$) is fixed at $\lambda = 2$. For the muSuSiE-DAG method, the penalty parameters are specified as $p^{-\omega_1} = p^{-1.75}/5$, $p^{-\omega_2} = p^{-2}/10$, $p^{-\omega_3} = p^{-2.25}/10$, $p^{-\omega_4} = p^{-2.5}/5$, and $p^{-\omega_5} = p^{-3}$, 
where $\omega_1, \dots, \omega_5$ denote the corresponding penalty exponents. We run muSuSiE-DAG for a total of 100,000 MCMC iterations, retaining the final 5,000 samples after a burn-in of 95,000 iterations.
For the proposed method, we choose $c_0 = 7$, and run for 20,000 iterations, with the first 10,000 iterations discarded as burn-in. 
Table~\ref{tab:comparison1} reports the results under the setting  $|\mathcal{E}_{\mathrm{com}}| = 50$ and $|\mathcal{E}^{(k)}_{\mathrm{pri}}| = 50$,
where the numbers of common and dataset-specific edges are balanced. Table~\ref{tab:comparison2} presents the results under  $|\mathcal{E}_{\mathrm{com}}| = 100$ and $|\mathcal{E}^{(k)}_{\mathrm{pri}}| = 20$, representing a scenario with stronger shared structure across datasets.

\begin{table}[h!]
\centering
\footnotesize
\captionsetup{font={footnotesize,stretch=1}}
\renewcommand{\arraystretch}{1.0}
%\small
\begin{tabular}{cccccc}
\toprule
 & \textbf{PC} & \textbf{GES} & \textbf{Joint GES} & \textbf{muSuSiE-DAG} & \textbf{Proposed} \\
\midrule
$\Delta$ & 36.928 $\pm$ 0.434 & 7.252 $\pm$ 0.483 & 47.204 $\pm$ 1.489 & 34.904 $\pm$ 2.871 & 9.117 $\pm$ 0.266\\
TPR & 0.887 $\pm$ 0.003 & 0.968 $\pm$ 0.002  & 0.879 $\pm$ 0.006 & 0.877 $\pm$ 0.010 & 0.961 $\pm$ 0.001\\
%FPR & 0.003 $\pm$ 0.000 & 0.000 $\pm$ 0.000 & 0.004 $\pm$ 0.000 & 0.002 $\pm$ 0.000 & 0.001 $\pm$ 0.000\\
FDR & 0.186 $\pm$ 0.003 & 0.055 $\pm$ 0.003 & 0.164 $\pm$ 0.003 & 0.192 $\pm$ 0.016 & 0.046 $\pm$ 0.001\\
\bottomrule
\end{tabular}
\caption{Comparison of methods for $K = 5$, $p = 100$, and $n_k = 240$ for $k \in [5]$, under the setting $|\mathcal{E}_{com}| = 50$, $|\mathcal{E}^{(k)}_{pri}| = 50$. $\Delta$ is the average Hamming distance, defined as in~\eqref{eq:delta}. Each value represents $\mathrm{mean}\pm 1$ standard error. }
\label{tab:comparison1}
\end{table}

\begin{table}[h!]
\centering
\footnotesize
\captionsetup{font={footnotesize,stretch=1}}
\renewcommand{\arraystretch}{1.0}
%\small
\begin{tabular}{cccccc}
\toprule
 & \textbf{PC} & \textbf{GES} & \textbf{Joint GES} & \textbf{muSuSiE-DAG} & \textbf{Proposed} \\
\midrule
$\Delta$ & 39.556 $\pm$ 0.644 & 11.972 $\pm$ 0.769 & 25.508 $\pm$ 0.602 & 38.856 $\pm$ 3.662 & 9.787 $\pm$ 0.401\\
TPR & 0.869 $\pm$ 0.003 & 0.957 $\pm$ 0.003  & 0.980 $\pm$ 0.002 & 0.905 $\pm$ 0.008 & 0.965 $\pm$ 0.001\\
%FPR & 0.002 $\pm$ 0.000 & 0.000 $\pm$ 0.000 & 0.002 $\pm$ 0.000 & 0.003 $\pm$ 0.000 & 0.001 $\pm$ 0.000\\
FDR & 0.224 $\pm$ 0.002 & 0.040 $\pm$ 0.003 & 0.283 $\pm$ 0.007 & 0.198 $\pm$ 0.015 & 0.052 $\pm$ 0.001\\
\bottomrule
\end{tabular}
\caption{Comparison of methods for $K = 5$, $p = 100$, and $n_k = 240$ for $k \in [5]$, under the setting $|\mathcal{E}_{com}| = 100$, $|\mathcal{E}^{(k)}_{pri}| = 20$. $\Delta$ is the average Hamming distance, defined as in~\eqref{eq:delta}. Each value represents $\mathrm{mean}\pm 1$ standard error. }
\label{tab:comparison2}
\end{table}

\subsection{Mixing behavior under other neighborhoods}\label{subsec:mixing_other}

With the same dataset used in Section~\ref{subsec:conv}, we examine how the choice of proposal neighborhood affects the mixing of Algorithm~\ref{alg:full} by comparing the R2R neighborhood with two other neighborhoods. One neighborhood, the adjacent transposition neighborhood $\mathcal{N}_{\mathrm{ADJ}}$, is defined in the main text.
The other neighborhood is the random transposition neighborhood, defined as $\mathcal{N}_{\mathrm{RTS}}(\sigma) = \{\sigma' \in \mathbb{S}^p \mid \sigma' = \mathrm{RTS}(\sigma, i, j),\; i < j,\; i,j \in [p]\},$ where $\mathrm{RTS}(\sigma, i, j)$ produces a new permutation by interchanging the $i$-th and $j$-th elements of $\sigma$, leaving all other positions unchanged,
\begin{align*}
    \mathrm{RTS}(\sigma, i, j):
(\sigma(1), \dots, \sigma(i), \dots, \sigma(j), \dots, \sigma(p))
\mapsto
(\sigma(1), \dots, \sigma(j), \dots, \sigma(i), \dots, \sigma(p)),
\end{align*}
for $i < j$ with $i,j \in [p]$. 
To account for differences in per-iteration computational cost across neighborhood types, we follow the Markov chain sample size calculation suggested by~\cite{Chang2022-id}. Accordingly, we set the number of iterations to 213,333 for the adjacent transposition neighborhood and to 32,000 for the other two neighborhoods. We present the results in Figure~\ref{fig:neighbor}. Notably, the adjacent transposition neighborhood exhibits poor convergence, with 16 chains becoming trapped in local modes. In contrast, random transposition performs similarly to the random-to-random neighborhood, although it is not supported by our theoretical analysis.

\begin{figure}[t!]
    \centering
    \captionsetup{font={footnotesize,stretch=1}}
    \begin{subfigure}[b]{0.30\linewidth}
        \centering
        \includegraphics[width=\linewidth]{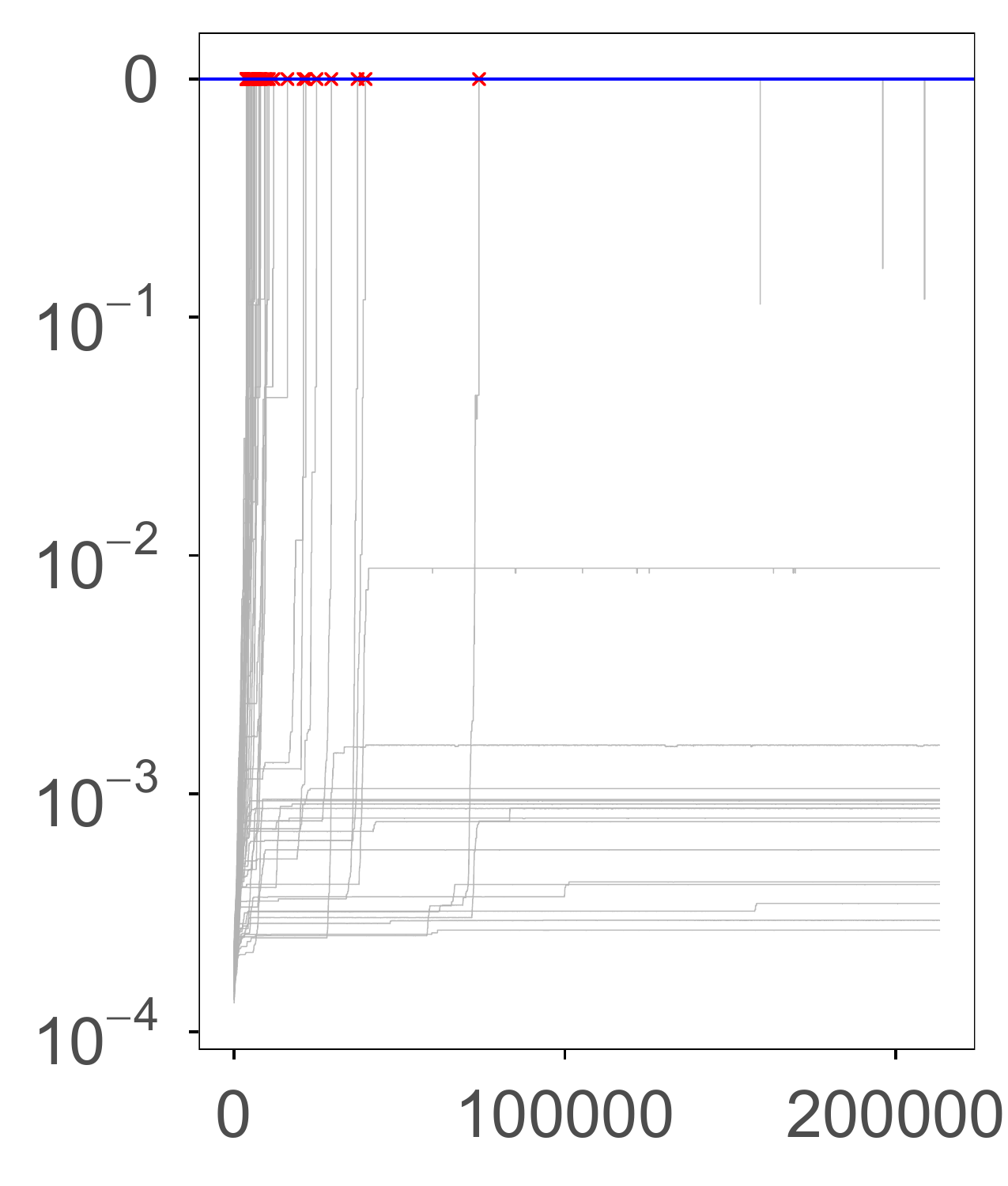}
        \caption{}
        \label{fig:neighbor_a}
    \end{subfigure}
    \hfill
    \begin{subfigure}[b]{0.30\linewidth}
        \centering
        \includegraphics[width=\linewidth]{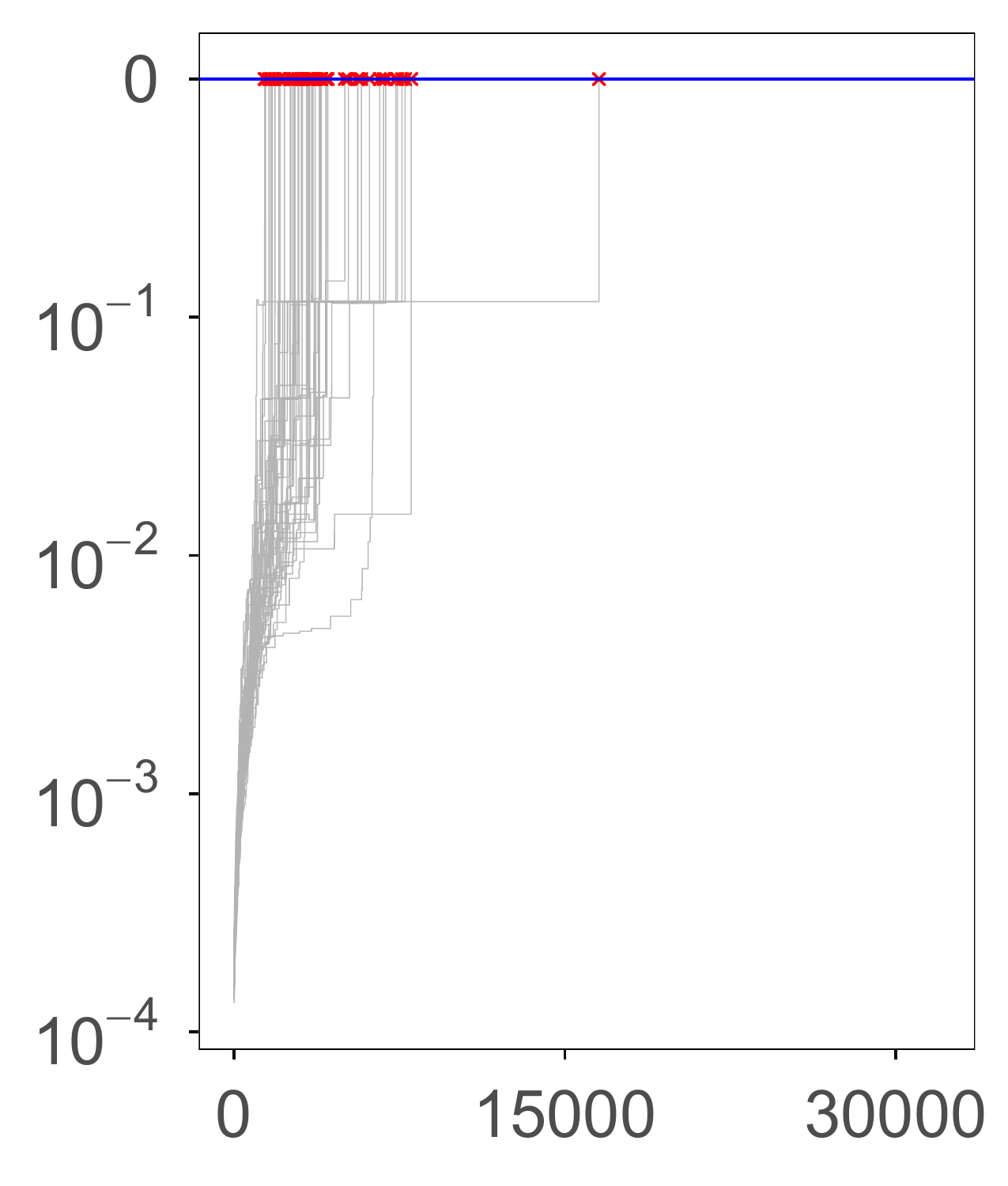}
        \caption{}
        \label{fig:neighbor_b}
    \end{subfigure}
    \hfill
    \begin{subfigure}[b]{0.30\linewidth}
        \centering
        \includegraphics[width=\linewidth]{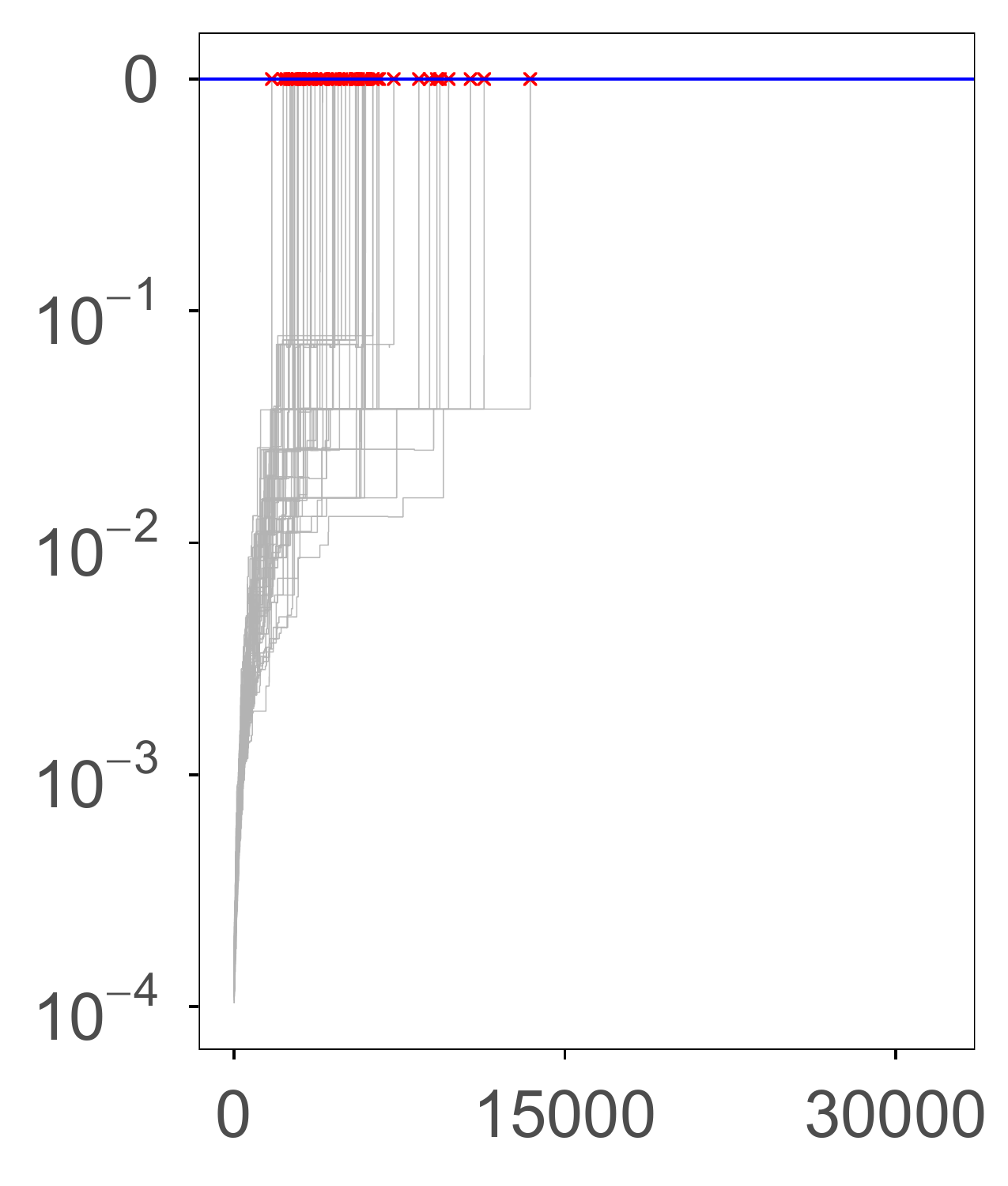}
        \caption{}
        \label{fig:neighbor_c}
    \end{subfigure}

    \caption{Scaled log posterior probability versus the number of iterations of 50 MCMC runs with random initialization for \(p = 40\) and \(K = 20\). The blue line represents the scaled log-posterior probability of the true ordering. Each red cross marks the first time the chain attains the value. Panels (a), (b), and (c) correspond to Algorithm~\ref{alg:full} with the adjacent transposition,  random transposition, and random-to-random, respectively.}
    \label{fig:neighbor}
\end{figure}

\section{Description of the real data}\label{sec:real_supp}

We analyze gene expression measurements from individuals diagnosed with major depressive disorder (MDD) and healthy controls; the sample sizes for each dataset are summarized in Table~\ref{tab:samplesize}.

\begin{table}[h!]
\centering
\footnotesize
\captionsetup{font={footnotesize,stretch=1}}
\renewcommand{\arraystretch}{0.7}
\begin{tabular}{cc|cc|cc|cc}
\toprule
\multicolumn{4}{c}{\textbf{Control}} & 
\multicolumn{4}{c}{\textbf{Case}} \\
\cmidrule(lr){1-4} \cmidrule(lr){5-8}
\textbf{ID} & \textbf{Sample size} & 
\textbf{ID} & \textbf{Sample size} & 
\textbf{ID} & \textbf{Sample size} & 
\textbf{ID} & \textbf{Sample size} \\
\midrule
1 & 1795 & 9  & 2630 & 1 & 3267 & 10 & 1688 \\
2 & 1363 & 10 & 3491 & 2 & 2423 & 11 & 3007 \\
3 & 2487 & 11 & 1953 & 3 & 1661 & 12 & 2196 \\
4 & 1234 & 12 & 2144 & 4 & 2275 & 13 & 2825 \\
5 & 2511 & 13 & 1061 & 5 & 2561 & 14 & 2690 \\
6 & 2433 & 14 & 1387 & 6 & 2511 & 15 & 2523 \\
7 & 3523 & 15 & 2461 & 7 & 3757 & 16 & 2077 \\
8 & 1239 & 16 & 3288 & 8 & 1758 & 17 & 3228 \\
 &      &    &      & 9 & 3227 &    &      \\
\bottomrule
\end{tabular}
\caption{Sample sizes for 16 healthy control datasets and 17 MDD case datasets.}
\label{tab:samplesize}
\end{table}

\end{document}